\documentclass[11pt]{article}
\usepackage{complexity}
\usepackage[T1,T2A]{fontenc}

\usepackage{amsmath}
\usepackage{amssymb}
\usepackage{amsthm}
\usepackage{amsfonts}
\usepackage[pagebackref]{hyperref}
\hypersetup{
    colorlinks=true, 
    linkcolor=blue, 
    citecolor=blue, 
    urlcolor=blue 
}
\usepackage{microtype}

\usepackage{accents}
\usepackage[boxed,titlenumbered]{algorithm2e}
\usepackage[toc,page]{appendix}
\usepackage{array}
\usepackage[utf8]{inputenc}
\usepackage[main=english, russian]{babel}
\usepackage{braket}
\usepackage{booktabs}
\usepackage{bigdelim}
\usepackage{cancel}
\usepackage[shortlabels]{enumitem}
\usepackage{float}
\usepackage[bottom]{footmisc}
\usepackage[margin=1in]{geometry}
\usepackage{graphicx}
\usepackage[capitalize,nameinlink]{cleveref} 
\usepackage{mathtools}
\usepackage{multirow}
\usepackage{physics}
\usepackage{tabularx}
\usepackage[dvipsnames]{xcolor}
\usepackage{tikz}
\usepackage[textsize=footnotesize]{todonotes}
\usepackage{qcircuit}
\usetikzlibrary{positioning}
\usetikzlibrary{decorations.pathreplacing}
\usetikzlibrary{cd}
\usetikzlibrary{arrows.meta}
\usetikzlibrary{calc}
\usepackage{bbold}

\usepackage{tikzsymbols}

\usepackage{xspace}

\renewcommand{\backref}[1]{}

\renewcommand{\backrefalt}[4]{%
\ifcase #1 %
\or
[p.\ #2]%
\else
[pp.\ #2]%
\fi}

\makeatletter
\renewcommand{\paragraph}{%
 \@startsection{paragraph}{4}%
 {\z@}{2.25ex \@plus .5ex \@minus .3ex}{-1em}%
 {\normalfont\normalsize\bfseries}%
}
\makeatother

\makeatletter
\newcommand{\para}{%
 \@startsection{paragraph}{4}%
 {\z@}{2ex \@plus 3.3ex \@minus .2ex}{-1em}%
 {\normalfont\normalsize\bfseries}%
}
\makeatother

\usepackage{nameref}

\newtheorem{lem}{Lemma}[section]
\newtheorem*{lem*}{Lemma}
\newtheorem{thm}[lem]{Theorem}
\crefname{thm}{Theorem}{Theorems}
\crefname{lem}{Lemma}{Lemmas}
\newtheorem*{thm*}{Theorem}
\newtheorem{prop}[lem]{Proposition}

\newtheorem{cor}[lem]{Corollary}
\newtheorem{defn}[lem]{Definition}

\theoremstyle{definition}

\newtheorem{rmk}[lem]{Remark}

\makeatletter
\newtheorem*{rep@theorem}{\rep@title}
\newcommand{\newreptheorem}[2]{%
\newenvironment{rep#1}[1]{%
 \def\rep@title{#2 \ref{##1}}%
 \begin{rep@theorem}}%
 {\end{rep@theorem}}}
\makeatother

\newreptheorem{thm}{Theorem}
\newreptheorem{lem}{Lemma}

\def\beq{\begin{equation}}
\def\eeq{\end{equation}}

\def\ben{\begin{enumerate}}
\def\een{\end{enumerate}}

\def\bem{\begin{bmatrix}}
\def\eem{\end{bmatrix}}

\newcommand{\df}[1]{{\emph{#1}}{\index{#1}}}

\def\cH{{\mathcal H}}
\def\cF{{\mathcal F}}

 \numberwithin{equation}{section}

\usepackage{imakeidx}

\makeindex
\makeindex[name=NOTA, title=Notation, columns=1]

\def\nx{\index[NOTA]}

\newcommand{\angles}[1]{\left \langle #1 \right \rangle}
\newcommand{\com}[2]{\left[#1,#2\right]}

\newcommand{\proj}[1]{\varphi_{#1}}
\newcommand{\cmap}[1]{\varphi^*_{#1}}

\newcommand{\cpathnew}[3]{\overline{P}_{#1}\left(#2,\; #3\right)}
\newcommand{\cpathnewone}[1]{\overline{P}_{#1}}
\newcommand{\cG}{\mathcal{G}}

\newcommand{\repnew}[3]{r_{#1}^{(#2) \rightarrow (#3)}}
\newcommand{\id}{1}
\newcommand{\gadget}[1]{f_{#1}}
\newcommand{\total}{F}
\newcommand{\clauseh}[1]{x_{a_#1}^{(1)}x_{b_#1}^{(2)}x_{c_#1}^{(3)}}
\newcommand{\protoGadget}{\gamma}
\newcommand{\merp}[1]{\exp(i#1 \sigma_z) \sigma_x \exp(-i #1 \sigma_z)}
\newcommand{\merptwo}[1]{\exp(2 i #1 \sigma_z) \sigma_x }
\newcommand{\merptwom}[1]{\exp(-2 i #1 \sigma_z) \sigma_x }

\usepackage{tikz}
\usetikzlibrary{backgrounds}
\usepackage{keycommand}
\usepackage{ifthen}

\newkeycommand{\vertex}[color=black,xpos=0,ypos=0,label=true,label name=default,label position=below, label color = black][1]{%

    \def\true{true}
    \def\default{default}

    \fill[\commandkey{color}] (\commandkey{xpos},\commandkey{ypos}) circle (0.025);

    \ifthenelse{\equal{\commandkey{label}}{\true}}{
    
        \ifthenelse{\equal{\commandkey{label name}}{\default}}{
            \def\nameoflabel{#1}
        }{
            \def\nameoflabel{\commandkey{label name}}
        }
        
        \node[\commandkey{label position}, \commandkey{label color}] at (\commandkey{xpos},\commandkey{ypos}) {\nameoflabel{}};  
    }{}
    
    {\expandafter\xdef\csname @vertex@#1@x\endcsname{\commandkey{xpos}}}
    {\expandafter\xdef\csname @vertex@#1@y\endcsname{\commandkey{ypos}}}
}

\newkeycommand{\edge}[color=black][2]{
    \draw [\commandkey{color}] ({\csname @vertex@#1@x\endcsname}, {\csname @vertex@#1@y\endcsname}) -- ({\csname @vertex@#2@x\endcsname}, {\csname @vertex@#2@y\endcsname});
}

\newkeycommand{\threeedge}[color=black,outline=true,fill=false,fill color=pink][3]{

    \def\true{true}

    \begin{scope}[on background layer]
    \ifthenelse{\equal{\commandkey{outline}}{\true}}{
        \draw [\commandkey{color}, double distance = 0.58cm, line cap=round, rounded corners = 0.15cm, fill=none, line width = 0.25mm] ({\csname @vertex@#1@x\endcsname}, {\csname @vertex@#1@y\endcsname}) -- ({\csname @vertex@#2@x\endcsname}, {\csname @vertex@#2@y\endcsname}) -- ({\csname @vertex@#3@x\endcsname}, {\csname @vertex@#3@y\endcsname});
    }{
    }
    \end{scope}
    
    \ifthenelse{\equal{\commandkey{fill}}{\true}}{
        \draw [\commandkey{fill color}, line width = 0.58cm, line cap=round, rounded corners = 0.15cm, opacity=0.5] ({\csname @vertex@#1@x\endcsname}, {\csname @vertex@#1@y\endcsname}) -- ({\csname @vertex@#2@x\endcsname}, {\csname @vertex@#2@y\endcsname}) -- ({\csname @vertex@#3@x\endcsname}, {\csname @vertex@#3@y\endcsname});
    }{
    }
}

\newcommand{\test}[1]{\csname @vertex@#1@x\endcsname}

\def\addvalue#1#2{\expandafter\gdef\csname my@data@#1\endcsname{#2}}
\newcommand{\usevalue}[1]{\csname my@data@#1\endcsname}

\usepackage{caption}

\title{\Large 3XOR Games with Perfect Commuting Operator Strategies Have Perfect Tensor Product Strategies and are Decidable in Polynomial Time}
\author{Adam Bene Watts\thanks{Massachusetts Institute of Technology. \texttt{abenewat@mit.edu}} \and J. William Helton\thanks{University of California San Diego. \texttt{helton@math.ucsd.edu}}}

\date{}

\begin{document}

\maketitle

\begin{abstract}

    We consider 3XOR games with perfect commuting operator strategies. Given any 3XOR game, we show existence of a perfect commuting operator strategy for the game can be decided in polynomial time. Previously this problem was not known to be decidable. Our proof leads to a construction, showing a 3XOR game has a perfect commuting operator strategy iff it has a perfect tensor product strategy using a 3 qubit (8 dimensional) GHZ state. This shows that for perfect 3XOR games the advantage of a quantum strategy over a classical strategy (defined by the quantum-classical bias ratio) is bounded. This is in contrast to the general 3XOR case where the optimal quantum strategies can require high dimensional states and there is no bound on the quantum advantage.
    
    To prove these results, we first show equivalence between deciding the value of an XOR game and solving an instance of the subgroup membership problem on a class of right angled Coxeter groups. We then show, in a proof that consumes most of this paper, that the instances of this problem corresponding to 3XOR games can be solved in polynomial time.
    \\
    
\end{abstract}

\tableofcontents

\section{Introduction} \label{sec:Intro}
One fantastic implication of quantum mechanics is that measurements made on quantum mechanical systems can produce correlated outcomes irreproducible by any classical system. This observation is at the heart of Bell's celebrated 1964 inequality~\cite{bell1964einstein} and has since found applications in cryptography~\cite{QKD1:barrett2005no,QKD2:ekert1991quantum,QKD3:vazirani2019fully,RGcolbeck2009quantum}, delegated computing~\cite{DQC:reichardt2013classical} and short depth circuits~\cite{SDCbravyi2018quantum,SDCwatts2019exponential,SDCgrier2019interactive}, among others. Recent results have shown the sets of correlations producible by measuring quantum states are incredibly difficult to characterize~\cite{natarajan2019neexp,dykema2019non,ji2020mip,coudron2019complexity,slofstra2020tsirelson,coladangelo2018unconditional}.

In this work, we present a result in the opposite direction. We consider a natural question concerning existence of quantum correlations which has been open for decades and is comparable to the one shown to be undecidable in \cite{slofstra2020tsirelson}. We show it can be answered in polynomial time. Furthermore we show that when these correlations can be produced, they can be produced by simple measurements of a finite dimensional quantum state. We begin by reviewing some necessary background. 

\paragraph{Nonlocal Games.} Nonlocal games describe experiments which test the correlations that can be produced by measurements of quantum systems. A nonlocal game involves a \df{referee} (also called the \df{verifier}) and $k \geq 1$ \df{players} (also called \df{provers}). In a round of the game, the verifer selects a \df{question vector} $q = (q_1, q_2, ... , q_k)$ randomly from a set $S$ of possible question vectors, then sends player $i$ question $q_i$. Each player responds with an answer $a_i$. The players cannot communicate with each other when choosing their answers. After receiving an answer from each player, the verfier computes a \df{score} $V(a_1, a_2, ..., a_k | q_1, q_2, ... , q_k)$ which depends on the questions selected and answers recieved. The players know the set of possible questions $S$ and the scoring function $V$. Their goal is to chose a \df{strategy} for responding to each possible question which maximizes their score in expectation. The difficulty for the players lies in the fact that in a given round each player only has partial information about the questions sent to other players. 

For a given game $G$, the supremum of the expected scores achievable by players is called the \df{value} of the game. The value depends on the resources available to the players. If players are restricted to classical strategies, the value is called the \df{classical value} and denoted $\omega(G)$. If players can make measurements on a shared quantum state (but still can't communicate) the value can be larger and is called the \df{entangled value}. More specifically, if the players shared state lives in a Hilbert space $\cH = \cH_1 \otimes \cH_2 \otimes ... \otimes \cH_k$ and the $i$-th player makes a measurement on the $i$-th Hilbert space, the supremum of the scores the players can obtain is called the \df{tensor product value}, denoted $\omega^*_{tp}$. If the players share an arbitrary state and the only restriction placed on their measurements is that the measurement operators commute (enforcing no-communication), the supremum of the achievable scores is called the \df{commuting operator value}, denoted $\omega^*_{co}$. When the state shared by the players is finite dimensional these definitions coincide. In the infinite dimensional case  $\omega^*_{tp} \leq \omega^*_{co}$, and there exist games for which the inequality is strict~\cite{ji2020mip}.

\paragraph{Bounds On the Value.} The commuting operator and tensor product values of a game are in general uncomputable \cite{ji2020mip, slofstra2020tsirelson}. Intuitively, this is because the nonlocal games formalism places no restriction on the dimension of the state shared by the players, and so a brute force search over strategies will never terminate. However, such a search can provide a lower bound on the value of a game. Given a game $G$, let $\omega^*_d(G)$ denote the maximum score achievable by players using states of dimension at most $d$. This value lower bounds the tensor product (hence, commuting operator) value, and converges to the tensor product value in the limit as $d \rightarrow \infty$~\cite{scholz2008tsirelson}, so
$
    \sup_{d < \infty}\left\{\omega^*_d\right\} = \omega^*_{tp}.
$
 Given a fixed $d$, $\omega^*_{d}$ can be computed by exhaustive search. Computing $\omega^*_{d}$ for an increasing sequence of $d$'s produces a sequence of lower bounds that converge to $\omega^*_{tp}$ from below. 
 
 It is also possible to bound the commuting operator value of a nonlocal game from above, via a convergent hierarchy of semidefinite programs known as the \df{NPA hierarchy}~\cite{navascues2008convergent,doherty2008quantum}. (Both these papers focus on upper bounds in the two player case, but $k$ player generalizations are straightforward.)
When run to a finite level, this hierarchy gives an upper bound on the commuting operator value of a game. However there is no guarantee that this bound can be achieved by any commuting operator strategy, hence no guarantee that the upper bound matches the true commuting operator value. In general all that can be said is that this hierarchy is complete, meaning that the bound computed necessarily converges to the commuting operator value of the game. Because of the previously mentioned undecidability results, no general bounds can be put on this rate of convergence. 

\paragraph{XOR Games.} XOR games are one family of games for which more concrete results are known. These are nonlocal games where each question $q_j$ is drawn from an alphabet of size $n$, player's responses are single bits $a_i \in \{0,1\}$ and the scoring function checks if the overall parity of the responses matches a desired parity $s_j$ associated with the question, that is 
\begin{align}
    V(a_1, a_2, ... , a_k | q_1, q_2, ... , q_k) = \begin{cases} 1 \text{ if } \sum_i a_i = s_j \pmod{2} \\
    0  \text{ otherwise.}
    \end{cases}
\end{align}
We refer to an XOR game with $k$ players as a $k$XOR Game. It is helpful to think of an $k$XOR game as testing satisfiabiliy of a set of clauses:
\begin{align}
    \hat{X}^{(1)}_{q_{11}} + \hat{X}^{(2)}_{q_{12}} + ... \hat{X}^{(k)}_{q_{1k}} = s_1, \, \hat{X}^{(1)}_{q_{21}} + \hat{X}^{(2)}_{q_{22}} + ... \hat{X}^{(k)}_{q_{2k}} = s_2, \,
    ... \nonumber , \,
    \hat{X}^{(1)}_{q_{m1}} + \hat{X}^{(2)}_{q_{m2}} + ... \hat{X}^{(k)}_{q_{mk}} = s_m,
\end{align}
where each clause $\hat{X}^{(1)}_{q_{j1}} + ... + \hat{X}^{(k)}_{q_{jk}} = s_j$ corresponds to a question vector $(q_{j1},q_{j2}, ..., q_{jk})$ with associated parity bit $s_j$. If question vectors are chosen uniformly at random, the classical value of the game corresponds to the maximum fraction of simultaneously satisfiable clauses (see \cref{subsec:Strategies} for a proof of this fact). The tensor product and commuting operator values have no such interpretation, and may be larger. 

Famously, Bell's inequality can be expressed as a 2XOR game called the CHSH game~\cite{clauser1969proposed}, with clauses
    $\hat{X}^{(1)}_1 + \hat{X}^{(2)}_1 = 1$,
    $\hat{X}^{(1)}_0 + \hat{X}^{(2)}_1 = 0$,
    $\hat{X}^{(1)}_1 + \hat{X}^{(2)}_0 = 0$, and 
    $\hat{X}^{(1)}_0 + \hat{X}^{(2)}_0 = 0$. 
At most 3 of these 4 clauses can be simultaneously satisfied, so the classical value of this game is $0.75$. However, there exists a strategy involving measurements on the two qubit Bell state $\ket{\Phi^{+}} = \frac{1}{\sqrt{2}}\left(\ket{11} + \ket{00}\right)$ which achieves an expected score of $\cos^2(\pi/8) \approx 0.85$. 2XOR games are well understood in general; in 1987 Tsirelson showed the optimal value for any 2XOR game can be achieved by a finite dimensional strategy which can be found in polynomial time~\cite{tsirel1987quantum}. This result shows the 2 qubit strategy is optimal for the CHSH game, so $\omega^*_{co}(\text{CHSH}) =  \omega^*_{tp}(\text{CHSH}) = \cos^2(\pi/8)$. More generally, Tsirelson's result showed $\omega^*_{co} = \omega^*_{tp}$ for any 2XOR game.

For $k$XOR games with $k > 2$ the situation is much more opaque. There exist polynomial time algorithms that can compute $\omega^*_{co}$ and $\omega^*_{tp}$ in special cases~\cite{watts2018algorithms,werner2001all}. On the other hand it is $\NP$-hard to compute the classical value of a 3XOR game~\cite{haastad2001some}, and there is no known upper bound on the runtime required to compute the commuting operator or tensor product value of a $k$XOR game when $k \geq 3$. Furthermore, the commuting operator and tensor product values of a $k$XOR game are not known to coincide. One natural and efficiently solvable problem involving $k$XOR games is identifying games with perfect classical value $\omega = 1$. This is equivalent to asking if the corresponding set of clauses is exactly solvable, so can be answered in polynomial time using Gaussian elimination.

Interestingly, there exist XOR games with $\omega^*_{tp} = 1$ and $\omega < 1$; the sets of clauses associated with these games appear perfectly solvable when the game is played by players sharing an entangled state, despite the clauses having no actual solution. The most famous of these XOR \df{pseudotelepathy games}~\cite{brassard2005quantum} is the GHZ game, a 3XOR game with 4 clauses and classical value $\omega = 3/4$. There is a perfect value tensor product strategy for this game involving measurements of the GHZ state $\frac{1}{\sqrt{2}}\left(\ket{000} + \ket{111}\right)$ so $\omega^*_{tp}(\text{GHZ}) = \omega^*_{co}(\text{GHZ}) = 1$ \cite{greenberger1990bell,mermin1990extreme}.

The relative difficulty of computing the classical value of $k$XOR games compared to the ease of identifying perfect value $k$XOR games motivates an analogous question concerning the entangled values. Does there exist a non-commutative analogue of Gaussian elimination that can easily identify $k$XOR games with $\omega^*_{co}$ or $\omega^*_{tp} = 1$? How hard is it to identify XOR pseudotelepathy games? 

\paragraph{Bias.} XOR games can also be characterized by their \df{bias} $\beta(G)$, defined by $\beta(G) = 2\omega(G) - 1$.\footnote{Some definitions vary by a factor of $2$, defining $\beta(G) = \omega(G) - 1/2$}  The \emph{entangled biases} $\beta^*_{co}$ and $\beta^*_{tp}$ are defined analogously. A completely random strategy for answering an XOR game will achieve a score of $1/2$, hence $\omega(G) \geq 1/2$ and $\beta(G) \in [0,1]$, with identical bounds holding on the other biases. When comparing classical and entangled biases, the quantity usually considered is the ratio $\beta_{tp}^*(G)/\beta(G)$ (or $\beta_{co}^*(G)/\beta(G)$), called the quantum-classical gap.

For 2XOR games this gap can be related to the Grothendieck inequality, with 
\begin{align}
    \beta_{co}^*(G)/\beta(G) = \beta_{tp}^*(G)/\beta(G) \leq K^{\mathbb{R}}_G,
\end{align}
where $K^{\mathbb{R}}_G$ is the real Grothendieck constant\footnote{Because $\omega_{co}^* = \omega_{tp}^*$ for 2XOR games, we also have $\beta_{co}^* = \beta_{tp}^*$}. For 3XOR games no such bound holds~\cite{perez2008unbounded,briet2013explicit}, and there exist families of games $\{G_n\}_{n \in \mathbb{N}}$ with
\begin{equation}
\label{eq:Qadv}
    \lim_{n \rightarrow \infty} \beta_{tp}^*(G_n)/\beta(G_n) = \infty. 
\end{equation}
All these families have the property that $\lim_{n \rightarrow \infty} \beta_{tp}^*(G_n) = 0$; 
it is open whether an unbounded quantum-classical gap can exist for $k$XOR games with $\beta_{co}^*$ bounded away from zero. One special case where a bound on the quantum-classical gap is known is 3XOR games with the players restricted to a GHZ state \cite{perez2008unbounded} (later generatlized to Schmidt states in \cite{briet2009multiplayer}). In this case the quantum-classical gap is bounded above by $4K^{\mathbb{R}}_G$~\cite{briet2009multiplayer}.

\paragraph{Our Main Results.} This paper considers perfect commuting operator strategies for XOR games. We first show a link between XOR games and algebraic combinatorics: proving a $k$XOR game has value $\omega^*_{co} = 1$ iff an instance of the subgroup membership problem on a right angled Coxeter group corresponding to the $k$XOR game has no for an answer.
For $k$XOR games with $k \geq 3$, the corresponding class of Coxeter groups has undecidable subgroup membership problem. A priori, it is not clear whether or not the instances determining if a game has value $\omega_{co}^*=1$ are decidable. 
In this paper we resolve the 3XOR case by proving an algebraic result (whose proof consumes most of this paper) showing the instances of the subgroup membership problem determining the value of $3$XOR games are equivalent to instances on a simpler group $G/K$ obtained from $G$ by modding out a particular normal subgroup $K$. This equivalence lets us construct a polynomial time algorithm that determines if 3XOR games do or do not have value $\omega^*_{co} = 1$. Previously this problem was not known to be decidable.  For $k \geq 4$ it remains open whether or not there is any  algorithm which can decide in finite time if a game has a perfect commuting operator strategy.

Combining this result with arguments from~\cite{watts2018algorithms} shows 3XOR games with $\omega^*_{co} = 1$ also have perfect value tensor product strategies, with the players sharing a three qubit GHZ state. Combining that observation with the known bounds on the quantum-classical gap for strategies using a GHZ state~\cite{perez2008unbounded,briet2009multiplayer} shows that 3XOR games with $\omega^*_{co} = 1$ have classical value bounded a constant distance above $1/2$.
In other words, when $\omega^*_{co} = 1$, how well quantum bias outperforms classical bias is bounded.
This is in contrast with the behavior, see
\Cref{eq:Qadv}, of not perfect games.

\Cref{sec:Overview} gives basic definitions, precise statements of the main theorems, and proofs or proof sketches where appropriate. \Cref{sec:Kmodding} gives proofs of the more involved algebraic results. The appendices fill in proof details and give perspectives, mostly about the subgroup $K$.

\paragraph{Comparison with Other Work.} Our result shares high-level structure with the work of Cleve and coauthors~\cite{cleve2014characterization, cleve2017perfect} and followup work by Slofstra~\cite{slofstra2020tsirelson} concerning linear systems games, though our work comes to a very different conclusion than theirs. In both that work and ours, perfect value commuting operator strategies are shown to exist for a family of nonlocal games iff an algebraic property is satisfied on a related group. In~\cite{slofstra2020tsirelson}, Slofstra showed that the algebraic property associated with linear systems games was undecidable, implying existence of a linear systems game whose only perfect value strategies were incredibly complicated (infinite dimensional). Here we show the algebraic property associated with perfect value 3XOR games can be checked in polynomial time, and give a finite dimensional strategy,
called a MERP strategy, that achieves value 1 whenever a perfect value commuting operator strategy exists. 

The MERP strategy  is a variant of the GHZ strategy that has been considered before. In \cite{werner2001all} this strategy was shown to be optimal for $k$XOR games with two questions per player. In \cite{watts2018algorithms} this strategy 
was shown to be optimal for a restricted class of XOR games (symmetric $k$XOR games)\footnote{Symmetric XOR games are XOR games whose scoring function is invariant under permutations of the players. 
As an example, this  would force $V(a_1 \oplus a_2 | q_1, q_2) = V(a_2 \oplus a_1 | q_2, q_1)$ for a two player symmetric XOR game.
} with perfect value. 
In \cite{hoyer2005quantum} a quantum circuit closely related to this strategy was used as a 
subroutine in short depth circuits.

\section{A Detailed Overview} 
\label{sec:Overview}

We begin this section by introducing notation necessary to state the main theorems of this work. Much of it is specific to this paper, so we suggest a reader familiar with the field still read \Cref{sec:Notation} fairly closely. \Cref{sec:Results} contains all the major theorem statements of this paper. 

\subsection{Background and Notation}\label{sec:Notation}

\subsubsection{Games}

    As mentioned in \Cref{sec:Intro}, we think of XOR games as testing satisfiability of an associated system of equations. Our starting point for defining any $k$XOR game is a system of equations of the form
    \begin{align}
        \hat{X}_{n_{11}}^{(1)} + \hat{X}_{n_{12}}^{(2)} + ... + \hat{X}_{n_{1k}}^{(k)} = s_1, \, 
        \hat{X}_{n_{21}}^{(1)} + \hat{X}_{n_{22}}^{(2)} + ... + \hat{X}_{n_{2k}}^{(k)} = s_2, \,
        ... , \,
        \hat{X}_{n_{m1}}^{(1)} + \hat{X}_{n_{m2}}^{(2)} + ... + \hat{X}_{n_{mk}}^{(k)} = s_m \nonumber
    \end{align}
    where $n_{i\alpha} \in [N]$, $s_i \in \{0,1\}$, $\hat{X}_n^{(\alpha)}$ are formal variables taking values in $\{0,1\}$ and the equations are all taken mod 2. $N$ is called the \df{alphabet size} of the game, and $m$ the number of \df{clauses}. The $k$XOR game associated to this system of equations has $m$ question vectors $\{(n_{11}, n_{12}, ... , n_{1k}), ..., (n_{m1}, n_{m2}, ... , n_{mk})\}$. In a round of the game the verifier selects a $i \in [m]$ uniformly at random, then sends question vector $(n_{i1}, n_{i2}, ... , n_{ik})$ to the players, i.e. player $j$ receives question $n_{ij}$. The players respond with single bit answers and win (get a score of 1 on) the round if the sum of their responses equals $s_i$ mod 2. They get a score of 0 otherwise. Any $k$XOR game where clauses are chosen uniformly at random can be described by specifying the associated system of equations.\footnote{Because we are concerned with the case of perfect value XOR games, fixing the distribution clauses are drawn from to be uniform doesn't change the scope of our results.}
    
    For the case of 3XOR games, we will simplify notation slightly by omitting a subindex and instead writing our system of equations as \begin{align}
        \hat{X}_{a_1}^{(1)} + \hat{X}_{b_1}^{(2)} + \hat{X}_{c_1}^{(3)} = s_1, \,
        \hat{X}_{a_2}^{(1)} + \hat{X}_{b_2}^{(2)} + \hat{X}_{c_2}^{(3)} = s_2,\,
        ... \nonumber, \,
        \hat{X}_{a_m}^{(1)} + \hat{X}_{b_m}^{(2)} + \hat{X}_{c_m}^{(3)} = s_m
    \end{align}
    where $a_i, b_i, c_i \in [N]$ for all $i \in [m]$. The question vector sent to the players is then $(a_j, b_j, c_j)$, with the players winning the round if their responses sum to $s_j$ mod 2.

\subsubsection{Strategies}  \label{subsec:Strategies}

For ease of notation, we will describe strategies in the special case of 3XOR games. We note that all the definitions given here generalize naturally to the $k$-player case. We begin this section with a brief discussion of classical strategies, then move on to consider entangled strategies. The discussion of classical strategies is included mostly for perspective, and can be skipped. The definitions related to entangled strategies are essential.

The most general classical strategy can be described by specifying a response for each player based on the question received and some shared randomness $\lambda$. If we are only concerned with strategies that maximize the players' score, a convexity argument shows that we can ignore the shared randomness (fix $\lambda$ to the value that maximizes the players' score in expectation), so optimal classical strategies can be described by fixing responses for each player to each possible question. To better align with the quantum case, we describe these strategies multiplicatively rather then additively. Define $X_i^{(\alpha)}$ to equal $1$ if player $\alpha$ responds to question $i$ with a $0$, and $X_i^{(\alpha)} = -1$ if the player responds with a $1$. Players win on the $j$-th question vector iff $X_{a_j}^{(1)}X_{b_j}^{(2)}X_{c_j}^{(3)}(-1)^{s_j} = 1$ so the expected score of the players conditioned on receiving the $j$-th question vector can be written 
\begin{align}
    \frac{1}{2} +   \frac{1}{2}X_{a_j}^{(1)}X_{b_j}^{(2)}X_{c_j}^{(3)}(-1)^{s_j}.
\end{align}
and the expected score this strategy achieves on a XOR game is given by
\begin{align}
    \frac{1}{2} + \frac{1}{2m} \sum_j  X_{a_j}^{(1)}X_{b_j}^{(2)}X_{c_j}^{(3)}(-1)^{s_j}.
\end{align}

We refer to strategies where players share and measure a quantum state before deciding their response as \df{entangled strategies}.\footnote{The name quantum strategies, while more natural, can cause confusion with strategies where questions and responses are themselves quantum states. Entanglement is not necessary for these strategies, but the players' achieve a value exceeding their classical value only if the state they share is entangled.} In the most general entangled strategy, players share an state $\ket{\psi}$ and randomness $\lambda$. Then they receive a question, make a measurement on the quantum state based on the question and shared randomness, and then send a response to the verifier based on the measurement outcome. Mathematically, any strategy can be described by fixing the state $\ket{\psi}$ and POVMs (Positive Operator-Valued Measures) for each possible question sent to the players. A Naimark
type dilation theorem tells us that any such strategy can be transformed to one where players' measurements are all described by PVMs (Projective Valued Measures) without changing the score that strategy achieves on a game (the finite dimensional case is standard, see \cite{fritz2012tsirelson} Section 3, for the infinite dimensional argument). Thus, when considering whether or not a game has an optimal strategy we are free to consider only strategies which can be described by a shared state $\ket{\psi}$ and PVMs (Projective Valued Measures) for each possible player and question. 

 In this paper we describe entangled strategies using the PVM formalism. More specifically, we study self adjoint operators associated with these PVMs. We define these self-adjoint operators as follows:

\begin{enumerate}
    \item First, specify the shared state $\ket{\psi}$ which lives in some Hilbert space $\cH$. 
    \item For each player $\alpha \in [k]$ and question $i \in [N]$, let $P^{(\alpha)}_i$ be the projector onto the subspace of $\cH$ associated with a 1 response by player $\alpha$ to question $i$. Similarly, let $Q^{(\alpha)}_i = \id - P^{(\alpha)}_i$ be the projector onto the subspace associated with a 0 response. Here $\id$ represents the identity operator.
    \item For every $\alpha$ and $i$, define the \df{strategy observable}
$
        X_i^{(\alpha)} = Q^{(\alpha)}_i - P^{(\alpha)}_i.
$
\end{enumerate}
The operators $X_i^{(\alpha)}$ satisfy some useful properties. Firstly, they are self-adjoint by construction with eigenvalues~$\pm1$. From this, or from direct calculation, it follows that 
\begin{align}
    \left(X_i^{(\alpha)}\right)^2 &= \left(Q_i^{(\alpha)}\right)^2 + \left(P_i^{(\alpha)}\right)^2 + 2Q_i^{(\alpha)}P_i^{(\alpha)} = Q_i^{(\alpha)} + P_i^{(\alpha)} = \id,
\end{align}
where we have used the fact that $Q_i^{(\alpha)}$ and $P_i^{(\alpha)}$ are orthogonal projectors on the last line. 

Secondly, the restriction that players be non-communicating means that a players chance of responding 1 (resp. 0) should be independent of another player's response. Hence
\begin{align}
    P_i^{(\alpha)} P_j^{(\beta)} = P_j^{(\beta)} P_i^{(\alpha)}
\end{align}
for any $i, j, \alpha \neq \beta$. Defining the group commutator of two observables $[y,z] := yzy^{-1}z^{-1}$ we see 
\begin{align}
    \com{X_i^{(\alpha)}}{ X_j^{(\beta)}} 
    = 1
\end{align}
whenever $\alpha \neq \beta$. 

Finally, we consider a product of operators corresponding to a question vector in the XOR game. A state is in the $1$ eigenspace of $X^{(1)}_{a_j}X^{(2)}_{b_j}X^{(3)}_{c_j}$ iff the sum mod 2 of the players responses to the verifier upon measuring this state is $0$. Similarly a state is in the $-1$ eigenspace iff the sum of the players responses upon measuring this state is $1$. Then, players win on question vector $j$ with probability
\begin{align}
    \frac{1}{2} + \frac{1}{2}\bra{\psi} X^{(1)}_{a_j}X^{(2)}_{b_j}X^{(3)}_{c_j}(-1)^{s_j} \ket{\psi}
\end{align}
and their overall score on the game is given by 
\begin{align}
     \frac{1}{2} + \frac{1}{2m} \sum_j \left(\bra{\psi} X^{(1)}_{a_j}X^{(2)}_{b_j}X^{(3)}_{c_j}(-1)^{s_j} \ket{\psi}\right). \label{eq:entValue}
\end{align}
An important consequence of \cref{eq:entValue} is that the players win the game with probability $1$ iff 
\begin{align} \label{eq:satcond}
    X^{(1)}_{a_j}X^{(2)}_{b_j}X^{(3)}_{c_j}(-1)^{s_j} \ket{\psi} = \ket{\psi}
\end{align}
for all $j \in [m]$. This is because each $X_i^{(\alpha)}$ has norm $\leq 1$. 

\subsubsection{Groups} \label{subsec:groups}

Now we introduce groups whose structure mimics the structure of the strategy observables introduced in \Cref{subsec:Strategies}. We describe these groups using the language of group presentations. The language in this section is, at times, technical and we alert the reader that explicit examples of this notation are given in \Cref{ssec:examples}. 

Given integers $k$ and $N$ we define the \df{game group} $G$ to be the group with generators $\sigma$ and $x_i^{(\alpha)}$ for all $i \in [n], \alpha \in [k]$, and relations:
\begin{enumerate}
    \item $\left(x_i^{(\alpha)}\right)^2 = 1$ for all $i,j \in [n], \alpha \in [k]$
    \item $\com{x_i^{(\alpha)}}{x_j^{(\beta)}} = 1$ for all $i,j \in [n], \alpha \neq \beta \in [k]$
    \item $\sigma^2 = \com{\sigma}{x_i^{(\alpha)}} = 1$ for all $i,j \in [n], \alpha \neq \beta \in [k].$
\end{enumerate}
 Here the $x_i^{(\alpha)}$ are group elements satisfying the same relations as the strategy observables defined in \Cref{subsec:Strategies}.  The element $\sigma$ is a formal variable playing the role of $-1$. Note $\sigma \neq 1$ in the group. While it is not needed for the paper, we remark here that G is a right angled Coxeter group.

Given an $k$-player XOR game testing the system of $m$ equations
\begin{align}
    \hat{X}_{n_{11}}^{(1)} + \hat{X}_{n_{12}}^{(2)} + ... + \hat{X}_{n_{1k}}^{(k)} = s_1, \, 
    \hat{X}_{n_{21}}^{(1)} + \hat{X}_{n_{22}}^{(2)} + ... + \hat{X}_{n_{2k}}^{(k)} = s_2, \,
    ... , \,
    \hat{X}_{n_{m1}}^{(1)} + \hat{X}_{n_{m2}}^{(2)} + ... + \hat{X}_{n_{mk}}^{(k)} = s_m \nonumber
\end{align}
we define the \df{clauses} $h_1, h_2, ... , h_m$ of the game by 
\begin{align}
\label{eq:sumToProd}
    h_i = \prod_{\alpha=1}^k X_{n_{i\alpha}}^{(\alpha)}\sigma^{s_i} \in G,
\end{align}
where $\sigma^0 = \id$. 
We denote the set of all clauses by $S$ and 
define the \df{clause group} $H \leq G$ to be the subgroup generated by the clauses, so $H= \angles{S} = \angles{\left\{ h_i : i \in [m] \right\}}.$ 

\def\nsg{{T}}
We note that this construction lets us associate any $k$-player XOR game with a subgroup $H$ of the group $G$. It is also worth noting that the clause group $H$ is, in general, not a normal subgroup of the game group $G$.  Here we recall that a subgroup $\nsg $ of $G$ is called a normal subgroup (denoted $T \triangleleft G$) if $g\nsg g^{-1} = \nsg$ for all $g \in G$, i.e. for all $g \in G$, $t \in \nsg$ we also have
\begin{align}
    g t g^{-1} \in \nsg. 
\end{align} 
Important subgroups of groups $G$ and $H$ are those consisting of even length words corresponding to each player. Define the \df{even subgroups} $G^{E}, H^{E}$ by
\begin{align}
    G^{E} := \angles{\left\{ x_i^{(\alpha)}x_j^{(\alpha)} : i,j \in [N], \alpha \in [k] \right\} \cup \{\sigma\} }  \text{ and } 
    H^{E} := \angles{\left\{ h_i h_j : i, j \in [m] \right\}} 
\end{align}
Note that $H^{E} < G^{E}$.

Given a set of elements $R \subseteq G^E$ the normal closure of $R$ in $G^E$, denoted in this paper by $\angles{R}^{G^E}$ is defined to be the smallest normal subgroup of $G^E$ containing the elements of $R$. Equivalently,  $\angles{R}^{G^E}$ is the subgroup of $G^E$ generated by the set of elements
\begin{align}
    \{ g r g^{-1} : g \in G^E,\; r \in  R\}.
\end{align}
Define the \df{even commutator subgroup} $K$ of $G^E$ by:
\begin{align}
\label{def:K}
    K = \angles{\left\{\com{x_i^\alpha x_j^\alpha}{x_k^\alpha x_l^\alpha} : i,j,k,l \in [n], \alpha \in [k] \right\}}^{G^E}.
\end{align}

In this paper we will frequently study the group $G^E/K$ obtained by \df{modding out} the group $G^E$ by the normal subgroup $K$. The first isomorphism theorem tells us that this is a well defined group whose elements can be identified with the cosets 
\begin{align}
    \{wK = Kw  :  w \in G^E\}
\end{align} of $K$ in the group $G^E$. In this paper we will denote the elements of $G^E/K$ as $[w]_K$ where $w \in G^E$ and
\begin{align}
    [w_1]_K = [w_2]_K
\end{align}
iff $w_1 k = w_2$ for some $k \in K$. The normal subgroup property ensures that elements in $G^E/K$ multiply as in the group $G^E$, with 
\begin{align}
    [w_1]_K[w_2]_K = [w_1 w_2]_K.\footnotemark
\end{align}\footnotetext{To see this observe that, for any $w_1, w_2 \in G$, $k_1, k_2 \in K$, there exists a $k_1' \in K$ with $w_1 k_1 w_2 k_2 = g_1 g_2 k_1' k_2$ by the definition of a normal subgroup.}

We can also understand subgroups of $G^E/K$ using the (second) isomorphism theorem. This theorem tells us that, given any subgroup $T$ of $G^E$, denoted $T < G^E$:

\begin{enumerate}
    \item $T K$ is a subgroup of $G^E$ 
    \item $T \cap K$ is a normal subgroup of $T$ 
    \item $(T K)/ K$ is isomorphic to $T/(T \cap K)$. 
\end{enumerate} 
Particularly important to this paper will be the group $(H^E K)/ K$ which we view as a subgroup of $G^E/K$. We have for any element $[w]_K \in G^E/K$ that 
$[w]_K \in (H^E K)/ K$ iff 
\begin{align}
w k_1 = h k_2 \Leftrightarrow h = w k_3 
\end{align}
for some $k_1, k_2, k_3 \in K$ and $h \in H$ (note in this equivalence we have again used the normal property of the subgroup $K$). This condition is also equivalent to the condition  
\begin{align}
    [h_1^E]_K [h_2^E]_K ... [h_l^E]_K = [w]_K
\end{align}
where $h_1^E, h_2^E, ... h_l^E$ are generators of $H^E$. This shows that $(H^E K)/ K$ is equal to the subgroup of $G^E/K$ generated by the elements 
\begin{align}
    \left\{ [h_i h_j]_K : i, j \in [m] \right\}
\end{align}
(that is the generators of $H^E$ taken mod $K$ generate the subgroup $(H^E K)/K$ of $G^E/K$). For this reason we use the notation $[H^E]_K$ to denote the group $(H^E K)/K$. Of particular importance to the rest of this paper will be the condition 
\begin{align}
    [\sigma]_K \in [H^E]_K
\end{align}
which we also sometimes state as $\sigma \in H^E \pmod{K}$.

\subsection{Precise Statements of Main Results} \label{sec:Results}

In this section we give theorem statements covering the main results of this paper, along with some relevant theorems from previous work.

\subsubsection{An algebraic characterization of perfect \texorpdfstring{$k$-player}{k-player} XOR Games}

Our first result shows the problem of determining if $\omega_{co}^* = 1$ is equivalent to an instance of the subgroup membership problem on the game group $G$. 

We should mention  that some  ingredients of this proof have appeared before in other contexts \cite{navascues2008convergent,watts2018algorithms}.
The key innovation of this theorem is the algebraic formulation of the issue. 

\begin{thm} \label{thm:IFF}
A $k$XOR game has commuting operator value $\omega_{co}^* = 1$ iff 
$
    \sigma \notin H,
$
where $\sigma, H$ are defined relative to the $k$XOR game as described in \Cref{subsec:groups}.
\end{thm}

\begin{proof}
For notational convenience, we prove the result here in the special case of $k = 3$ players. The proof generalizes easily to other values of $k$.

We first show that $\sigma \in H \Rightarrow w^* < 1$. Assume for contradiction that $\sigma \in H$ and $w^* = 1$. Then, since $\sigma \in H$, there exists a sequence of clauses whose product
\begin{align}
    h_{t_1}h_{t_2}...h_{t_l} = \sigma, \label{eq:clause_prod_equals_sigma}
\end{align}
where each clause 
\begin{align}
    h_{t_i} = x_{a_{t_i}}^{(1)}x_{b_{t_i}}^{(2)}x_{c_{t_i}}^{(3)}\sigma^{s_{t_i}} \in S \;\;\;
    \text{ with } \;\;\; a_{t_i}, b_{t_i}, c_{t_i} \in [n],\; s_{t_i} \in \{0,1\}
\end{align} 
 is a generator of the clause group $H$. At the same time, by definition of a perfect commuting operator strategy (see \Cref{subsec:Strategies}) there exists a Hilbert space $\cH$ and a state $\ket{\psi} \in \cH$ with the property that, for every clause $x_{a_{t_i}}^{(1)}x_{b_{t_i}}^{(2)}x_{c_{t_i}}^{(3)}\sigma^{s_{t_i}} \in H$, there exist strategy observables $X_{a_{t_i}}^{(1)}, X_{b_{t_i}}^{(2)},$ and $X_{c_{t_i}}^{(3)}$ satisfying 
\begin{align}
    X_{a_{t_i}}^{(1)}X_{b_{t_i}}^{(2)}X_{c_{t_i}}^{(3)}(-1)^{s_{t_i}} \ket{\psi} = \ket{\psi} \label{eq:clause_fix_psi}
\end{align}
for all $t_i \in \{t_1, t_2, ... , t_l\}$. We can relate the group elements $x_{j}^{(\alpha)}$ and $\sigma$ to the observables $X_{j}^{(\alpha)}$ and $-1$ via a respresentation. 
By construction the strategy observables $X_j^{(\alpha)}$ satisfy the same relations as the elements of $x_{j}^{(\alpha)} \in G$ and the element $\sigma$ satisfies the same relations as $-1$ (viewed as an element of the the algebra of bounded linear operators acting on the Hilbert space $\cH$, denoted $\mathcal{B}(\cH))$. Then we can define a representation $\pi: G \rightarrow \mathcal{B}(\cH)$ with $\pi(x_j^{(\alpha)}) = X_j^{(\alpha)}$ and $\pi(\sigma) = -1$. Then we have 
\begin{align}
    \pi(h_{t_1}h_{t_2}...h_{t_l}) = \pi(\sigma) = -1
\end{align}
by \Cref{eq:clause_prod_equals_sigma} and also 
\begin{align}
    \pi(h_{t_1}h_{t_2}...h_{t_l})\ket{\psi} = \ket{\psi}
\end{align}
by repeated application of \Cref{eq:clause_fix_psi}. We conclude 
\begin{align}
    - \ket{\psi} = \ket{\psi},
\end{align}
a contraction.

It remains to show $\sigma \notin H \Rightarrow w^* = 1$. A proof of this fact that relies on completeness of the nsSoS hierarchy is given in \cite{watts2018algorithms} (Theorem 6.1, in which a sequence of clauses $h_{t_1}h_{t_2}...h_{t_l}$ satisfying $h_{t_1}h_{t_2}...h_{t_l} = \sigma$ is referred to as a \df{refutation}). Here we give a standalone proof, which can be viewed as a special case of the GNS construction. We assume $\sigma \notin H$, and construct the strategy observables and state $\ket{\psi}$ explicitly. 

First we define a Hilbert space $\cH$ with orthogonal basis vectors corresponding to the left cosets of $H$ in $G$. That is, $\cH$ is spanned by basis vectors $\{ \ket{H}, \ket{g_1 H}, ... \}$ with inner product 
\begin{align}
    \braket{g_1H}{g_2H} = \begin{cases}
    1 &\text{ if } g_1^{-1}g_2 \in H \\
    0 &\text{otherwise.}
    \end{cases}
\end{align}
Next we define the representation $\pi: G \rightarrow GL(\cH)$ to be the representation given by the left action of $G$ on $H$, so
\begin{align}
    \pi(g_1)\ket{g_2 H} = \ket{g_1 g_2 H}. 
\end{align}
Finally, define 
 \begin{align}
    \ket{\psi} = \frac{1}{\sqrt{2}}\left(\ket{H} - \ket{\sigma H}\right),
\end{align}
and note that $\sigma \notin H$ by assumption implies $\ket{\psi} \neq 0$. We claim that strategy observables $\pi(x_i^{(\alpha)})$ and state $\ket{\psi}$  achieve value $\omega^* = 1$ for the game. To see this, first note that 
\begin{align}
    \pi(\sigma)\ket{\psi} &= \pi(\sigma)\left(\ket{H} - \ket{\sigma H} \right) = \ket{\sigma H} - \ket{H} = -\ket{\psi}
\end{align}
and for word $w \in H$ we have 
\begin{align}
    \pi(w) \ket{\psi} &= \pi(w) \left(\ket{\sigma H} - \ket{H}\right) = \ket{ \sigma w H} - \ket{w H} = \ket{ \sigma H} - \ket{H} = \ket{\psi} \label{eq:sigma_H}
\end{align}
since $\sigma$ commutes with all elements of $G$.
Then, for any $j \in [m]$ we have 
\begin{align}
    \pi(x_{a_j}^{(1)})\pi(x_{b_j}^{(2)})\pi(x_{c_j}^{(3)})(-1)^{s_j}\ket{\psi} &= \pi(x_{a_j}^{(1)}x_{b_j}^{(2)}x_{c_j}^{(3)}\sigma^{s_j})\ket{\psi} = \pi(h_j) \ket{\psi} = \ket{\psi}
\end{align}
where we used that $h_j \in H$ and \Cref{eq:sigma_H} on the final line. Then the strategy achieves value $\omega^* = 1$ by \cref{eq:satcond}. 
\end{proof}

As we shall see in this paper we find it much easier to study the question of whether $\sigma \in H^E$ rather than if $\sigma \in H$. The following lemma shows that that these conditions are equivalent. 

\begin{lem}
For any $k$XOR game, $\sigma \in H $ iff $\sigma \in H^E$. 
\label{lemma:HvsHE}
\end{lem}

\begin{proof}
The direction $\sigma \in H^E \Rightarrow \sigma \in H$ is immediate. 

 To see the converse direction, note that each clause $h_i$ contains exactly one generator $x_i^{(\alpha)}$ for each $\alpha \in [k]$. Then an odd length sequence of clauses contains an odd number of generators $x_i^{(\alpha)}$ for each $\alpha \in [k]$. Because all the relations of $G$ relate words containing an even number of $x_i^{(\alpha)}$ generators to the identity, the parity of the number of generators corresponding to each player remains fixed when applying the relations of $G$.  Then any word in $G$ which is equal to the product of an odd number of clauses from $H$ contains an odd number of generators corresponding to each player $\alpha$. Thus the word contains at least one generator corresponding to each player $\alpha$ and hence cannot equal $\sigma$. 

From this, we conclude that if $\sigma \in H$ there is a product an even number of clauses $h_1h_2 ... h_{2l} \in H^E$ which equals $\sigma$, thus $\sigma \in H^E$ as well. 
\end{proof}

\subsubsection{Sufficient conditions for \texorpdfstring{$k$XOR}{kXOR} games to have \texorpdfstring{$\omega_{co}^* = 1$}{commuting operator value 1} } 
\label{subsubsec:Sufficient_conditions}

\Cref{thm:IFF,lemma:HvsHE} imply that we could identify XOR games with value $\omega^*=1$ by solving instances of the subgroup membership problem in the groups $G$ or $G^E$. Unfortunately, the subgroup membership problem in these groups is, in general, undecidable.\footnote{A game group $G$ with $k\geq2$ and $n\geq3$ contains $\cF_2 \times \cF_2$ as a subgroup, where $\cF_2$ is the free group on two elements. This group has undecidable subgroup membership problem by \cite{mihajlova1971occurrence}. A similar argument applies to the group $G^E$.} Instead of reasoning about this problem directly it is helpful to consider a computationally simpler subgroup membership problem obtained by modding out the group $G^E$ by the normal subgroup $K$. We show this simpler problem can be solved in polynomial time.

\begin{thm} 
\label{thm:solvableModK}
Let $\sigma, H^E, K$ be defined relative to an $k$XOR game as described in \Cref{subsec:groups}. 
Let $[\sigma]_K$ be the coset containing $\sigma$ after modding $G^E$ out by $K$ and let $[H^E]_K = (H^E K)/K $ be the subgroup of $G^E/K$ generated by the cosets corresponding to generators of $H^E$. Then we can check if
$
    \left[\sigma \right]_K \notin [H^{E}]_K
$
in polynomial time. 
\end{thm}
\begin{proof}
First note that $K \triangleleft \; G^E$ and $H^E < G^E$, so the question is well defined. To show a polynomial time algorithm, note that $G^{E}/K$ is an abelian group -- in fact we have modded out by exactly the commutator subgroup of $G^E$. The subgroup membership problem for any abelian group can be solved in polynomial time (see \Cref{thm:SMPabelian} in \Cref{app:AlgComb}), so the result follows. 
\end{proof}

An immediate consequence of \Cref{lemma:HvsHE,thm:IFF} is that
\begin{align}
    [\sigma]_K \notin [H^{E}]_K  \implies \sigma \notin H^E \Leftrightarrow \sigma \notin H
    \implies \text{ the associated XOR game has $\omega_{co}^*=1$. }
    \end{align}
    Then, \Cref{thm:solvableModK} tells us that a sufficient condition for an XOR game to have $\omega_{co}^* = 1$ can be checked in polynomial time. In fact we can say something stronger: when the condition given by \Cref{thm:solvableModK} is met 
    an optimal strategy 
    can be chosen from
    a simple family of strategies which generalize the regular 3 qubit GHZ strategy. We introduce these strategies in \Cref{def:MERP}.

\begin{defn}
\label{def:MERP}
[MERP strategies] 
A MERP (maximally entangled, relative phase) strategy for a $k$XOR game is one where the players share the $k$-qubit GHZ state $\ket{\psi} = \frac{1}{\sqrt{2}}\left(\ket{11...1} + \ket{00...0} \right)$ and, given question j, the $\alpha$-th player measures the $\alpha$-th qubit of the state with a strategy observable of the form
\begin{align}
  M_j^{(\alpha)} :=  \exp(i \theta_j^{(\alpha)} \sigma_z)\sigma_x\exp(-i \theta_j^{(\alpha)} \sigma_z)
\end{align}
 where $\sigma_x, \sigma_z$ are the Pauli $X$ and $Z$ matrices: 
$\sigma_x = \begin{pmatrix} 0 & 1 \\
    1 & 0
    \end{pmatrix} 
$ and $
\sigma_z = \begin{pmatrix} 1 & 0 \\
    0 & -1
\end{pmatrix}.$
\footnote{In the language of \Cref{subsec:Strategies}, the state $\ket{\psi}$ lives in the Hilbert space $\left(\mathbb{C}^2\right)^k$ and, given question $j$, player $\alpha$ measures a strategy observable of the form $ I^{\otimes \alpha -1}\otimes M_j^{(\alpha)} \otimes I^{\otimes k - \alpha}$ where $I$ is the 2 by 2 identity matrix.}

The angle $\theta_j^{(\alpha)}$ depends on the player index $\alpha$ along with the question $j$ sent to the player. To specify a MERP strategy we just need to specify the angles $\theta_j^{(\alpha)}$ for every $j$ and $\alpha$. For this reason we refer to the set of angles 
$\{ \theta_j^{(\alpha)} : \alpha \in [k], j \in [N] \} $
as a description of the strategy. 
\end{defn}

The MERP strategy observables for any choice of $\theta_j^{(\alpha)}$ are valid strategy observables, that is, they are hermitian with  eigenvalues $\pm1$ and observables corresponding to different players commute. 

We can now state the relationship between MERP strategies and the condition $\sigma \notin H^E \pmod{K}$. 

\begin{thm}
\label{thm:PREF-MERP}
Let $\sigma, H^E, K$ be defined relative to an $k$XOR game as described in \Cref{subsec:groups} and define $[\sigma]_K,  [H^E]_K$ as in \Cref{thm:solvableModK}. Then
$
    [\sigma]_K \notin [H^E]_K 
$
iff the game has $\omega_{co}^* = \omega_{tp}^* = 1$ with a perfect value MERP strategy. A description of this strategy can be found in polynomial time. 
\end{thm}

\begin{proof}
This theorem is a rephrasing of
Theorem 5.30 from  \cite{watts2018algorithms}, where the condition $[\sigma]_K \notin H \pmod{K}$ was referred to as existence of a PREF (parity refutation). The equivalence between the $\sigma \in H \pmod{K}$ condition and existence of a parity refutation is elaborated on in \Cref{app:PREF-sigmaInH}.

In \Cref{subsec:MerpModK} we prove the theorem in one direction by showing that MERP matrices satisfy the defining relations for $K$. The other direction is proved by defining a system of linear diophantine equations which are solved only when $[\sigma]_K \in H^E \pmod{K}$ then showing, via a theorem of alternatives, that these equations being unsatisfied implies a MERP strategy can achieve value 1. 

\end{proof}

\subsubsection{For 3 player games the sufficient conditions are necessary}

 Theorems \labelcref{thm:PREF-MERP} gives a necessary and sufficient condition characterizing when an XOR game has a perfect MERP strategy. This also gives a sufficient (but not, in general, necessary) condition for a game to have $\omega^*_{co} = 1$.\footnote{The ``123 game'' introduced in \cite{watts2018algorithms} gives an example of a game with a perfect commuting operator strategy but no perfect MERP strategy.} \Cref{thm:K_modding}, the main mathematical engine underlying
this paper, gives the surprising result that this sufficient condition is also necessary for 3XOR games. 

\begin{thm}
Let $\sigma, H^E, K$ be defined relative to an $3$XOR game as described in \Cref{subsec:groups} and define $[\sigma]_K,  [H^E]_K$ as in \Cref{thm:solvableModK}. Then 
\begin{align}
    [\sigma]_K \in [H^E]_K \quad \Leftrightarrow \quad \sigma \in H^E.
\end{align}

\label{thm:K_modding}
\end{thm}
The proof of this result is purely algebraic, but involved. We give the full proof in \Cref{sec:Kmodding}.

We now state the main result of the paper, which follows as a consequence of  \Cref{thm:IFF,thm:solvableModK,thm:PREF-MERP,thm:K_modding,lemma:HvsHE}.

\begin{thm} \label{thm:wholeAvocado}
A 3XOR game has value $\omega_{co}^* = 1$ iff it has a perfect value MERP strategy, implying $\omega_{co}^* = \omega_{tp}^* = 1$. Additionally, there exists a polynomial time algorithm which decides if a 3XOR game has value $\omega^*_{co} = 1$, and outputs a description of the perfect value MERP strategy if one exists. 
\end{thm}

\begin{proof}
By \Cref{thm:IFF}, an 3XOR game has $\omega_{co}^* = 1$ iff $\sigma \notin H$ in the associated group. By \Cref{thm:K_modding}, this is also equivalent to the statement $[\sigma]_K \in H \pmod{K}$. By \Cref{thm:PREF-MERP} this implies a MERP strategy, and the first part of the result follows. 

To get the polynomial time algorithm, we just need to check if $[\sigma]_K \in H \pmod{K}$, which we can do in polynomial time by \Cref{thm:solvableModK}. If true, there exists a MERP strategy and we can find it by \Cref{thm:PREF-MERP}. If false, the same chain of implications as above shows $\omega^*_{co} < 1$. 
\end{proof}

For $k > 3$ players, our arguments break down because we have no analog of \Cref{thm:K_modding}. 
Indeed it remains open whether there is any finite time algorithm for identifying  perfect $k$-player XOR games when $k > 3.$
Some speculation about possible $k$-player analogues of  \Cref{thm:K_modding} is provided in \Cref{app:K_intuition_and_K_players}.

\subsubsection{Bounds on the bias ratio}

Combining \Cref{thm:wholeAvocado} with a result from~\cite{perez2008unbounded} gives the following. 
\begin{thm}
A 3XOR game with $\omega^*_{co} = 1$ also has classical value $\omega > 1/2 + \frac{1}{8K_G^{\mathbb{R}}} \geq 0.57$, where $K_G^{\mathbb{R}}$ is the real Grothendieck constant.
\end{thm}

\begin{proof}
By \Cref{thm:wholeAvocado}, a 3XOR game $G$ with $\omega^*_{co} = 1$ must also have a perfect value MERP strategy. This strategy uses a GHZ state for the players, and a bound from~\cite{perez2008unbounded} gives that 
\begin{align}
    \beta_{GHZ}^*/\beta \leq 4K_G^{\mathbb{R}},
\end{align}
where $\beta_{GHZ}^*$ is the maximum bias achieved with a strategy using a GHZ state. 
But then 
\begin{align}
    \beta(G) &\geq \frac{\beta_{GHZ}^*(G)}{4K_G^{\mathbb{R}}} = \frac{1}{4K_G^{\mathbb{R}}} \\
    \implies \omega(G) &\geq \frac{1}{2} + \frac{1}{8K_G^{\mathbb{R}}}
\end{align}
and the result follows. 
\end{proof}

\subsection{Examples} 
\label{ssec:examples}

In this subsection we re-analyze some well known XOR games using the techniques developed in this paper. 

\subsubsection{The CHSH Game}

The first game we analyze is the CHSH game, introduced in~\cite{clauser1969proposed}. This is a two question, two player XOR game. Following convention, questions sent to the players are indicated with labels in $\{0,1\}$. The CHSH tests a system of 4 equations: 
\begin{align}
    \hat{X}^{(1)}_0 + \hat{X}^{(2)}_0 &= 0 & \hat{X}^{(1)}_0 + \hat{X}^{(2)}_1 &= 0 \nonumber \\
    \hat{X}^{(1)}_1 + \hat{X}^{(2)}_0 &= 0 & \hat{X}^{(1)}_1 + \hat{X}^{(2)}_1 &= 1. \nonumber
\end{align}
Following the procedure  as outlined in \Cref{subsec:groups} (\Cref{eq:sumToProd}) we see the clause group $H_{\text{CHSH}}$ associated with this game is generated by the clauses
\begin{align}
    \left\{x_0^{(1)}x_0^{(2)},\;\; x_0^{(1)}x_1^{(2)},\;\;
    x_1^{(1)}x_0^{(2)},\;\;
    x_1^{(1)}x_1^{(2)}\sigma\right\}.
\end{align}
We can multiply these clauses together and then simplifying using the relations of the game group $G$ to show
\begin{align}
    \left(x_0^{(1)}x_0^{(2)}\right) \left(x_0^{(1)}x_1^{(2)}\right)
    \left(x_1^{(1)}x_1^{(2)}\sigma\right)
    \left(x_1^{(1)}x_0^{(2)}\right) = \sigma \in H_{CHSH}.
\end{align}
We conclude the CHSH game does not have a perfect commuting operator strategy by \Cref{thm:IFF}.

\subsubsection{The GHZ Game}
Next, we analyze the GHZ game, introduced in \cite{greenberger1990bell}. This is a 3 player game testing a system of equations
\begin{align}
    \hat{X}^{(1)}_0 + \hat{X}^{(2)}_0 + \hat{X}^{(3)}_0 &= 0 &
    \hat{X}^{(1)}_1 + \hat{X}^{(2)}_1 + \hat{X}^{(3)}_0 &= 1 \nonumber \\
    \hat{X}^{(1)}_0 + \hat{X}^{(2)}_1 + \hat{X}^{(3)}_1 &= 1 & 
    \hat{X}^{(1)}_1 + \hat{X}^{(2)}_0 + \hat{X}^{(3)}_1 &= 1. \nonumber
\end{align}
Thus the associated clause group $H_{GHZ}$ is generated by clauses
\begin{align}
\left\{x_0^{(1)}x_0^{(2)}x_0^{(3)},\;\; 
    x_1^{(1)}x_1^{(2)}x_0^{(3)}\sigma,\;\;
    x_0^{(1)}x_1^{(2)}x_1^{(3)}\sigma,\;\;
    x_1^{(1)}x_0^{(2)}x_1^{(3)}\sigma\right\}.
\end{align}
The GHZ game has a perfect  MERP strategy. Here, we reprove this result using the techniques developed in the paper. 

The first step is to construct the even clause group $H_{GHZ}^E$, which is generated by the pairs of clauses
\begin{align}
    &\big\{ x_0^{(1)}x_0^{(2)}x_0^{(3)}x_1^{(1)}x_1^{(2)}x_0^{(3)}\sigma, \;\;\;
    x_0^{(1)}x_0^{(2)}x_0^{(3)} x_0^{(1)}x_1^{(2)}x_1^{(3)}\sigma, \;\;\;
    x_0^{(1)}x_0^{(2)}x_0^{(3)} x_1^{(1)}x_0^{(2)}x_1^{(3)}\sigma, \;\;\; \nonumber \\
    &\hspace{20pt}
    x_1^{(1)}x_1^{(2)}x_0^{(3)}\sigma x_0^{(1)}x_1^{(2)}x_1^{(3)}\sigma, \;\;\;
    x_1^{(1)}x_1^{(2)}x_0^{(3)}\sigma x_1^{(1)}x_0^{(2)}x_1^{(3)}\sigma, \;\;\;
    x_0^{(1)}x_1^{(2)}x_1^{(3)}\sigma x_1^{(1)}x_0^{(2)}x_1^{(3)}\sigma
    \big\}.
\end{align}
(and, by definition, their inverses). Simplifying these using the relations of the game group $G$ gives generating set 
\begin{align}
    &\Big\{ \left(x_0^{(1)}x_1^{(1)}\right) \left(x_0^{(2)}x_1^{(2)} \right) \sigma, \;\;\;
    \left(x_0^{(2)}x_1^{(2)} \right) \left( x_0^{(3)}x_1^{(3)} \right) \sigma, \;\;\;
    \left( x_0^{(1)}x_1^{(1)} \right) \left( x_0^{(3)}x_1^{(3)} \right) \sigma, \;\;\; 
    \nonumber 
    \\
    &\hspace{40pt}
    \left( x_1^{(1)}x_0^{(1)} \right) \left(x_0^{(3)}x_1^{(3)}\right), \;\;\;
    \left( x_1^{(2)}x_0^{(2)} \right) \left(x_0^{(3)}x_1^{(3)}\right), \;\;\;
    \left( x_0^{(1)}x_1^{(1)} \right) \left( x_1^{(2)}x_0^{(2)} \right)
    \Big\}.
\end{align}
where bracketed terms now indicate generators of $G^E$. Working mod $K$ all the bracketed terms commute with each other,\footnote{A careful reader might notice that all the bracketed terms actually commute with each other even before modding out by the subgroup $K$. This is a consequence of the fact that the $GHZ$ game is a two question game, but doesn't hold in general. Elaborating on this observation, it is possible to show that a two question XOR game with any number of players has a perfect commuting operator strategy iff it has a perfect MERP strategy, giving a special case of the result shown in \cite{werner2001all}.} so now straightforward linear algebra can be used to show that 
\begin{align}
    \sigma \notin H^E_{GHZ} \pmod{K}.
\end{align}
Then we see that the GHZ game has a perfect MERP strategy by \Cref{thm:PREF-MERP}.

While we didn't use it in either of these examples \Cref{thm:K_modding} tells us that the techniques used above to analyze the GHZ game can be used to analyze any 3 player game. In particular, analyzing any 3-player game $G$ with even clause group $H^E_G$ we will either find that that 
\begin{align}
    \sigma \notin H^E_G \pmod{K}
\end{align}
and the game (like the GHZ game) has a perfect MERP strategy or 
\begin{align}
    \sigma \in H^E_G \pmod{K} \Leftrightarrow \sigma \in H^E
\end{align}
by \Cref{thm:K_modding} and so the game has no perfect commuting operator strategy of any kind.

\subsection{Acknowledgements}
The authors thank Igor Klep, Aram Harrow, Gurtej Kanwar, Anand Natarajan, William Slofstra, and Jop Briet for discussions and helpful comments. They also thank Zinan Hu and Zehong Zhao for providing numerical examples valuable to our understanding. They heartily thank both an anonymous referee and Taro Spirig for careful reading and helpful comments.

\newpage

\section{Technical Details} \label{sec:Kmodding}

This section begins with definitions, then compares the algebraic structure defined in this paper to the one introduced in \cite{cleve2017perfect}, then proves \cref{thm:K_modding}.

\subsection{Background and Definitions}  \label{sec:techDefn}

We briefly recap the definitions given in \Cref{sec:Notation}, then give some additional notation that will be useful in this section. In everything that follows $\com{\;}{\;}$ denotes the group commutator, so $\com{x}{y} = xyx^{-1}y^{-1}$. 

\subsubsection{Recap} \label{subsec: recap}

We consider a 3XOR game with questions drawn from an alphabet of size $[N]$. The game has $m$ question vectors labeled $(a_1,b_1,c_1), ...., (a_m,b_m,c_m)$ with $a_i, b_i, c_i \in [N]$. When asked the $i$-th question vector $(a_i,b_i,c_i)$ players win the game if their responses sum (mod 2) to the parity bit $s_i \in \{0,1\}$. Parity bits are defined for all $i \in [m]$. 

There are several algebraic objects associated with the game. The first is the game group $G$, defined to be the group generated by the set of elements
\begin{align}
    \{ x_i^{(a)} : i \in [n], \alpha \in [3] \} \cup \{\sigma\} 
\end{align}
with relations 
\begin{enumerate}
    \item $\left(x_i^{(\alpha)}\right)^2 = 1$ for all $i,\alpha$
    \item $\com{x_i^{(\alpha)}}{x_j^{(\beta)}} = 1$ for all $i,j, \alpha \neq \beta$
    \item $\sigma^2 = 1$
     \item $\com{\sigma}{{x_i^{(\alpha)}}} = 1$ for all $i, \alpha$.
\end{enumerate}

The generators $x_i^{(\alpha)}$ correspond to the observables measured by player $\alpha$ upon receiving question number $i$. The group element $\sigma$ should be though of as a formal variable corresponding to $-1$ in the group. Note $\sigma$ has order two ($\sigma^2 = 1$) and commutes with all elements of group ($\com{\sigma}{w} = \id$ for any $w \in G$). 

For all $i \in [m]$ we define the associated clause
\begin{align}
    h_i = x^{(1)}_{a_i}x^{(2)}_{b_i}x^{(3)}_{c_i}\sigma^{s_i}.
\end{align}
The clause set $S = \{h_i\}_{i \in [m]}$ contains all clauses of the game. The clause group $H = \langle S \rangle$ is the subgroup of $G$ generated by the clauses. 

The even game group $G^{E}$ is the subgroup of $G$ consisting of words with an even number of generators corresponding to each player and possibly the element $\sigma$, so  
\begin{align}
    G^{E} = \angles{\left\{ x_i^{(\alpha)}x_j^{(\alpha)} : i,j \in [N], \alpha \in [k] \right\} \cup \{\sigma\} }.
\end{align}
The even clause group is the subgroup of $G$ generated by an even number of clauses
\begin{align}
H^{E} = \angles{\left\{ h_ih_j : i,j \in [m] \right\}}.
\end{align}
An important observation is that $H^{E}$ is a subgroup of $G^E$. 

Finally, $K$ is the commutator subgroup of $G^E$, defined to be the normal closure of the set of commutators of the generators of $G^E$. In math:
\begin{align}
    K = \angles{\left\{\com{x_i^\alpha x_j^\alpha}{x_k^\alpha x_l^\alpha} : i,j,k,l \in [n], \alpha \in [3] \right\}}^{G^E}
\end{align}
Where $\langle X \rangle^{Y}$ denotes the normal closure of the set $X$ in the group $Y$. 

\subsubsection{Projections and Clause Graphs} \label{subsec: new def}

It will be helpful to have notation for referring to just the observables associated with a single player. To this end, define \df{player subgroups} $ G_\alpha \leq G $ by 
\begin{align}
    G_\alpha = \left\langle \left\{x_i^{(\alpha)}: i \in [N] \right\} \right\rangle
\end{align}
and $ G_\alpha^{E} \leq G^E$ by 
\begin{align}
    G_\alpha^{E} = \left\langle \left\{x_i^{(\alpha)}x_j^{(\alpha)}: i,j \in [N] \right\} \right\rangle
\end{align}
for all $\alpha \in \{1,2,3\}$. One advantage to working with the player subgroups $G_\alpha$ and $G_\alpha^E$ is that they have simple group presentations. We give these presentations in the following lemmas. 
\begin{lem} \label{claim:G_alpha_pres}
The group $G_\alpha \leq G$ is finitely presented, with generating set
\begin{align}
    \{x_i^{(\alpha)} : i \in [N] \} \label{eq:G_alpha_gens}
\end{align} and relations 
\begin{align}
    \left(x_i^{(\alpha)}\right)^2 = 1 \;\; \text{ for all } \;\; i \in [N]. \label{eq:G_alpha_relns}
\end{align}
\end{lem}

\begin{proof}
Since $G_\alpha$ was defined to be the subgroup of $G$ generated by the elements $x_i^{(\alpha)}$, it is clear that any element in $G_\alpha$ is a word in the generators presented above. 

Because the relations given above are clearly true in the group $G_{\alpha}$, all that remains to show is that the relations given in \Cref{eq:G_alpha_relns} can transform two words into one another if they are equal in the group $G_{\alpha}$. To prove this, we first say a word consisting of $x_i^{(\alpha)}$ generators is in fully reduced form iff:
\begin{enumerate}
\item There are no $\left(x_i^{(\alpha)}\right)^{-1}$ in the word and
\item No two $x_i^{(\alpha)}$ with the same value of $i$ are adjacent in the word.
\end{enumerate}
Any word made up of the generators given in \Cref{eq:G_alpha_gens} can be put in fully reduced form by repeated application of the relations given in \Cref{eq:G_alpha_relns} (first by the replacement 
\begin{align}
    \left(x_i^{(\alpha)}\right)^{-1} = \left(x_i^{(\alpha)}\right)^{-1}\left(x_i^{(\alpha)}\right)^2 = x_i^{(\alpha)}
\end{align}
and then by deleting any two adjacent instances of the generator $x_i^{(\alpha)}$). Furthermore, it is clear that two words in $G_\alpha$ made up of $x_i^{(\alpha)}$ generators are equal iff their fully reduced forms are equal (i.e. fully reduced forms are canonical forms for words in $G_\alpha$). This shows that any two words in $G_\alpha$ made up of $x_i^{(\alpha)}$ generators can be transformed into each other via the relations given in \Cref{eq:G_alpha_relns} iff their canonical forms are equal. The claim follows. 
\end{proof}

\begin{lem}
\label{claim:G_alpha^E_pres}
The group $G_\alpha^E \leq G^E$ is finitely presented, with generating set
\begin{align}
    \{x_i^{(\alpha)}x_j^{(\alpha)} : i,j \in [N] \} \label{eq:G_alpha^E_gens}
\end{align} and relations 
\begin{enumerate}[(1)]
    \item $x_i^{(\alpha)} x_j^{(\alpha)} x_j^{(\alpha)} x_k^{(\alpha)} = x_i^{(\alpha)}x_k^{(\alpha)}$ for all $i,j,k \in [N]$ \label{item:G_alpha^E_rel_adj}
    \item $x_i^{(\alpha)} x_j^{(\alpha)} x_j^{(\alpha)} x_i^{(\alpha)} = 1$ for all $i,j \in [N]$. \label{item:G_alpha^E_rel_inv}
\end{enumerate}
\end{lem}

\begin{proof}
Similarly to the proof of \Cref{claim:G_alpha_pres}, it is immediate that the generators given in \Cref{eq:G_alpha^E_gens} generate the group $G_\alpha^E$. 

Also similarly to the proof of \Cref{claim:G_alpha_pres}, we show that the relations given above can transform two words constructed from the generating set given in \Cref{eq:G_alpha^E_gens} into one another if they are equal in the group $G_\alpha^E$ by showing the relations can put words in fully reduced form. To see this first notice we can remove inverses using relation~\labelcref{item:G_alpha^E_rel_inv} and the argument 
\begin{align}
    \left(x_i^{(\alpha)}x_j^{(\alpha)}\right)^{-1} =
    \left(x_i^{(\alpha)}x_j^{(\alpha)}\right)^{-1} x_i^{(\alpha)} x_j^{(\alpha)} x_j^{(\alpha)} x_i^{(\alpha)} = x_j^{(\alpha)} x_i^{(\alpha)}
\end{align}
and then remove any adjacent $x_i^{(\alpha)}$ elements using relation \labelcref{item:G_alpha^E_rel_adj}. The proof follows. 
\end{proof}

Because observables corresponding to different players commute, we can write any $w \in G$ as 
\begin{align}
    w = w_1 w_2 w_3 \sigma^{s_w} \label{eq:playerDecomp}
\end{align}
where $w_\alpha \in G_{\alpha}$ for all $\alpha \in \{1,2,3\}$, and $s_w \in \{0,1\}$. Similarly, any $w' \in G^{E}$ can be written as 
\begin{align} \label{eq:typo1}
    w' = w_1' w_2' w_3' \sigma^{s_w'}
\end{align}
with $w_\alpha' \in G_{\alpha}^{E}$ and $s_w' \in \{0,1\}$. 

For any $\alpha \in \{1,2,3\}$ we also define the projector onto player subgroups $\proj{\alpha} : G \rightarrow G_\alpha$ by defining its action on the generators of $G$:
\begin{align}
    \proj{\alpha}(x_i^{(\beta)}) = \begin{cases}
    x_i^{(\beta)} &\text{ if } \alpha = \beta \\
    \id &\text{ otherwise} 
    \end{cases}
    \hspace{26pt} \text{ and } \hspace{26pt}  \proj{\alpha}(\sigma) = \id
\end{align}
then extending $\proj{\alpha}$ to a homorphism on $G$. To see this defines a valid homomorphism note that it preserves the group relations:
\begin{align}
\proj{\alpha}\left(x_i^{(\beta)}\right)^2 = \begin{cases}
\left(x_i^{(\beta)}\right)^2 = 1 &\text{ if } \alpha = \beta \\
\id^2 = \id &\text{ otherwise }
\end{cases}
\end{align}
with a similarly simple argument showing commutation relations are preserved. It is also helpful to define a projection $\proj{\sigma}$ which acts on the generators of $G$ as 
\begin{align}
    \proj{\sigma}\left(x_i^{(\beta)}\right) = 1 \hspace{26pt} \text{ and } \hspace{26pt}  \proj{\sigma}(\sigma) = \sigma.
\end{align}
Combining \Cref{eq:playerDecomp} with the definition of $\proj{\alpha}$ gives the equation
\begin{align}
    w = \proj{1}(w)\proj{2}(w)\proj{3}(w)\proj{\sigma}(w) \label{eq:playerDecomp2}
\end{align}
for any $w \in G$. 

Next, we define the clause (hyper)graph\footnote{A hypergraph is a graph with edges passing through more than two vertices.} $\cG_{123}$ which gives a useful way of visualizing the clause structure of a game. The graph has $3N$ vertices which we identify with the generators $x_i^{\alpha}$ of the group $G$. We label the vertices by the corresponding generator. Hyperedges in the graph correspond to clauses, with a hyperedge going through vertices $x_{a_i}^{(1)}$, $x_{b_i}^{(2)}$, and $x_{c_i}^{(3)}$ for every clause $x_{a_i}^{(1)}x_{b_i}^{(2)}x_{c_i}^{(3)}\sigma^{s_i} \in S$. Note that the existence of the hyperedge is independent of the value of $s_i$, so the clause graph contains no information about the parity bits. Because edges in the hypergraph correspond to clauses $h \in S$, we can identify any sequence of edges in $\cG_{123}$ with a word $w \in H$. We will use this relationship frequently in the future. 

We also define important subgraphs of $\cG_{123}$ by taking the induced graphs on vertices corresponding to a subset of players.\footnote{For a graph $\mathcal{X} = (V,E)$, the subhypergraph induced by a set of vertices $V' \subseteq V$ is the hypergraph with vertex set $V'$ and edge set $E' = \{e \cap V': e \in E\}$. Essentially, edges are all truncated to the vertices in $V'$.} For any $\alpha \neq \beta \in \{1,2,3\}$ we define the multigraph $G_{\alpha\beta}$ to be subgraph of $\cG_{123}$ induced by the vertices corresponding to generators of $G_{\alpha}$ and $G_{\beta}$.
See \Cref{fig:Sampleinducedhypergraph} for an example. As with the graph $\cG_{123}$, edges in the graph $\cG_{\alpha \beta}$ can be identified with clauses in $H$ and sequences of edges in $\cG_{\alpha \beta}$ can be identified with words $w \in H$. 

\begin{figure}[p]
    \centering
    \captionsetup{singlelinecheck=off}
    \resizebox{10cm}{!}{
        \begin{tikzpicture}

\vertex[xpos=0, ypos=0, label name = $x_1^{(1)}$]{11}
\vertex[xpos=1, ypos=0, label name = $x_2^{(1)}$]{12}
\vertex[xpos=2, ypos=0, label name = $x_3^{(1)}$]{13}
\vertex[xpos=3, ypos=0, label name = $x_4^{(1)}$]{14}
\vertex[xpos=4, ypos=0, label name = $x_5^{(1)}$]{15}
\vertex[xpos=5, ypos=0, label name = $x_6^{(1)}$]{16}

\vertex[xpos=0, ypos=2, label position = above, label name = $x_1^{(2)}$]{21}
\vertex[xpos=1, ypos=2, label position = above, label name = $x_2^{(2)}$]{22}
\vertex[xpos=2, ypos=2, label position = above, label name = $x_3^{(2)}$]{23}
\vertex[xpos=3, ypos=2, label position = above, label name = $x_4^{(2)}$]{24}
\vertex[xpos=4, ypos=2, label position = above, label name = $x_5^{(2)}$]{25}
\vertex[xpos=5, ypos=2, label position = above, label name = $x_6^{(2)}$]{26}

\vertex[xpos=0, ypos=4, label position = above, label name = $x_1^{(3)}$]{31}
\vertex[xpos=1, ypos=4, label position = above, label name = $x_2^{(3)}$]{32}
\vertex[xpos=2, ypos=4, label position = above, label name = $x_3^{(3)}$]{33}
\vertex[xpos=3, ypos=4, label position = above, label name = $x_4^{(3)}$]{34}
\vertex[xpos=4, ypos=4, label position = above, label name = $x_5^{(3)}$]{35}
\vertex[xpos=5, ypos=4, label position = above, label name = $x_6^{(3)}$]{36}

\threeedge[outline = false, fill = true, fill color = gray]{11}{22}{31}

\threeedge[outline = false, fill = true, fill color = gray]{12}{22}{32}

\threeedge[outline = false, fill = true, fill color = gray]{13}{24}{34}
\threeedge[outline = false, fill = true, fill color = gray]{14}{24}{33}

\threeedge[outline = false, fill = true, fill color = gray]{15}{26}{35}

\threeedge[outline = false, fill = true, fill color = gray]{16}{26}{36}

\threeedge[outline = false, fill = true, fill color = gray]{11}{21}{31}
\threeedge[outline = false, fill = true, fill color = gray]{11}{23}{33}
\threeedge[outline = false, fill = true, fill color = gray]{12}{23}{34}
\threeedge[outline = false, fill = true, fill color = gray]{15}{24}{34}
\threeedge[outline = false, fill = true, fill color = gray]{15}{25}{35}

\end{tikzpicture}
    }
    \caption[]{Sample hypergraph $\cG_{123}$ for a game with alphabet size $N = 6$ and 11 clauses. The hypergraph is generated by clause set ($\sigma$ terms omitted since they don't affect the graph): 
    \begin{align*}
    S = \{ x_1^{(1)}x_1^{(2)}x_1^{(3)},
    x_1^{(1)}x_2^{(2)}x_1^{(3)},
    x_2^{(1)}x_2^{(2)}x_2^{(3)},
    x_1^{(1)}x_3^{(2)}x_3^{(3)},
    x_2^{(1)}x_3^{(2)}x_4^{(3)}, 
    x_3^{(1)}x_4^{(2)}x_4^{(3)}, \; \; \; \; \; \; \; \; \\
    x_4^{(1)}x_4^{(2)}x_3^{(3)},
    x_5^{(1)}x_4^{(2)}x_4^{(3)},
    x_5^{(1)}x_6^{(2)}x_5^{(3)},
    x_5^{(1)}x_5^{(2)}x_5^{(3)},
    x_6^{(1)}x_6^{(2)}x_6^{(3)}
    \}     
    \end{align*}
    }
    \label{fig:Samplehypergraph}
\vspace{20pt}
    \resizebox{10cm}{!}{
        \begin{tikzpicture}


\vertex[xpos=0, ypos=2, label position = below, label name = $x_1^{(2)}$]{21}
\vertex[xpos=1, ypos=2, label position = below, label name = $x_2^{(2)}$]{22}
\vertex[xpos=2, ypos=2, label position = below, label name = $x_3^{(2)}$]{23}
\vertex[xpos=3, ypos=2, label position = below, label name = $x_4^{(2)}$]{24}
\vertex[xpos=4, ypos=2, label position = below, label name = $x_5^{(2)}$]{25}
\vertex[xpos=5, ypos=2, label position = below, label name = $x_6^{(2)}$]{26}

\vertex[xpos=0, ypos=4, label position = above, label name = $x_1^{(3)}$]{31}
\vertex[xpos=1, ypos=4, label position = above, label name = $x_2^{(3)}$]{32}
\vertex[xpos=2, ypos=4, label position = above, label name = $x_3^{(3)}$]{33}
\vertex[xpos=3, ypos=4, label position = above, label name = $x_4^{(3)}$]{34}
\vertex[xpos=4, ypos=4, label position = above, label name = $x_5^{(3)}$]{35}
\vertex[xpos=5, ypos=4, label position = above, label name = $x_6^{(3)}$]{36}

\edge{22}{31}

\edge{22}{32}

\edge{24}{34}
\edge{24}{33}

\edge{26}{35}

\edge{26}{36}

\edge{21}{31}
\edge{23}{33}
\edge{23}{34}
\edge{24}{34}
\edge{25}{35}

\end{tikzpicture}
    }
    \caption[]{Induced graph $\cG_{23}$ corresponding to the same clause set as \Cref{fig:Samplehypergraph}}
    \label{fig:Sampleinducedhypergraph}
\end{figure}

In \Cref{sec: Connectivity} we show that we can restrict our attention to the case where $\cG_{123}$ is a connected graph. The induced graph $\cG_{\alpha\beta}$ can be disconnected, and the different connected components of this graph (and representative elements from each) play an important role in the proof in \Cref{sec:K_moddingProof}.

\subsubsection{Defining Homomorphisms via Group Presentations}

We now recap a standard algebraic result which we shall use frequently when making arguments involving the groups $G_\alpha$ and $G_\alpha^E$. In the following Lemma we describe a group as being presented by a set of generators $S$ and relations $R$. It is understood that these relations correspond to the set of equations $\{r = 1\}$ for all $r \in R$. 

\begin{lem} \label{thm:alg_com_hom}
Let $G$ be a group presented by the set of generators $S$ and relations $R$. Let the group $H$ be arbitrary and 
\begin{align}
    f : S \rightarrow H
\end{align}
be some function mapping generators of $G$ to elements in the group $H$. Then $f$ can be extended to a homomorphism $f: G \rightarrow H$ which acts on inverses as
\begin{align}
    f(s^{-1}) = f(s)^{-1}
\end{align}
and on words $s_1 s_2 ... s_l \in G$ as
\begin{align}
    f(s_1 s_2 ... s_l) = f(s_1) f(s_2) ... f(s_l)
\end{align}
iff
\begin{align}
    f(r) = 1 
\end{align}
for all $r \in R$.
\end{lem}

\begin{proof} The only if direction is clear, since $f(r) \neq 1$ implies that $f(1) \neq 1$ and so $f$ can't be a homomorphism.

To prove the if direction, we first show is that $f$ is well defined. To see this, note that any two words $s_1 s_2 ... s_l$ and $t_1 t_2 ... t_k$ made up of elements from the generating set $S$ are equal in $G$ iff 
\begin{align}
    t_1 t_2 ... t_k = s_1 s_2 ... s_k \prod_{i} w_i r_i w_i^{-1}
\end{align}
where words $w_i \in G_\alpha$ are arbitrary, each $r_i$ is in $R$, and equality in the equation above now holds as words (that is, the only thing that needs to be cancelled are elements adjacent to their own inverse). Then
\begin{align}
    f(t_1 t_2 ... t_k) &= f\left(s_1 s_2 ... s_k \prod_{i} w_i r_i w_i^{-1}\right) \\
    &=f(s_1 s_2 ... s_k)  \prod_{i} f(w_i) f(r_i) f(w_i)^{-1} \\
    &=f(s_1 s_2 ... s_k)
\end{align}
and it is clear the function $f$ is well defined. From here it is clear that $f$ is a homomorphism, since for any words $w_1$ and $w_2$ we have 
\begin{align}
    f(w_1w_2) = f(w_1) f(w_2),
\end{align}
and we are done. 
\end{proof}

In practice, given a function $f$ mapping the generators of some group $G$ into a group $H$ and satisfying the conditions of \Cref{thm:alg_com_hom}, we will refer to the homomorphism $f: G \rightarrow H$ constructed using the above procedure as the homorphism constructed by ``extending $f$ in the natural way'', or with similar language.

 \subsection{Comparison with Linear Systems Games}
 \label{sec:comp}

A reader familiar with the work of Cleve, Liu and Slofstra concerning linear systems games~\cite{cleve2017perfect} may notice a similarity between the solution group defined in that paper and the clause group defined in this work. 
In this section we give a direct comparison between the two. Our goal in doing this is not to provide any deep insights -- we simply hope a direct comparison will help a reader already familiar with linear systems games to better understand our work. We do not define linear systems games here, and point readers to \cite{cleve2017perfect} for a formal introduction to them. This section is not critical and a reader can safely skip it without impacting their understanding of the rest of this paper. 

Following \cite{cleve2017perfect}, we consider a binary linear system of $m$ equations on $n$ variables $Mx = b$, with $M \in \mathbb{Z}_2^{m \times n}$ and $b \in \mathbb{Z}^m$. $M_{ij}$ specifies an individual entry in the matrix $M$, and $b_i$ specifies an entry from the vector $b$. The solution group of the binary linear system is a group with generators $g_1, g_2 ..., g_n, J$ and relations
\begin{enumerate}
    \item $g_i^2 = 1$ for all $i \in [n]$ and $J^2 = 1$
    \item $\com{g_i}{J} = 1$ for all $i \in [n]$
    \item $\com{g_i}{g_j} =1$ if $x_i$ and $x_j$ appear in the same equation (that is $M_{li} = M_{lj} = 1$ for some $l \in [m]$).
    \item $\prod_i \left( g_i^{M_{li}}\right) J^{b_l} = 1$ for all $l \in [m]$.
\end{enumerate}
In~\cite{cleve2017perfect} the authors showed the following result:
\begin{thm}[Implied by Theorem 4 of \cite{cleve2017perfect},  paraphrased] The linear system game associated to the system of equations $Mx = b$ has a perfect value commuting operator strategy iff in the associated solution group we have $J \neq 1$. 
\label{thm:LSG}
\end{thm}

\Cref{thm:IFF} can be thought of as an analog of \Cref{thm:LSG} for 3XOR games. We can restate \Cref{thm:LSG} in a way that makes the comparison even more apparent. 

Given a system of equations $Mx = b$, define the group $G_{lsg}$ to be the group with generators $g_1, g_2, ... , g_n, J$ and relations 1-3 above. Note that $J \neq 1$ in this group. Next, define the subgroup $H_{lsg} \, \triangleleft \, G_{lsg}$ to be the normal closure in $G_{lsg}$ of the words corresponding to equations in the system of equations $Mx = b$ (that is, the words involved in relation 4 above) so 
\begin{align}
    H_{lsg} = \left\langle \left\{ \prod_i  \left(g_i^{M_{li}}\right) J^{b_l}  : l \in [m] \right\} \right\rangle ^{G_{lsg}}.
\end{align}
Using these definitions, an equivalent statement of  \Cref{thm:LSG} is: 
\begin{thm}[Restatement of \Cref{thm:LSG}] The linear system game associated to the system of equations $Mx = b$ has a perfect value commuting operator strategy iff $J \notin H_{lsg}$.
\label{thm:LSG2}
\end{thm}
We can compare the above theorem and \Cref{thm:IFF} directly. We list, and briefly discuss, the key differences:
\begin{enumerate}[\it{i)}]
    \item {\it The group $G$ contains an element for every question player combination, while $G_{lsg}$ only contains an element for every question}. In a commuting operator (or tensor product) strategy for an XOR game, different players can measure completely different observables when sent the same question and so we need a different group element to correspond to each player-question combo.\footnote{Put (informally) in slightly different terms: XOR games can be very far from synchronous, as defined in \cite{helton2017algebras}.} Conversely, in linear systems games there is a close relationship between Alice and Bob's measurements given the same question, and both players measurement operators can be constructed from representations (right and left actions) of the same group elements. 
    
    \item {\it Generators of $G_{lsg}$ commute with each other if they appear in the same equation (relation 3 above). Generators of $G$ satisfy no such relation.} This difference reflects a difference between linear system games and XOR game strategies. In a linear systems game a single player must make simultaneous measurements of all the operators corresponding to a question in the game. This never happens in XOR games. From an algebraic point of view, these extra relations place a restriction on elements of $G_{lsg}$ that is not placed on elements of $G$.
    
    \item {\it The group $H_{lsg}$ is a normal subgroup of $G_{lsg}$, while $H$ is not a normal subgroup of $G$.} This has an algebraic consequence: asking if $J \in H_{lsg}$ is an instance of the word problem (mod out by the generators of $H_{lsg}$, then ask if $J$ equals the identity), while asking if $\sigma \in H$ is an instance of the subgroup membership problem. The word problem is in a sense ``easier" than the subgroup membership problem: there are groups with solvable word problem but undecidable subgroup membership problem~\cite{mihajlova1971occurrence}.  Still, both problems are undecidable in general. This difference also has consequences for game strategies. In a linear systems game, an identity of the form 
    \begin{align}
        \prod_i  \left(g_i^{M_{li}}\right) J^{b_l} = 1
    \end{align}
    holds in the group, hence holds as an operator identity on the strategy observables as well. In an XOR game, the operator identities codified in $H$ only need hold acting on the state $\ket{\psi}$ and there are games (for example, the GHZ game) where products of strategy observables act as the identity on $\ket{\psi}$, but the operators themselves do not multiply to the identity.
\end{enumerate}

We should also point out that a linear systems game can be defined for any system of equations of the form $Mx = b$, while XOR games require equations of a special form: exactly one variable corresponding to each player is involved in each equation. It is possible to define a slightly more general form of $k$XOR games with a subset of players, as opposed to all players, queried on each question but those are not considered here. 

\Cref{thm:wholeAvocado}, in combination with \cite{slofstra2020tsirelson} shows that there cannot exist a mapping which is computable in finite time and transforms linear systems games into XOR games while preserving the commuting operator value of the game. (Or else this mapping, in combination with \Cref{thm:wholeAvocado}, would give a finite time algorithm for deciding whether or not a linear systems game has perfect commuting-operator value. This is impossible by \cite{slofstra2020tsirelson}.)  The question of finding a natural map in the other direction remains open.

\subsection{Connectivity of the Clause Graph} \label{sec: Connectivity}

In \Cref{subsec: new def} we introduced the \emph{clause graph} $\cG_{123}$ -- a graphical representation of the clause structure of a 3XOR game. In this section we consider 3XOR games whose associated clause graph is not connected. Given such a game we can always define smaller games, each involving only the clauses corresponding to a single connected component of the clause graph. Here, we show a 3XOR game has $\omega^*_{co} = 1$ iff each of these smaller games has a perfect commuting operator strategy. 

This result is easy to prove from a strategies point of view. Recall that a clause $\clauseh{i}$ corresponds to a question vector $(a_i,b_i,c_i)$ that could be sent to the players in a round of the game. If a game has a disconnected clause graph $\cG_{123}$, players will never be sent a question vector asking them to make measurements from different connected components of the graph. Thus, players can consider the measurements in each connected component of $\cG_{123}$ independently when coming up with a strategy for the game. If they come up with strategies that win for each connected component of clauses they can always combine them  (given a question, a player follows the strategy corresponding to the connected component that question came from) to create a strategy that wins on the larger game. 

Below, we prove the result using algebraic techniques. The proof is considerably less natural in this setting, but provides a useful exercise in proving results about XOR games using the groups formalism.

\begin{thm} \label{thm:connectedComponents}
Let $G$ be a 3XOR game with clause set $S$, clause group $H$, and clause graph $\cG_{123}$.
Then $\sigma \in H$ iff there exists a subset of clauses $S' \subseteq S$ corresponding to all the edges in a connected component of $\cG_{123}$ with $\sigma \in \langle S' \rangle$. 
\end{thm}

\begin{proof}
First note that if the clause graph $\cG_{123}$ is connected \Cref{thm:connectedComponents} is trivial, since the only subset of $S$ corresponding to a connected component of $\cG_{123}$ is $S$ itself. Also note that one direction of the above claim is immediate by the observation that $\langle S' \rangle < \langle S \rangle$ and so $\sigma \in \langle S' \rangle \Longrightarrow \sigma \in \langle S \rangle = H$. 

To deal with the converse direction, consider a game $G$ with clause group $H \ni \sigma$ and a disconnected clause graph $\cG_{123}$. Let $S_1, S_2, ... , S_l$ be subsets of $S$ corresponding to all the edges in the connected components of the clause graph;  note that sets $S_1, ..., S_l$ partition the $S$. For all $i \in [l]$, define a map $\rho_i$ which acts on the generators of $H$ as\footnote{Somewhat surprisingly, we cannot extend this map to a homomorphism on $H$ (because it's action on $\sigma$ may be undefined).}
\begin{align}
    \rho_i(h_j) = 
    \begin{cases} 
    h_j &\text{ if } h_j \in S_i \\
    \id &\text{otherwise}
    \end{cases}
\end{align}

We have by assumption that $\sigma \in H$. Then there exists a sequence of clauses $h_{r_1}h_{r_2}...h_{r_t} = \sigma$. We prove two claims:
\begin{enumerate}
    \item For all $\alpha \in \{1,2,3\}$, $i \in [l]$ we have : $\proj{\alpha}(\rho_i(h_{r_1})\rho_i(h_{r_2})...\rho_i(h_{r_t})) = \id$. 
    \item For some $i' \in [l]$, we have $\rho_{i'}(h_{r_1})\rho_{i'}(h_{r_2})...\rho_{i'}(h_{r_t}) = \sigma$.
\end{enumerate}
To prove the first, define the set $V_i$ to consist of all generators $x_j^{(\alpha)}$ corresponding to vertices in the connected component of $\cG_{123}$ containing clauses $S_i$. Then, for all $\alpha \in \{1,2,3\}$, define $V_i^{(\alpha)} = V_i \cap G_{\alpha}$ to be the subset of generators in $V_i$ corresponding to player $\alpha$. Finally, we define the homomorphism $\pi_i: G \rightarrow G$ by its action on the generators of $G$:
\begin{align}
    \pi_i(x_j^{(\alpha)}) = \begin{cases} 
    x_j^{(\alpha)} &\text{ if } x_j^{(\alpha)} \in V_i \\
    \id &\text{ otherwise}
    \end{cases} \hspace{26pt} \text{and} \hspace{26pt} \pi_i(\sigma) = \id.
\end{align}
Routine calculation shows that $\pi_i$ preserves the relations of $G$, and thus, is a valid homomorphism. Now, to prove Claim 1 we show 
\begin{align}
    \proj{\alpha}(\rho_i(h_{r_1})\rho_i(h_{r_2})...\rho_i(h_{r_t})) = \proj{\alpha}(\pi_i(h_{r_1}h_{r_2}...h_{r_t})) = \proj{\alpha}(\pi_i(\sigma)) = \id.
\end{align}
The second equality follows because we assumed $h_{r_1}h_{r_2}...h_{r_t} = \sigma$, and the third equality holds by definition of $\proj{\alpha}$. All that remains to show is the first, but this is straightforward since 
\begin{align}
    \proj{1}(\rho_i(h_{r_j}))  = \proj{1}(\pi_i(h_{r_j})) = x_{a_{r_j}}^{(1)}
\end{align}
if $h_{r_j} \in S_i$ and 
\begin{align}
    \proj{1}(\rho_i(h_{r_j}))  = \proj{1}(\pi_i(h_{r_j})) = \id
\end{align}
otherwise, since $h_{r_j} \notin S_i \implies \proj{1}(h_{r_j}) \notin V_i$ by definition of $V_i$.

Now, to prove the second claim, note Claim 1 in combination with \Cref{eq:playerDecomp2} gives 
\begin{align}
    \rho_i(h_{r_1})\rho_i(h_{r_2})...\rho_i(h_{r_t}) &=  \proj{\sigma}(\rho_i(h_{r_1})\rho_i(h_{r_2})...\rho_i(h_{r_t})) \prod_{\alpha \in [3]} \proj{\alpha}(\rho_i(h_{r_1})\rho_i(h_{r_2})...\rho_i(h_{r_t})) \\
    &= \proj{\sigma}(\rho_i(h_{r_1})\rho_i(h_{r_2})...\rho_i(h_{r_t})).
\end{align}
If $\proj{\sigma}(\rho_i(h_{r_1})\rho_i(h_{r_2})...\rho_i(h_{r_t})) = \sigma$ for any $i \in [l]$ the above equation proves Claim 2. Assume for contradiction that $\proj{\sigma}(\rho_i(h_{r_1})\rho_i(h_{r_2})...\rho_i(h_{r_t})) = \id$ for all $i \in [l]$. Then we have
\begin{align}
    \proj{\sigma}(h_{r_1}h_{r_2}...h_{r_t}) &= \proj{\sigma}(h_{r_1})\proj{\sigma}(h_{r_2})...\proj{\sigma}(h_{r_t})\\
    &= \proj{\sigma}\left(\prod_{i \in [l]}\rho_i(h_{r_1})\right)\proj{\sigma} \left(\prod_{i \in [l]}\rho_i(h_{r_2})\right)... \proj{\sigma}\left(\prod_{i \in [l]}\left(h_{r_t}\right)\right) \\
    &= \prod_{i \in [l]} \proj{\sigma} \left( \rho_i( h_{r_1})\rho_i(h_{r_2})...\rho_i(h_{r_t}) \right) = 1
\end{align}
Where we used the fact that $\sigma$ commutes with all elements of $G$ to reorder elements and get from the second line to the third, and our assumption for the sake of contradiction on the final line. But, by our assumption at the start of this section we also have $\proj{\sigma}(h_{r_1}h_{r_2}...h_{r_t}) = \sigma$. The contradiction proves Claim 2. 

Finally, to complete the proof we note
\begin{align}
    \rho_{i'}(h_{r_1})\rho_{i'}(h_{r_2})...\rho_{i'}(h_{r_t}) &=  \proj{\sigma}(\rho_{i'}(h_{r_1})\rho_{i'}(h_{r_2})...\rho_{i'}(h_{r_t})) \\
    &= \sigma
\end{align}
by \Cref{eq:playerDecomp2}, Claim 1, and Claim 2, and $\rho_{i'}(h_{r_1})\rho_{i'}(h_{r_2})...\rho_{i'}(h_{r_t}) \in S_{i'}$ by definition of $\rho_{i'}$. Thus the claim holds with $S' = S_{i'}$.
\end{proof}

To prove the strongest form of \Cref{thm:K_modding}, we also need a version of \Cref{thm:connectedComponents} that applies to  words $\sigma \in H^E \pmod{K}$. We give that theorem next. The proof is very similar to the proof of \Cref{thm:connectedComponents}, with a few more technical details.\footnote{Actually, \cref{thm:connectedCompononetsEven} in combination with \Cref{thm:K_modding} provide an alternate proof of \Cref{thm:connectedComponents}. Here we proved \Cref{thm:connectedComponents} directly both because the proof serves as a good warm up to the proof of \Cref{thm:connectedCompononetsEven}, and to enphasize the result can be proved independtly from \Cref{thm:K_modding}.}

\begin{thm} \label{thm:connectedCompononetsEven}
Let $G$ be a 3XOR game with clause set $S$, clause group $H$, and clause graph $\cG_{123}$.
For any subset of clauses $S' \subseteq S$, define $H_{S'} = \angles{S'}$ to be the clause group generated by just the clauses in $S'$, and define $H_{S'}^E$ analogously. Then $\sigma \in H^{E} \pmod{K}$ iff there exists a subset of clauses $S' \subseteq S$ corresponding to all the edges in a connected component of $\cG_{123}$ with $\sigma \in H_{S'}^E \pmod{K}$. 
\end{thm}  

\begin{proof}
As with the proof of \Cref{thm:connectedComponents}, the case where $\cG_{123}$ is connected and the direction $\sigma \in H_{S'}^E \pmod{K} \Rightarrow \sigma \in H^{E} \pmod{K}$ are immediate. 

To deal with the remaining case, let $G$ be an XOR game with disconnected clause graph $\cG_{123}$ and $\sigma \in H^E \pmod{K}$. Let $S_1, S_2, ..., S_l$ be subsets of $S$ corresponding to all edges in the connected components of the clause graph. For each $S_i$, we pick some representative clause $\hat{h}_i \in S_i$. Then, define a map $\tilde{\rho}_i$ which acts on the generators of $H$ as 
\begin{align}
    \tilde{\rho}_i(h_j) = 
    \begin{cases}
    h_j &\text{ if } h_j \in S_i \\
    \hat{h}_i &\text{ otherwise.}
    \end{cases}
\end{align}

Note that for any generator of $h_j h_{j'}$ of the even clause group $H^E$ we have 
\begin{align}
    \tilde{\rho}_i(h_j) \tilde{\rho}_i(h_{j'}) \in H_{S_i}^E.
\end{align}
As in the proof of \Cref{thm:connectedComponents}, define the subset of generators $V_i$ to be the $x_i^{(\alpha)}$ corresponding to vertices in the same connected component as the edges in $S_i$. Then define the projector $\tilde{\pi}_i$ which acts on the generators of $G$ as 
\begin{align} \label{eq:typo2}
    \tilde{\pi}_i(x_j^{(\alpha)}) = \begin{cases} 
    x_j^{(\alpha)} & \text{ if } x_j^{(\alpha)} \in V_i \\
    \proj{\alpha}(\hat{h}_i) & \text{ otherwise}
    \end{cases} \hspace{26pt} \text{and} \hspace{26pt} \tilde{\pi}_i(\sigma) = \id.
\end{align}
An important observation is that $\tilde{\pi}_i$ maps commutators of even pairs of generators to commutators of even pairs of generators (or the identity) so $\tilde{\pi}_i(K) \subseteq K$. 

By assumption we have $\sigma \in H^{E} \pmod{K}$. Then there exists an even length sequence of clauses $h_{r_1}h_{r_2}...h_{r_t} = w$ with $w = \sigma w_k$ and $w_k \in K$. We claim: 
\begin{enumerate}
    \item For all $i \in [l]$, $\alpha \in \{1,2,3\}$ we have : 
    $\proj{\alpha}(\tilde{\rho}_i(h_{r_1})\tilde{\rho}_i(h_{r_2})...\tilde{\rho}_i(h_{r_t})) = \proj{\alpha}(\tilde{\pi}_i(w_k)) \in K.$
    \item There exists an $i' \in [l]$ satisfying $\proj{\sigma}(\tilde{\rho}_{i'}(h_{r_1})\tilde{\rho}_i(h_{r_2})...\tilde{\rho}_i(h_{r_t})) = \sigma.$
\end{enumerate}
The proof of the first equality in Claim 1 follows identically to the proof of Claim 1 in \Cref{thm:connectedComponents}. The second inequality holds because $\tilde{\pi}_i(K) \subseteq K$. 

Proving Claim 2 requires a little more work. The complicating issue is that we can encounter a case where $\proj{\sigma}(\tilde{\rho}_i(h_j)) = \sigma$ even if $h_j \notin S_i$. Thus the equation 
\begin{align} 
\proj{\sigma}(h_j) = \proj{\sigma}\left(\prod_i \tilde{\rho}_i(h_j)\right)
\end{align}
might not hold, and we can't simply copy the proof of Claim 2 in \Cref{thm:connectedComponents}. However, copying the proof of Claim 2 does give us that there exists an $i' \in [l]$ for which $\proj{\sigma}({\rho}_{i'}(h_{r_1}){\rho}_{i'}(h_{r_2})...{\rho}_{i'}(h_{r_t})) = \sigma$, that is, the claim holds when the map $\tilde{\rho}_i$ is replaced by the map $\rho_i$ defined in the proof of \Cref{thm:connectedComponents}. Let $n_{i'}$ be the number of clauses in the sequence $h_{r_1}h_{r_2}...h_{r_t}$ not contained in $S_{i'}$, that is 
\begin{align}
    n_{i'} = \abs{ \{j \in [l] : r_{r_j} \notin S_{i'} \} }. 
\end{align}
We claim $n_{i'}$ is even. To see this, note that any word $w \in K$ contains each generator $x_i^{(\alpha)}$ an even number of times, since the even commutators contain the generators $x_i^{(\alpha)}$ an even number of times, and the $x_i^{(\alpha)}$ are self-inverse. Then the number of occurrences of all the $x_i^{(1)} \notin V_{i'}$ in the word $h_{r_1}h_{r_2}...h_{r_t}$ must be even (the $1$ here is arbitrary, all that matters is that we fix a player). But this is equal to $n_{i'}$ mod 2, and we conclude $n_{i'}$ is even.  Finally, we note that 
\begin{align}
    \proj{\sigma}(\tilde{\rho}_{i'}(h_{r_1})\tilde{\rho}_{i'}(h_{r_2})...\tilde{\rho}_{i'}(h_{r_t})) 
    &= \proj{\sigma}({\rho}_{i'}(h_{r_1}){\rho}_{i'}(h_{r_2})...{\rho}_{i'}(h_{r_t})) (\proj{\sigma}(\hat{h}_{i'}))^{n_{i'}} \label{eq:typo3} \\
    &= \proj{\sigma}({\rho}_{i'}(h_{r_1}){\rho}_{i'}(h_{r_2})...{\rho}_{i'}(h_{r_t})) = \sigma, \label{eq:typo4}
\end{align}
since $n_{i'}$ is even and $\sigma$ has order two. 

Combining Claims 1 and 2 with \Cref{eq:playerDecomp2} gives 
\begin{align}
    \tilde{\rho}_{i'}(h_{r_1})\tilde{\rho}_{i'}(h_{r_2})...\tilde{\rho}_{i'}(h_{r_t}) &=  \proj{\sigma}\left(\tilde{\rho}_{i'}(h_{r_1})\tilde{\rho}_{i'}(h_{r_2})...\tilde{\rho}_{i'}(h_{r_t}) \right) \prod_{\alpha}\proj{\alpha}\left(\tilde{\rho}_{i'}(h_{r_1})\tilde{\rho}_{i'}(h_{r_2})...\tilde{\rho}_{i'}(h_{r_t})\right)\\
    &= \sigma \prod_{\alpha}\proj{\alpha}\left(\tilde{\rho}_{i'}(h_{r_1})\tilde{\rho}_{i'}(h_{r_2})...\tilde{\rho}_{i'}(h_{r_t})\right) 
\end{align}
with $\tilde{\rho}_{i'}(h_{r_1})\tilde{\rho}_{i'}(h_{r_2})...\tilde{\rho}_{i'}(h_{r_t}) \in H_{S_{i'}}^{E}$ and $\prod_{\alpha}\proj{\alpha}\left(\tilde{\rho}_{i'}(h_{r_1})\tilde{\rho}_{i'}(h_{r_2})...\tilde{\rho}_{i'}(h_{r_t})\right) \in K$.

\end{proof}

To close this section we observe that \Cref{thm:connectedCompononetsEven} implies that we can prove \Cref{thm:K_modding} for all 3XOR games by proving it in the special case of games whose clause graph $\cG_{123}$ is connected.
To see why, consider a 3XOR game $G$ with clause set $S$, a disconnected clause graph and $\sigma \in H^E \pmod{K}$. \Cref{thm:connectedCompononetsEven} says that we can find a connected subset of clauses $S' \subset S$ with $\sigma \in H_{S'}^E  \pmod{K}$. Then, we restrict to the 3XOR game $G'$ defined only on these clauses and note is has a fully connected clause graph. \Cref{thm:K_modding} then says $\sigma \in \langle S' \rangle$, which implies $\sigma \in \langle H \rangle$ for the original game $G$ as well. For this reason, we assume the clause graph $\cG_{123}$ is connected in \Cref{sec:K_moddingProof}.

\subsection{Proof of \texorpdfstring{\Cref{thm:K_modding}}{Theorem 2.6}}
\label{sec:K_moddingProof}

The proof is involved, and we will build up to it slowly over the course of many lemmas. First, we recap the theorem and give an outline of the first stages of the proof. Note that notation, particularly the $w, w'$ and $\tilde{w}$, in this outline is simplified, and does not match the notation used in the remainder of this section. 

\begin{repthm}{thm:K_modding}[Repeated]
Let $\sigma, H^E, K$ be defined relative to an $k$XOR game as described in \Cref{subsec:groups} and define $[\sigma]_K,  [H^E]_K$ as in \Cref{subsec:groups}. Then 
\begin{align}
    [\sigma]_K \in [H^E]_K \quad \Longleftrightarrow \quad \sigma \in H^E.
\end{align} 
\end{repthm}

\begin{proof}[Proof Outline (Part 1) of \Cref{thm:K_modding}]
The forwards direction is immediate from the discussion in \Cref{subsubsec:Sufficient_conditions}. The backwards direction takes work. 

Our starting point is the observation that $[\sigma]_K \in [H^E]_K$ iff there exists some $h \in H^E$ satisfying 
$h = \sigma w$, with $w \in K$. Our goal, given such an $h$ is to show that $\sigma \in H^E$. To do this we modify the word $h$ by right multiplying by words in $H^E$ until we have removed the $w$ portion, producing a word $\sigma \in H^E$. We refer to this process as ``clearing" the word $w$ from the word $h$. To begin, we break $w$ into three words: since $G_1, G_2$ and $G_3$ group elements all commute with each other we can separate them out and write $w = w_1w_2w_3$ with each $w_\alpha \in G_\alpha^E \cap K$. Then we clear the word $w$ one $w_\alpha$ at a time.

In \Cref{subsec:ProjAndRI} we show how to clear the $w_1$ part of the word $w$. 
To do this we define a homomorphism $\cmap{1}$ which maps any word $v_1 \in G_1$ to a word in $h \in H$ with the $G_1$ portion of the word $h$ equal to $v_1$.
Applying this homomorphism to $w_1$ produces a word $\cmap{1}(w_1) =  w_1\widetilde{w}_2\widetilde{w}_3 \in H^E$, where words $\widetilde{w}_2 \in G_2^E \cap K$ and $\widetilde{w}_3 \in G_3^E \cap K$ are arbitrary.  Now $h \cmap{1}(w_1)^{-1}$ is a word of the form $w' \sigma = w_2' w_3' \sigma$ with $w_2' \in G_2^E \cap K$ and $w_3' \in G_3^E \cap K$. Importantly $w'$ contains no terms in the $G_1$ subgroup, that is, we have successfully cleared the $G_1$ portion of the word $w$.

Our next step is to right multiply by a word which will clear the $w_2'$ term, while not introducing any new terms in the $G_1$ subgroup. We do this by constructing another homomorphism $\cmap{2,1}$, which takes a word $v_2$ in  $G_2^E$ and produces a word in $H^E$ which equals $v_2$ in the $G_2$ subgroup and projects to the identity in the $G_1$ subgroup whenever possible. Details are given in \Cref{subsec:IPRightInverse}.

\Cref{subesc:PreprocessingModK} performs the process of removing the $w_1$ and $w_2$ words from $h$. The final result is a word 
\begin{align}
    w'' = \cmap{2,1}\left( w_2' \right)^{-1} w' = w_3'' \sigma \in H^E,
\end{align}
where $w_3'' \in G_3^E \cap K$.

Finally, we want to clear the word $w_3''$ without introducing any words in the $G_1$ or $G_2$ subgroups. Unlike previous sections, we do not do this by constructing a homomorphism. Instead, in \Cref{subsec:WordPRocessGadget,subsec:proof_of_gadgets} we construct a series of gadget words designed to make a word easier to clear. Then, in \Cref{susec:finalProof} we right multiply the word $w_3''$ by the gadget words, and clear the word with the gadgets introduced. This procedure is elaborated on in Part 2 of this proof outline, in \Cref{subsec:WordPRocessGadget}.
\end{proof}

We now begin the proof in earnest.

\subsubsection{Projectors and a simple right inverse.} \label{subsec:ProjAndRI}

We start with some useful notation. Recall the projector  $\varphi_\alpha : G \rightarrow G_\alpha$ onto group elements corresponding to player $\alpha$ defined in \Cref{subsec: new def} . It is a homomorphism, defined by

\begin{align}
    \varphi_\alpha(x_i^{(\beta)}) := 
    \begin{cases} 
        x_i^{(\beta)} &\text{ if } \alpha = \beta \\
        \id &\text{ otherwise}
    \end{cases}
\end{align}
\nx{ $ \varphi_\alpha(x_i^{(\beta)})$}
and
\begin{align}
    \proj{\alpha}(\sigma) = \id.
\end{align}
We also defined a projector onto the $\sigma$ subgroup, $\proj{\sigma} : G \rightarrow \{\sigma, 1\}$ which satisfies 
\begin{align}
    \proj{\sigma}(x_i^{(j)}) = \id \text{ and } \proj{\sigma}(\sigma) = \sigma. 
\end{align}
Because the map $\proj{}$ is many to one, there are many choices of right inverse: in the course of the paper we will define several. We use the notation $\cmap{}$, with various subscripts, when referring to right inverses of $\proj{}$.

We first define the simple right inverse $\cmap{\alpha} : G_\alpha \rightarrow H$ 
\nx{ $\cmap{\alpha} : G_\alpha \rightarrow H$}
which maps each $x_i^{(\alpha)}$ to a single clause in $S$. For ease of notation, we give the definition when $\alpha = 1$. $\cmap{1}$ is a homomorphism which acts on the generators of $G_1$ by
\begin{align}
    \cmap{1}(x_i^{(1)}) = h_j 
\end{align}
where $j \in [m]$ is chosen so that $\proj{1}(h_j) = x_i^{(1)}$. Note that some clause $x_i^{(1)}x_j^{(2)}x_k^{(3)}\sigma^l$ must exist in $S$ or else the question $x_i^{(1)}$ is never asked, and the group element $x_i^{(1)}$ can be removed from the game group (this can be viewed as a special case of the proof given in \Cref{sec: Connectivity} that we can assume the game group is connected). If there are multiple clauses which contain the element $x_i^{(1)}$, we pick one arbitrarily. To verify $\cmap{1}$ is indeed a homomorphism, we can check
\begin{align}
    \cmap{1}(x_i^{(1)})^2 &= h_j^2 \\
    &= x_{a_j}^{(1)}x_{b_j}^{(2)}x_{c_j}^{(3)}\sigma^{s_j}x_{a_j}^{(1)}x_{b_j}^{(2)}x_{c_j}^{(3)}\sigma^{s_j} \\
    &= \left(x_{a_j}^{(1)}\right)^2\left(x_{b_j}^{(2)}\right)^2\left(x_{c_j}^{(3)}\right)^2\left(\sigma^{s_j}\right)^2 = \id.
\end{align}
$\cmap{\alpha}$ for general $\alpha$ is defined similarly. 

\subsubsection{Identity preserving right inverse.} \label{subsec:IPRightInverse}

The next right inverse we define, $\cmap{\alpha,\beta}$, acts as a right inverse to $\proj{\alpha}$ while also producing a word $h\in H$ satisfying $\proj{\beta}(h) = \id$ whenever such a mapping is possible. In order to define $\cmap{\alpha,\beta}$ as a homomorphism, we restrict it's action to the subgroup of even length words $G_{\alpha}^{E}$. 

Now we give a ``trick'' we will use repeatedly to construct homomorphisms on the even subgroups. 

\begin{lem} \label{lem:homFromGens}
Let $f : G_\alpha \rightarrow H$ be an arbitrary map. Define $\tilde{f} : G_\alpha^{E} \rightarrow H^{E}$ by its action on the generators of $G_\alpha^{E}$
\begin{align}
    \tilde{f}(x_i^{(\alpha)} x_j^{(\alpha)}) = f(x_i^{(\alpha)}) f( x_j^{(\alpha)})^{-1},
\end{align}
and extend it to act on elements in $G_\alpha^{E}$ in the natural way, so for any word 
\begin{align}
   w_\alpha = \prod_{l}x_{i_l}^{(\alpha)} x_{j_l}^{(\alpha)} \in G_{\alpha}^E. 
\end{align}
we have
\begin{align}
    \tilde{f}\left(\prod_{l}x_{i_l}^{(\alpha)} x_{j_l}^{(\alpha)} \right) = \prod_l \tilde{f}\left(x_{i_l}^{(\alpha)} x_{j_l}^{(\alpha)}\right) 
\end{align}

Then $\tilde{f}$ is a homomorphism.
\end{lem} 

\begin{proof} By \Cref{thm:alg_com_hom} the only thing we need to show is that $\tilde{f}$ respects the relations of $G_\alpha^E$. By \Cref{claim:G_alpha^E_pres} $G_\alpha^E$ has only two families of relations, namely that 
\begin{enumerate}[(1)]
    \item $x_i^{(\alpha)} x_j^{(\alpha)} x_j^{(\alpha)} x_k^{(\alpha)} = x_i^{(\alpha)}x_k^{(\alpha)}$ for all $i,j,k \in [N]$, and that \label{item:G^E_reln_1}
    \item $x_i^{(\alpha)} x_j^{(\alpha)} x_j^{(\alpha)} x_i^{(\alpha)} = 1$ for all $i,j \in [N]$. \label{item:G^E_reln_2} 
\end{enumerate}
We check that $\tilde{f}$ satisfies these through straightforward computation. Noting that 
\begin{align}
    \tilde{f}(x_i^{(\alpha)} x_j^{(\alpha)}x_j^{(\alpha)}x_k^{(\alpha)}) &= f(x_i^{(\alpha)})f( x_j^{(\alpha)})^{-1} f(x_j^{(\alpha)})f(x_k^{(\alpha)})^{-1} = \tilde{f}(x_i^{(\alpha)}x_k^{(\alpha)})
\end{align}
shows $\tilde{f}$ satisfies relation \labelcref{item:G^E_reln_1}, while noting that 
\begin{align}
    \tilde{f}(x_i^{(\alpha)} x_j^{(\alpha)}x_j^{(\alpha)}x_i^{(\alpha)}) &= f(x_i^{(\alpha)})f( x_j^{(\alpha)})^{-1} f(x_j^{(\alpha)})f(x_i^{(\alpha)})^{-1} = 1
\end{align}
shows $\tilde{f}$ satisfies relations \labelcref{item:G^E_reln_2}. 
\end{proof}

Now we turn to introducing an important homomorphism
$\cmap{a,b}$. Our organization
is unusual in that we give its properties 
first as \Cref{lem:making_homomorphisms}
and then  define it
and its key ingredients, \Cref{eq:rep_def,eq:cpath_def,eq:cmap_def}, 
during the proof of the lemma.
 We alert the reader that these  objects  will be re-used in future proofs.

\nx{ $\cmap{\alpha,\beta} : G_\alpha^{E} \rightarrow H^{E}$ satisfying \\
   A1. \  $\proj{\alpha}(\cmap{\alpha,\beta}(w)) = w$ for all $w \in G_\alpha^{E}$. \\
  A2. \  $\proj{\beta}\left( \cmap{\alpha,\beta}(w) \right) = \id$ whenever there exists an $h\in H^E$ satisfying $\proj{\alpha}(h) = w$ and $\proj{\beta}(h) = \id$.  
}
\begin{lem} \label{lem:making_homomorphisms}
For each $\alpha, \beta \in [3]$, with $\alpha \neq \beta$ there exists a homomorphism $\cmap{\alpha,\beta} : G_\alpha^{E} \rightarrow H^{E}$ satisfying 
\begin{enumerate}[label=A\arabic*., ref=A\arabic*]
    \item $\proj{\alpha}(\cmap{\alpha,\beta}(w)) = w$ for all $w \in G_\alpha^{E}$. \label{abprop:afix}
    \item $\proj{\beta}\left( \cmap{\alpha,\beta}(w) \right) = \id$ whenever there exists an $h\in H^E$ satisfying $\proj{\alpha}(h) = w$ and $\proj{\beta}(h) = \id$.  \label{abprop:bid}
\end{enumerate}
An important consequence of Property~\labelcref{abprop:bid} is that $\proj{\beta}\left( \cmap{\alpha,\beta} \left( \proj{\alpha} (h) \right) \right) = \id$ for any $h \in H^E$ satisfying $\proj{\beta}(h) = \id$.
\end{lem}
\begin{proof}
For ease of notation, we prove the result when $\alpha = 1$, $\beta = 2$. The proof is identical for other $\alpha, \beta$. 

Recall the (multi)graph $\cG_{12}$, defined in \Cref{subsec: new def}. $\cG_{12}$ has $2N$ vertices, labeled by the group elements $x_1^{(1)},x_2^{(1)}, ... , x_N^{(1)}$, $x_1^{(2)},x_2^{(2)}, ... , x_N^{(2)}$. We identify vertices in the graph with generators of game group $G$, and abuse notation slightly by referring to the two objects interchangeably. Edges in the graph correspond to clauses; the graph has one edge $(x_i^{(1)},x_j^{(2)})$ for every clause $x_i^{(1)}x_j^{(2)}x_k^{(3)}\sigma^{(l)}$ in $S$. ($k \in [N]$ and $l \in \{0,1\}$ are arbitrary.) Then $\cG_{12}$ is bipartite, with the vertices $x_i^{(1)}$ for $i \in [N]$ forming one half of the graph and $x_j^{(2)}$ for $j \in [N]$ forming the other. See \Cref{fig:sampleG12} for an example. Recall that, sequences of edges in $\cG_{12}$ (and in particular, paths) can be identified with words in $H$.
\begin{figure}
    \centering
    \resizebox{10cm}{!}{
        \begin{tikzpicture}

\vertex[xpos=0, ypos=0, label name = $x_1^{(1)}$]{11}
\vertex[xpos=1, ypos=0, label name = $x_2^{(1)}$]{12}
\vertex[xpos=2, ypos=0, label name = $x_3^{(1)}$]{13}
\vertex[xpos=3, ypos=0, label name = $x_4^{(1)}$]{14}
\vertex[xpos=4, ypos=0, label name = $x_5^{(1)}$]{15}
\vertex[xpos=5, ypos=0, label name = $x_6^{(1)}$]{16}

\vertex[xpos=0, ypos=2, label position = above, label name = $x_1^{(2)}$]{21}
\vertex[xpos=1, ypos=2, label position = above, label name = $x_2^{(2)}$]{22}
\vertex[xpos=2, ypos=2, label position = above, label name = $x_3^{(2)}$]{23}
\vertex[xpos=3, ypos=2, label position = above, label name = $x_4^{(2)}$]{24}
\vertex[xpos=4, ypos=2, label position = above, label name = $x_5^{(2)}$]{25}
\vertex[xpos=5, ypos=2, label position = above, label name = $x_6^{(2)}$]{26}

\edge{11}{21}
\edge{11}{22}
\edge{11}{23}
\edge{12}{22}
\edge{12}{23}

\edge{13}{24}
\edge{14}{24}

\edge{15}{25}
\edge{15}{26}
\edge{16}{25}
\edge{16}{26}

\end{tikzpicture}
        }
    \caption{Sample graph $\cG_{12}$ for a game with alphabet size $N = 6$ and $m= 11$ clauses. The middle component for example corresponds to clauses $x_3^{(1)}x_4^{(2)}x_{k_1}^{(3)}\sigma^{l_1}$ and $x_4^{(1)}x_4^{(2)}x_{k_2}^{(3)}\sigma^{l_2}$, where $k_1, k_2 \in [N]$ and $l_1, l_2 \in \{0,1\}$ are arbitrary.}
    \label{fig:sampleG12}
\end{figure}

Now, consider a word $P\left(x_{i_1}^{(1)}, x_{j_t}^{(2)}\right)$ corresponding to a path in $\cG_{12}$ from a vertex associated with player 1 to a vertex associated with player 2. Note the path has odd length because $\cG_{12}$ is bipartite, so the word $P\left(x_{i_1}^{(1)}, x_{j_t}^{(2)}\right)$ consists of an odd sequence of clauses. All generators in $G_1$, $G_2$ other than $x_{i_1}^{(1)}$ and $x_{j_t}^{(2)}$ are repeated adjacent to each other in the word $P\left(x_{i_1}^{(1)}, x_{j_t}^{(2)}\right)$. 
These generators cancel, and so
\begin{align}
    P\left(x_{i_1}^{(1)}, x_{j_t}^{(2)}\right) 
    &=   x_{i_1}^{(1)}x_{j_1}^{(2)}x_{k_1}^{(3)}\sigma^{l_1}
    x_{i_2}^{(1)}x_{j_1}^{(2)}x_{k_2}^{(3)}\sigma^{l_2}
    x_{i_2}^{(1)}x_{j_2}^{(2)}x_{k_3}^{(3)}\sigma^{l_3}    
    ...
    x_{i_t}^{(1)}x_{j_t}^{(2)}x_{k_{2t+1}}^{(3)}\sigma^{l_{2t+1}}\\
    &=  x_{i_1}^{(1)}x_{j_t}^{(2)}x_{k_1}^{(3)}x_{k_2}^{(3)} ... x_{k_{2t+1}}^{(3)} \sigma^{l_1 + l_2 + ... l_{2t+1}}.    
\end{align}
Hence,
\begin{align}
    \proj{1}\left(P(x_{i_1}^{(1)}, x_{j_t}^{(2)})\right) = x_{i_1}^{(1)} \label{eq:1fix}
\end{align}
and 
\begin{align}
    \proj{2}\left(P(x_{i_1}^{(1)}, x_{j_t}^{(2)})\right) = x_{j_t}^{(2)}. \label{eq:2fix} 
\end{align}
We note that we can construct a path with the above properties between any two vertices $x_{i_1}^{(1)}, x_{j_t}^{(2)}$ in the same connected component of the multigraph $\cG_{12}$.

Next we develop some notation related to these connected components of $\cG_{12}$. Arbitrarily pick a pair of vertices $x^{(1)}_{j_1} \in G_1$ and $x^{(2)}_{j_2} \in G_2$ from each component. Call $x^{(1)}_{j_1}$ and $x^{(2)}_{j_2}$ representative vertices. 
\nx{ $
\repnew{1,2}{\alpha}{\beta} : G_\alpha \rightarrow  G_\beta    \;\;\;\;\; \text{for} \ \alpha,\beta \in \{1,2\}$ 
 }
Then define the maps 
\begin{align}
\repnew{1,2}{\alpha}{\beta} : G_\alpha \rightarrow  G_\beta    \;\;\;\;\; \text{for $\alpha,\beta \in \{1,2\}$ } \label{eq:rep_def}
\end{align} 
to take each generator of $G_\alpha$ (vertices in $\cG_{12}$) to the unique representative vertex in $G_\beta$ in the same component as that generator. Each function  $\repnew{1,2}{\alpha}{\beta}$ maps generators which square to the identity to generators which square to the identity, so can be extended to a homomorphism acting on words in $G_\alpha$. Note that the homomorphism $\repnew{1,2}{\alpha}{\beta}$ constructed in this way necessarily satisfies $\repnew{1,2}{\alpha}{\beta}(1) = 1$ for any $\alpha$, $\beta$.

\begin{figure}[ht]
    \centering
    \resizebox{10cm}{!}{
        \begin{tikzpicture}

\vertex[xpos=0, ypos=0, label name = $x_1^{(1)}$]{11}
\vertex[xpos=1, ypos=0, label name = $x_2^{(1)}$]{12}
\vertex[xpos=2, ypos=0, label name = $x_3^{(1)}$]{13}
\vertex[xpos=3, ypos=0, label name = $x_4^{(1)}$]{14}
\vertex[xpos=4, ypos=0, label name = $x_5^{(1)}$]{15}
\vertex[xpos=5, ypos=0, label name = $x_6^{(1)}$]{16}

\vertex[xpos=0, ypos=2, label position = above, label name = $x_1^{(2)}$, label color = red, color = red]{21}
\vertex[xpos=1, ypos=2, label position = above, label name = $x_2^{(2)}$]{22}
\vertex[xpos=2, ypos=2, label position = above, label name = $x_3^{(2)}$]{23}
\vertex[xpos=3, ypos=2, label position = above, label name = $x_4^{(2)}$, label color = red, color = red]{24}
\vertex[xpos=4, ypos=2, label position = above, label name = $x_5^{(2)}$, label color = red, color = red]{25}
\vertex[xpos=5, ypos=2, label position = above, label name = $x_6^{(2)}$]{26}

\edge{11}{21}
\edge{11}{22}
\edge{11}{23}
\edge{12}{22}
\edge{12}{23}

\edge{13}{24}
\edge{14}{24}

\edge{15}{25}
\edge{15}{26}
\edge{16}{25}
\edge{16}{26}

\end{tikzpicture}
        }
    \captionsetup{singlelinecheck=off}
    \caption[.]{Sample graph repeated from \Cref{fig:sampleG12} with a choice of representative vertices in $G_2$ indicated in red. As an example of our notation, consider the first connected component and note that
    \begin{align*}
    \repnew{1,2}{1}{2}(x_2^{(1)}) = \repnew{1,2}{1}{2}(x_1^{(1)}) = x_1^{(2)}
    \end{align*}
    and that 
    \begin{align*}
        \repnew{1,2}{2}{2}(x_1^{(2)}) = \repnew{1,2}{2}{2}(x_2^{(2)}) = \repnew{1,2}{2}{2}(x_3^{(2)}) = x_1^{(2)}.
    \end{align*}
    } 
    \label{fig:sampleG12repvertex}
\end{figure}

Next for each $x_i^{(1)} \in G_1$ fix a path, denoted 
\begin{align}
    \cpathnew{1,2}{{x^{(1)}_i}}{\repnew{1,2}{1}{2}\left({x^{(1)}_i}\right)}, \label{eq:cpath_def} 
\end{align}
between the vertex $x_i^{(1)}$ and the (connected) representative vertex (see \Cref{fig:sampleG12path}).\footnote{We emphasize that the path $\cpathnew{1,2}{{x^{(1)}_i}}{\repnew{1,2}{1}{2}\left({x^{(1)}_i}\right)}$ can be chosen arbitrarily.} Define the homomorphism $\cmap{1,2} : G_1^E \rightarrow H^E$ by its action on the generators of $G^E$,
\begin{align}
    \cmap{1,2} \left(x_i^{(1)}x_j^{(1)}\right) := \cpathnew{1,2}{{x^{(1)}_i}}{\repnew{1,2}{1}{2}\left({x^{(1)}_i}\right)}\cpathnew{1,2}{{x^{(1)}_j}}{\repnew{1,2}{1}{2}\left({x^{(1)}_j}\right)}^{-1}. \label{eq:cmap_def}
\end{align}
Recall the conflation of notation defined above, so $\cpathnew{1,2}{{x^{(1)}_i}}{\repnew{1,2}{1}{2}\left(x^{(1)}_i\right)}$ defines both a path in the graph $\cG_{12}$ and a word in $H$. The function $\cmap{1,2}$ is a valid homomorphism by \Cref{lem:homFromGens}. 

\begin{figure}[ht]
    \centering
    \resizebox{10cm}{!}{
        \begin{tikzpicture}

\vertex[xpos=0, ypos=0, label name = $x_1^{(1)}$]{11}
\vertex[xpos=1, ypos=0, label name = $x_2^{(1)}$, label color = black, color = blue]{12}
\vertex[xpos=2, ypos=0, label name = $x_3^{(1)}$]{13}
\vertex[xpos=3, ypos=0, label name = $x_4^{(1)}$]{14}
\vertex[xpos=4, ypos=0, label name = $x_5^{(1)}$]{15}
\vertex[xpos=5, ypos=0, label name = $x_6^{(1)}$]{16}

\vertex[xpos=0, ypos=2, label position = above, label name = $x_1^{(2)}$, label color = red, color = red]{21}
\vertex[xpos=1, ypos=2, label position = above, label name = $x_2^{(2)}$]{22}
\vertex[xpos=2, ypos=2, label position = above, label name = $x_3^{(2)}$]{23}
\vertex[xpos=3, ypos=2, label position = above, label name = $x_4^{(2)}$, label color = red, color = red]{24}
\vertex[xpos=4, ypos=2, label position = above, label name = $x_5^{(2)}$, label color = red, color = red]{25}
\vertex[xpos=5, ypos=2, label position = above, label name = $x_6^{(2)}$]{26}

\edge[color = blue]{11}{21}
\edge[color = blue]{11}{22}
\edge{11}{23}
\edge[color = blue]{12}{22}
\edge{12}{23}

\edge{13}{24}
\edge{14}{24}

\edge{15}{25}
\edge{15}{26}
\edge{16}{25}
\edge{16}{26}

\end{tikzpicture}
        }
    \captionsetup{singlelinecheck=off}
    \caption[.]{
    Sample graph with representative vertices indicated in red and the path $\cpathnew{1,2}{{x^{(1)}_2}}{\repnew{1,2}{1}{2}\left(x^{(1)}_2\right)} = \cpathnew{1,2}{{x^{(1)}_2}}{x_1^{(2)}}$ indicated in blue. This path corresponds to a word 
    \begin{align*}
        &\left(x_2^{(1)}x_2^{(2)}x_{k_1}^{(3)}\sigma^{l_1}\right)\left(x_1^{(1)}x_2^{(2)}x_{k_2}^{(3)}\sigma^{l_2}\right)\left(x_1^{(1)}x_1^{(2)}x_{k_3}^{(3)}\sigma^{l_3}\right)
        =
        x_2^{(1)}x_1^{(2)}x_{k_1}^{(3)}x_{k_2}^{(3)}x_{k_3}^{(3)} \sigma^{l_1}\sigma^{l_2}\sigma^{l_3}
        , 
    \end{align*}
    where $k_1, k_2, k_3 \in [N]$ and $l_1, l_2, l_3 \in \{0,1\}$ are arbitrary. 
    }
    \label{fig:sampleG12path}
\end{figure}
It remains to show $\cmap{1,2}$ satisfies Properties \labelcref{abprop:bid,abprop:afix}. Property~\labelcref{abprop:afix} requires that 
$
    \proj{1}(\cmap{1,2}(w)) = w
$
for all $w \in G_1^{E}$. To prove this property we show $\cmap{1,2}$ acts as desired on the generators of $G_1^E$.
This follows from \Cref{eq:1fix}, which gives
\begin{align}
    \proj{1}\left(\cmap{1,2}(x_i^{(1)}x_j^{(1)})\right) =  x_i^{(1)}(x_j^{(1)})^{-1} = x_i^{(1)}x_j^{(1)}.
\end{align}

Property \labelcref{abprop:bid} requires that 
$
    \proj{2}\left( \cmap{1,2}(w) \right) = \id
$
whenever there exists an $h\in H^E$ satisfying $\proj{1}(h) = w$ and  $\proj{2}(h) = \id$. To show this we first show that
\begin{align}
    \proj{2}(\cmap{1,2}(\proj{1}(h))) = \repnew{1,2}{2}{2}(\proj{2}(h)) = \repnew{1,2}{1}{2}(\proj{1}(h)) \label{eq:proj2r12}
\end{align}
for any $h \in H^{E}$ (we only need the first equality to prove Property~\labelcref{abprop:bid}, but the second equality is an easy consequence and will be useful to us later). The equality can be verified by checking the action of the two maps on generators $h_ih_j$ of $H^{E}$:
\begin{align}
    \proj{2}(\cmap{1,2}(\proj{1}(h_i h_j))) &=  \proj{2}(\cmap{1,2}(\proj{1}(x_{a_i}^{(1)}x_{b_i}^{(2)}x_{c_i}^{(3)}x_{a_j}^{(1)}x_{b_j}^{(2)}x_{c_j}^{(3)} \sigma^{s_i + s_j} ))) \\
    &= \proj{2}(\cmap{1,2}(x_{a_i}^{(1)}x_{a_j}^{(1)} )) \\
    &= \proj{2}\left(\cpathnew{1,2}{{x^{(1)}_{a_i}}}{\repnew{1,2}{1}{2}(x^{(1)}_{a_i})}\cpathnew{1,2}{{x^{(1)}_{a_j}}}{\repnew{1,2}{1}{2}({x^{(1)}_{a_j}})}^{-1}\right) \label{eq:preprephi12}\\
    &= \repnew{1,2}{1}{2}\left(x_{a_i}^{(1)}\right)\repnew{1,2}{1}{2}\left(x_{a_j}^{(1)}\right) \label{eq:prepphi12}\\
    &= \repnew{1,2}{2}{2}\left(x_{b_i}^{(2)}\right)\repnew{1,2}{2}{2}\left(x_{b_j}^{(2)}\right) \label{eq:keyphi12} \\
    &= \repnew{1,2}{2}{2}(\proj{2}(h_i h_j))
\end{align}
Line~\labelcref{eq:preprephi12} follows by definition of $\cmap{1,2}$ while Line~\labelcref{eq:prepphi12} follows from \Cref{eq:2fix}. The key observation comes in line~\labelcref{eq:keyphi12}; because $a_i$ and $b_i$ are both in the clause $h_i$, they are in the same connected component in the graph $\cG_{12}$. Then they have the same representative vertex and 
\begin{align}
    \repnew{1,2}{1}{2}\left(x_{a_i}^{(1)}\right) = \repnew{1,2}{2}{2}\left(x_{b_i}^{(2)}\right). 
\end{align}
Line \labelcref{eq:keyphi12} follows. This argument proves both equalities in \Cref{eq:proj2r12}.

Now any $h \in H$ satisfying $\proj{2}(h) = \id$ must have even length, so $h \in H^{E}$ and we have 
\begin{align}
    \proj{2}(\cmap{1,2}(\proj{1}(h))) &= \repnew{1,2}{2}{2}(\proj{2}(h)) = \repnew{1,2}{2}{2}(\id) = \id.
\end{align}
Using the fact that $\repnew{1,2}{2}{2}$ is a homomorphism in the last two equalities. This proves Property \labelcref{abprop:bid}, and completes the proof. 
\end{proof}

The next lemma proves that right inverses $\cmap{\alpha}$ and $\cmap{\alpha,\beta}$ map within the $K$ subgroup. That is, they map words in $K \cap G_\alpha^E$ to words in $K \cap H^E$. 

\begin{lem} \label{lem:Kfixed}
Let $v \in K \cap G_\alpha^{E}$ be arbitrary. Then 
\begin{align}
    \cmap{\alpha}(v) \in K \cap H^E
\end{align}
and 
\begin{align}
   \cmap{\alpha,\beta}(v) \in K \cap H^E
\end{align}
for all $\beta \neq \alpha$. 
\end{lem}

\begin{proof}
For notational convenience we prove the result when $\alpha = 1$, $\beta = 2$. 

The proof is mechanical: any word $v \in K \cap G_1^{E}$ can be written
\begin{align}
    v = \prod_i u_i \com{x_{a_{i_1}}^{(1)}x_{a_{i_2}}^{(1)}}{x_{a_{i_3}}^{(1)}x_{a_{i_4}}^{(1)}} u_i^{-1}.
\end{align}
with $u_i \in G_1$ arbitrary. 
We pick labels $b_{i_1}, ... ,b_{i_4}, c_{i_1}, ... ,c_{i_4} \in [N]$ and  $s_{i_1}, ... ,s_{i_4} \in \{0,1\}$ so that 
\begin{align}
    \cmap{1}\left(x_{a_{i_j}}^{(1)}\right) = x_{a_{i_j}}^{(1)}x_{b_{i_j}}^{(2)}x_{c_{i_j}}^{(3)}\sigma^{s_{i_j}}
\end{align}
for all $x_{a_{i_j}}$. Then
\begin{align}
    \cmap{1}(v) &= \prod_i \cmap{1}\left(u_i\right) \com{\cmap{1}(x_{a_{i_1}}^{(1)})\cmap{1}(x_{a_{i_2}}^{(1)})}{\cmap{1}(x_{a_{i_3}}^{(1)})\cmap{1}(x_{a_{i_4}}^{(1)})} \cmap{1}\left(u_i\right)^{-1} \\
    &= \prod_i \cmap{1}\left(u_i\right) \com{x_{a_{i_1}}^{(1)}x_{a_{i_2}}^{(1)}}{x_{a_{i_3}}^{(1)}x_{a_{i_4}}^{(1)}} 
    \com{x_{b_{i_1}}^{(2)}x_{b_{i_2}}^{(2)}}{x_{b_{i_3}}^{(2)}x_{b_{i_4}}^{(2)}} 
   \com{x_{c_{i_1}}^{(3)}x_{c_{i_2}}^{(3)}}{x_{c_{i_3}}^{(3)}x_{c_{i_4}}^{(3)}} \label{eq:clauses_com_to_gen}
    \cmap{1}\left(u_i\right)^{-1}\in K.
\end{align}
noting that any factors of $\sigma$ cancel in the commutator.

A similar argument shows $\cmap{1,2}(v) \in K$. To start assume $ v \in K \cap G^E$ and write
\begin{align}
    \cmap{1,2}(v) &= \prod_i \cmap{1,2}\left(u_i\right) \com{\cmap{1,2}(x_{i_1}^{(1)}x_{i_2}^{(1)})}{\cmap{1,2}(x_{i_3}^{(1)}x_{i_4}^{(1)})} \cmap{1,2}\left(u_i\right)^{-1}  \\
    &=  \prod_i \cmap{1,2}\left(u_i\right)\left(\prod_{\alpha=1}^3 \com{\proj{\alpha}\left(\cmap{1,2}(x_{i_1}^{(1)}x_{i_2}^{(1)})\right)}{\proj{\alpha}\left(\cmap{1,2}(x_{i_3}^{(1)}x_{i_4}^{(1)})\right)} \right) \cmap{1,2}\left(u_i\right)^{-1} 
\end{align}
Then note the words $\proj{\alpha}\left(\cmap{1,2}(x_{i_1}^{(1)}x_{i_2}^{(1)})\right)$ and $\proj{\alpha}\left(\cmap{1,2}(x_{i_3}^{(1)}x_{i_4}^{(1)})\right)$ are in $G_\alpha^E$ for each $\alpha$. Repeatedly applying the commutator identities
\begin{align}
    \com{x}{yz} &= \com{x}{y}y^{-1}\com{x}{z}y \\
\text{ and } 
\com{xy}{z} &= y^{-1}\com{x}{z}y\com{y}{z}
\end{align}
shows those words are in $K$. The full argument is given in an appendix (\Cref{lem:even_com_in_K}).
\end{proof}
An important consequence of \Cref{lem:Kfixed} is the following corollary. 
\begin{cor} \label{cor:rightInvId}
Let $v \in K \cap G_\alpha^{E}$ be arbitrary and $\alpha \neq \beta$. Then 
\begin{align}
    \proj{\sigma}(\cmap{\alpha}(v)) = \proj{\sigma}(\cmap{\alpha,\beta}(v)) = \id.
\end{align}
\end{cor}
\begin{proof}
First observe that, because $\sigma$ is central and not contained in any generators of $K$, we have $\proj{\sigma}(k) = 1$ for all $k \in K$ (a detailed proof of this fact is given in \Cref{lem:noSigmaInK}). Then, by \Cref{lem:Kfixed} $\cmap{\alpha}(v) \in K$. Hence
\begin{align}
    \proj{\sigma}\left(\cmap{\alpha}\left(v \right) \right) = 1
\end{align}
The proof for $\proj{\sigma}\left(\cmap{\alpha, \beta}\left( v \right) \right)$ is identical. 
\end{proof}

\subsubsection{Clearing the \texorpdfstring{$G_1$}{G1} and \texorpdfstring{$G_2$}{G2} subgroups} \label{subesc:PreprocessingModK}
The next lemma makes critical use of right inverses $\cmap{\alpha}$ and $\cmap{\alpha,\beta}$. It should be be thought of as a ``pre-processing" step, that puts words in a convenient form to prove \Cref{thm:K_modding}. 

\begin{lem} \label{lem:pre-process}
If there exists a word $w \in H^E$ satisfying $w = \sigma \pmod{K}$, then there exists a word $w'$ in $H^E$ satisfying: 
\begin{enumerate}
    \item $w' = \sigma \pmod{K}$ \label{prepro:sigma}
    \item $\proj{1}(w') = \proj{2}(w') = \id$. \label{prepro:clear}
\end{enumerate}
\end{lem}
\begin{proof}

We construct $w'$ by right multiplying $w$ by $\cmap{1}\left(\proj{1}\left(w^{-1}\right)\right)$ to clear the $G_1$ subgroup elements, then multiplying by $\cmap{2,1}\left(\proj{2}\left(\left(w \cmap{1}\left(\proj{1}\left(w^{-1}\right)\right)\right)^{-1}\right)\right)$ to clear the $G_2$ subgroup. 
Formally:
\begin{align}
    w' = w \cmap{1}\left(\proj{1}\left(w^{-1}\right)\right) \cdot \cmap{2,1}\left(\proj{2}\left(\left(w \cmap{1}\left(\proj{1}\left(w^{-1}\right)\right)\right)^{-1}\right)\right).
\end{align}

First, we show that $\cmap{2,1}\left(\proj{2}\left(\left(w \cmap{1}\left(\proj{1}\left(w^{-1}\right)\right)\right)^{-1}\right)\right)$ is well defied, and that $w' = \sigma \pmod{K}$. By assumption, $w = \sigma \pmod{K}$. Equivalently, $w = k \sigma$, for some $k \in K$. Then
\begin{align}
    \proj{1}(w) = \proj{1}(k \sigma) = \proj{1}(k) \in K \cap G_1^E,
\end{align}
since $\proj{1}$ maps words in $K$ to words inside $K$ and words in $H^E$ to words in $G_1^{E}$. The map $\proj{1}$ is a homomorphism, so we also have $\proj{1}(w^{-1}) \in K \cap G_1^E$. Then, by \Cref{lem:Kfixed}, 
\begin{align}
\cmap{1}\left(\proj{1}(w^{-1})\right) \in K \cap H^E. \label{eq:halway:proj_1_in_K}
\end{align}
A similar argument shows $\proj{2}(w) \in K \cap G_2^E$. From this, and equation \ref{eq:halway:proj_1_in_K} it follows that 
\begin{align}
    \proj{2}\left(\left(w \cmap{1}\left(\proj{1}\left(w^{-1}\right)\right)\right)^{-1}\right) \in  K \cap G_2^E. 
\end{align} 
Then, by \Cref{lem:Kfixed}
\begin{align}
    \cmap{2,1}\left(\proj{2}\left(\left(w \cmap{1}\left(\proj{1}\left(w^{-1}\right)\right)\right)^{-1}\right)\right) \in K \cap H^E. 
\end{align}
Putting this all together gives
\begin{align}
    w \cdot \cmap{1}\left(\proj{1}\left(w^{-1}\right)\right) \cdot \cmap{2,1}\left(\proj{2}\left(\left(w \cmap{1}\left(\proj{1}\left(w^{-1}\right)\right)\right)^{-1}\right)\right) &= w \cdot \id \cdot \id \pmod{K} \\
    &= \sigma \pmod{K}, \label{eq:projsigmaword}
\end{align}
as desired.

To show $\proj{1}(w') = \proj{2}(w') = \id$, set $h = w  \cmap{1}\left(\proj{1}\left(w^{-1}\right)\right)$ and note
\begin{align}
    \proj{1}\left(h\right) &= \proj{1}\left(w \right) \; \proj{1}\left(\cmap{1}\left(\proj{1}\left(w^{-1}\right)\right)\right) \\
    &= \proj{1}(w) \; \proj{1}\left(w^{-1}\right) \\
    &= \id \label{eq:proj1h}.
\end{align}
 Also note $w \in H^E$ by assumption and $\cmap{1}\left(\proj{1}\left(w^{-1}\right)\right) \in H^E$ because $\Im(\cmap{1}) \in H^E$. Then $h \in H^E$ and, by Property~\labelcref{abprop:bid} of the map $\cmap{2,1}$ and \Cref{eq:proj1h}, we have 
\begin{align}
    \proj{1}\left(\cmap{2,1}\left(\proj{2}\left(h\right)\right)\right) = \id.
\end{align}
The maps $\proj{\alpha}, \cmap{\alpha},$ and $\cmap{\alpha,\beta}$ are all homomorphisms, so we also have
\begin{align}
    \proj{1}\left(\cmap{2,1}\left(\proj{2}\left(h^{-1}\right)\right)\right) = \id. \label{eq:proj1cmap21proj2hinv}
\end{align}
Then we put this all together to see 
\begin{align}
    \proj{1}(w') &= \proj{1}\left( h \; \cmap{2,1}\left(\proj{2}\left(h^{-1}\right)  \right)  \right) \\
    &= \proj{1}\left(h\right) \; \proj{1}\left( \cmap{2,1}\left(\proj{2}\left(h^{-1}\right)  \right)\right)  \\
    &= \id, \label{eq:proj1word}
\end{align}
using \Cref{eq:proj1h,eq:proj1cmap21proj2hinv} in the last line. 

Additionally, property~\labelcref{abprop:afix} of the map $\cmap{2,1}$ gives
\begin{align}
    \proj{2}(w') &= \proj{2} \left( h \; \cmap{2,1}\left(\proj{2}\left(h^{-1}\right)\right) \right) \\
    &= \proj{2} \left(h \right) \; \proj{2} \left( \cmap{2,1}\left(\proj{2}\left(h^{-1}\right)\right) \right) \\
    &= \proj{2} \left( h \right) \; \proj{2}\left(h^{-1}\right) \\
    &= \id. \label{eq:proj2word}
\end{align}
\Cref{eq:projsigmaword,eq:proj1word,eq:proj2word} complete the proof. 
\end{proof}

\subsubsection{Gadgets for processing words in \texorpdfstring{$G_3$}{G3}} \label{subsec:WordPRocessGadget}

 We are now almost ready to prove \Cref{thm:K_modding}. Before we do this, we introduce two final homomorphisms $f_1, f_2 : G_3^{E} \rightarrow H^{E}$.\footnote{We could define analogues of $f$ mapping from any $G_\alpha$. We only need the maps from $G_3$, so we give the more specific construction for notational simplicity.} As in \Cref{lem:making_homomorphisms} we introduce the properties of these homomorphisms in the following lemma, then define the homomorphisms in the lemma's proof.

\nx{ $f_\alpha$ for $\alpha \in \{1,2\}$  
homomorphism$G_3^{E} \rightarrow H^{E}$ and satisfy:
\\
    B1. \ If $\proj{\alpha}\left(\cmap{3,\alpha}(v)\right) = \id$, then $\proj{\alpha}\left(f_\alpha(v)\right) = \id$ \\
    B2. \  $\proj{\beta}\left(f_\alpha(v)\right) = \proj{\beta}\left(\cmap{3,\alpha}(v) \right)$, provided $\beta \in \{1,2\}$ and $\beta \neq \alpha$ 
    \\
    B3. \  $\proj{\alpha} \left( \cmap{3,\alpha} \left( \proj{3} \left( f_\alpha (v) \right) \right) \right) = 1$ 
    \\
    B4. \  $\proj{\beta} \left( \cmap{3,\beta} \left( \proj{3} \left( f_\alpha (v) \right) \right) \right) = \proj{\beta} \left( \cmap{3,\beta} (v) ) \right)$, provided $\beta \in \{1,2\}$ and $\beta \neq \alpha$ 
\\
for any $v \in G_3^{E}$. 
}

\begin{lem} \label{lem:gadgetmap}
There exist homomorphisms $f_\alpha$ for $\alpha \in \{1,2\}$  which map $G_3^{E} \rightarrow H^{E}$ and satisfy:
\begin{enumerate}[label=B\arabic*., ref=B\arabic*]
    \item If $\proj{\alpha}\left(\cmap{3,\alpha}(v)\right) = \id$, then $\proj{\alpha}\left(f_\alpha(v)\right) = \id$ \label{prop:mapid}
    \item $\proj{\beta}\left(f_\alpha(v)\right) = \proj{\beta}\left(\cmap{3,\alpha}(v) \right)$, provided $\beta \in \{1,2\}$ and $\beta \neq \alpha$ \label{prop:mapmatch}
    \item $\proj{\alpha} \left( \cmap{3,\alpha} \left( \proj{3} \left( f_\alpha (v) \right) \right) \right) = 1$ \label{prop:pmapid}   
    \item $\proj{\beta} \left( \cmap{3,\beta} \left( \proj{3} \left( f_\alpha (v) \right) \right) \right) = \proj{\beta} \left( \cmap{3,\beta} (v) ) \right)$, provided $\beta \in \{1,2\}$ and $\beta \neq \alpha$ \label{prop:pmapmatch}
\end{enumerate}
for any $v \in G_3^{E}$. 
\end{lem}

\begin{rmk}
 Properties \labelcref{prop:mapid,prop:mapmatch,prop:pmapmatch} are all satisfied if the homomorphism $f_\alpha$ is replaced by~$\cmap{3,\alpha}$. 
Property \labelcref{prop:pmapid} is not, but it is satisfied by $f_\alpha$ in the special case that the graph $\cG_{3\alpha}$ is connected. 
Thus, the homomorphism $f_\alpha$ can be thought of as producing words similar to those produced by the map $\cmap{3,\alpha}$, with the additional feature that they also behave as if the graph $\cG_{3\alpha}$ is connected and hence satisfy Property~\labelcref{prop:pmapid}. A motivated reader can also check that (with appropriately chosen conventions) the construction of $f_\alpha$ given later satisfies $f_\alpha = \cmap{3,\alpha}$ when $\cG_{3\alpha}$ is connected.
\end{rmk}

\Cref{lem:gadgetmap} is the last major result needed to prove \Cref{thm:K_modding}. Before proving the Lemma we build intuition for it's significance by sketching how Properties~\labelcref{prop:mapid,prop:mapmatch,prop:pmapid,prop:pmapmatch} are used in the proof of \Cref{thm:K_modding}.

\begin{proof}[Proof Outline (Part 2) of \Cref{thm:K_modding}]

Recall that \Cref{lem:pre-process} (as foreshadowed in Part1 of this proof outline) shows that existence of a word $u \in H^E$ with $u = \sigma \pmod{K}$ implies existence of a word $u'\sigma \in H^E$ with $u' \in G_3^E \cap K$. We now show how \Cref{lem:gadgetmap} lets us argue that $u'\sigma \in H^E$ implies that $\sigma \in H^E$. For simplicity,\footnote{This can be compared with the general case given in \Cref{eq:full_proof_q}.} we consider the case where $u'$ has the very basic form $u' = [v_1,v_2]$ with $v_1$ and $v_2$ in $G_3^E$. However, the intuition given here applies more generally. 

Properties~\labelcref{prop:mapid,prop:mapmatch} are used to reason about words $w \in G_3^E \cap H^E$. 
They show that (up to a factor of $\sigma$) existence of a word $w \in G_3^{E} \cap H^E$ implies that the words $\proj{3}(f_\alpha(w))$ are in $G_3^{E} \cap H^E$ for $\alpha \in \{1,2\}$.\footnote{For any $\alpha \in \{1,2\}$, we have $f_\alpha(w) \in H^E$ by definition. The content of this claim, then, is that $\proj{3}(f_\alpha(w))$ is \textit{also} in $H^E$, i.e. that $H^E$ also contains (up to a possible $\sigma$) a word with just the the $G_3^E$ portion of the word $f_\alpha(w)$.} 
To understand why, note for any $w \in G_3^{E} \cap H^E$ we have $\proj{3}(w) = w$ and $\proj{2}(w) = \id$, so $\proj{2}(\cmap{3,2}(w)) = 1$ by Property~\labelcref{abprop:bid}. 
Then $\proj{2}(\gadget{2}(w)) = \proj{2}(\cmap{3,2}(w)) = 1$ by Property \labelcref{prop:mapid} and $\proj{1}(\gadget{2}(w)) = \proj{1}(\cmap{3,2}(w))$ by Property \labelcref{prop:mapmatch}. 
Thus, 
\begin{align}
    \proj{1}((\cmap{3,2}(w))^{-1}\gadget{2}(w)) = \proj{2}((\cmap{3,2}(w))^{-1}\gadget{2}(w)) = \id. \label{eq:w'_proof_sketch}
\end{align}
Now, define $w' = w(\cmap{3,2}(w))^{-1}\gadget{2}(w)$. Since $\gadget{2}$ and $\cmap{3,2}$ both map into $H^E$ and $w \in H^E$, we have $w' \in H^E$.  We have  $\proj{1}(w') = \proj{2}(w') = \id$ by definition of $w$ and \Cref{eq:w'_proof_sketch}. Furthermore,
\begin{align}
    \proj{3}(w') &= \proj{3}(w) \proj{3}\left(\cmap{3,2}(w)^{-1}\right)\proj{3}(\gadget{2}(w)) \\
    &= w w^{-1} \proj{3}(\gadget{2}(w)) = \proj{3}(\gadget{2}(w))
\end{align}
by Property \labelcref{abprop:afix}.
We conclude that, up to a potential factor of $\sigma$, 
\begin{align}
    w \in G_3^E \cap H^E \implies  \proj{3}(\gadget{2}(w)) = w' \in G_3^E \cap H^E. \label{eq:w'_existence_proof_sketch}
\end{align} 
A similar argument applies to the homomorphism $f_1$ and proves 
\begin{align}
    w \in G_3^E \cap H^E \implies  \proj{3}(\gadget{1}(w)) \in G_3^E \cap H^E. \label{eq:w'_alt_existence_proof_sketch}
\end{align}

Property \labelcref{prop:pmapid} gives us a powerful tool for working with words of the form $\proj{3}(\gadget{\alpha}(w))$. Recall that we want to show that a word $[v_1, v_2]\sigma \in H^E$ with $v_1,v_2 \in G_3^E$ implies that $\sigma \in H^E$. Similar logic as used to show \Cref{eq:w'_existence_proof_sketch} can show 
\begin{align}
    \com{v_1}{v_2}\sigma \in H^E \implies \com{\proj{3}(f_1(v_1))}{\proj{3}(f_2(v_2))}\sigma \in H^E.
\end{align}
Now we define
\begin{align}
    q := \com{\cmap{3,1}(\proj{3}(f_1(v_1)))}{\cmap{3,2}(\proj{3}(f_2(v_2)))}
\end{align}
and note that $q \in H^E$ (because $\cmap{3,1}$ and $\cmap{3,2}$ map into $H^E$).
Using Property~\labelcref{prop:pmapid} we see 
\begin{align}
    \proj{1}\left(q \right) 
    = \com{1}{\proj{1}\left(\cmap{3,2}(\proj{3}(f_2(v_2)))\right)} = 1
\end{align}
and 
\begin{align}
    \proj{2}\left(q\right) 
    = \com{\proj{2}\left(\cmap{3,1}(\proj{3}(f_1(v_1)))\right)}{1}
    = 1
\end{align}
while Property~\labelcref{abprop:afix} of the maps $\cmap{3,\alpha}$ gives 
\begin{align}
    \proj{3}\left(q\right) 
    = \com{\proj{3}(f_1(v_1))}{\proj{3}(f_2(v_2))}
\end{align}
and direct computation gives 
\begin{align}
    \proj{\sigma}\left(q\right)
    = \com{\proj{\sigma}\left(\cmap{3,1}(\proj{3}(f_1(v_1)))\right)}{\proj{\sigma}\left(\cmap{3,2}(\proj{3}(f_2(v_2)))\right)} = 1
\end{align}
Thus, checking the action of each projection on $\proj{\alpha}$ on the word we see 
\begin{align}
    \sigma = q^{-1}  \com{\proj{3}(f_1(v_1))}{\proj{3}(f_2(v_2))}\sigma  \in H^E,
\end{align}
which is the containment needed to prove \Cref{thm:K_modding}. 

For technical reasons in the full proof of \Cref{thm:K_modding} we do not apply the homomorphisms $f_1$ and $f_2$ to separate parts of a word $w \in G_3^{E}$, but instead chain them together as $\proj{3}(f_1(\proj{3}(f_2(w))))$.\footnote{In the full proof, this composition is defined in \Cref{eq:full_proof_composition}} Property~\labelcref{prop:pmapmatch} is a technical result that tells us this chaining together of maps $f_\alpha$ behaves as desired. 

\end{proof}

\subsubsection{Proof of \texorpdfstring{\Cref{lem:gadgetmap}}{Lemma 3.13}} 
\label{subsec:proof_of_gadgets}

Now we turn to the proof of  \Cref{lem:gadgetmap}. To prepare, we construct ``gadget" words which will be used in the definition of $f_\alpha$. These words depend on the representative vertices chosen from the connected components of $\cG_{13}$ and $\cG_{23}$ when constructing the right inverses $\cmap{3,1}$ and $\cmap{3,2}$. 

To work with these representative vertices we define, for $\alpha,\beta \in \{1,3\}$, the functions $\repnew{3,1}{\alpha}{\beta}$ analogously to \Cref{eq:rep_def}.
These functions map a vertex $x_i^{(\alpha)}$ in $G_\alpha$ to the representative vertex of $G_\beta$ in the connected component of multigraph $\cG_{13}$ containing $x_i^{(\alpha)}$.  
We also define the functions $\repnew{3,2}{\alpha}{\beta}$, analogously for $\alpha, \beta \in \{2,3\}$. Next, recall the hypergraph $\cG_{123}$ defined in \Cref{subsec: new def}. Vertices are identified with elements $x_i^{(\alpha)}$, with $i \in [N]$, $\alpha \in \{1,2,3\}$. $\cG_{123}$ contains a hyperedge $(x_i^{(1)},x_j^{(2)},x_k^{(3)})$ for each clause $x_i^{(1)}x_j^{(2)}x_k^{(3)}\sigma^l \in S$, where $l$ has value $0$ or $1$. By the arguments of \Cref{sec: Connectivity}, we can assume this hypergraph is connected. Then there exist paths in $\cG_{123}$ between any two vertices. 

Now for each pair of vertices $x_i^{(\alpha)}, x_j^{(\beta)}$ let $Q(x_i^{(\alpha)},x_j^{(\beta)})$ denote some fixed minimal length path between these vertices. Then we fix some arbitrary vertex in $G_3$, wlog chosen to be $x_1^{(3)}$, and for each 
representative vertex $\repnew{3,\alpha}{\beta}{\alpha}(x_i^{(\beta)})$
with 
$\alpha \in \{1,2\}$ and $\beta \in \{3, \alpha\}$ 
consider the minimal length path $Q \left(\repnew{3,\alpha}{\beta}{\alpha}(x_i^{(\beta)}), x_1^{(3)} \right)$ from the representative vertex to $x_1^{(3)}$.
Each path corresponds to a sequence of clauses, and we can identify sequences of clauses with words in $H$. A sample hypergraph $\cG_{123}$ is introduced in \Cref{fig:hypergraph}, and a sample path is illustrated in \Cref{fig:hypergraphPath}.

Next, given a sequence of clauses $h_{p_1}, h_{p_2}, ..., h_{p_s}$ corresponding to a path in $\cG_{123}$, define the subsequence of clauses $s_{\beta}(h_{p_1}, h_{p_2}, ..., h_{p_s})$ to be the sequence including only pairs consisting of adjacent clauses  which are connected through the $G_\beta$ vertices. That is, $s_{\beta}(h_{p_1}, h_{p_2}, ..., h_{p_s})$ includes only adjacent clauses $h_{p_i}h_{p_{i+1}}$ which satisfy
\begin{align}
    \proj{\beta}(h_{p_i}) = \proj{\beta}(h_{p_{i+1}}). \label{eq:subesquence_condn}
\end{align}
Note $s_{\beta}(h_{p_1}, h_{p_2}, ..., h_{p_s})$ need not be a path. 

Finally, define words
\begin{samepage}
\begin{align}
    \protoGadget_1\left(x_i^{(\alpha)}, x_1^{(3)} \right) := s_2\left(Q\left(\repnew{3,1}{\alpha}{1}\left(x_i^{(\alpha)}\right), x_1^{(3)} \right)\right) \text{ for } \alpha \in \{1,3\} \label{eq:gamma1def}
\end{align}
and
\begin{align}
    \protoGadget_2\left(x_i^{(\alpha)}, x_1^{(3)} \right) := s_1\left(Q\left(\repnew{3,2}{\alpha}{2}\left(x_i^{(\alpha)}\right), x_1^{(3)} \right)\right) \text{ for } \alpha \in \{2,3\}. \footnotemark \label{eq:gamma2def}
\end{align}
\footnotetext{The $x_1^{(3)}$ entry in the definition of the words $\protoGadget_\beta \left(x_i^{(\alpha)}, x_1^{(3)}\right)$ is fixed, so nonessential. We keep it in the notation only to remind ourselves that the words correspond to a subsequence chosen from a sequence of clauses which corresponds to a path terminating at vertex $x_1^{(3)}$.} 
The full sequence of steps involved in the construction of $\protoGadget_2$ is visualized in \Cref{fig:hypergraphPathSubPath,fig:hypergraph,fig:hypergraphgraph,fig:hypergraphPath}. We alert the reader that we will most frequently use these gadget words with the fixed index $\alpha = 3$, but will occasionally require this more general definition.
\begin{figure}[p]
    \centering
    \captionsetup{singlelinecheck=off}
    \resizebox{10cm}{!}{
        \begin{tikzpicture}

\vertex[xpos=0, ypos=0, label name = $x_1^{(1)}$]{11}
\vertex[xpos=1, ypos=0, label name = $x_2^{(1)}$]{12}
\vertex[xpos=2, ypos=0, label name = $x_3^{(1)}$]{13}
\vertex[xpos=3, ypos=0, label name = $x_4^{(1)}$]{14}
\vertex[xpos=4, ypos=0, label name = $x_5^{(1)}$]{15}
\vertex[xpos=5, ypos=0, label name = $x_6^{(1)}$]{16}

\vertex[xpos=0, ypos=2, label position = above, label name = $x_1^{(2)}$, color = red, label color=red]{21}
\vertex[xpos=1, ypos=2, label position = above, label name = $x_2^{(2)}$]{22}
\vertex[xpos=2, ypos=2, label position = above, label name = $x_3^{(2)}$, color = red, label color=red]{23}
\vertex[xpos=3, ypos=2, label position = above, label name = $x_4^{(2)}$]{24}
\vertex[xpos=4, ypos=2, label position = above, label name = $x_5^{(2)}$, color = red, label color=red]{25}
\vertex[xpos=5, ypos=2, label position = above, label name = $x_6^{(2)}$]{26}

\vertex[xpos=0, ypos=4, label position = above, label name = $x_1^{(3)}$]{31}
\vertex[xpos=1, ypos=4, label position = above, label name = $x_2^{(3)}$]{32}
\vertex[xpos=2, ypos=4, label position = above, label name = $x_3^{(3)}$]{33}
\vertex[xpos=3, ypos=4, label position = above, label name = $x_4^{(3)}$]{34}
\vertex[xpos=4, ypos=4, label position = above, label name = $x_5^{(3)}$]{35}
\vertex[xpos=5, ypos=4, label position = above, label name = $x_6^{(3)}$]{36}

\threeedge[outline = false, fill = true, fill color = gray]{11}{21}{31}
\threeedge[outline = false, fill = true, fill color = gray]{11}{22}{31}
\threeedge[outline = false, fill = true, fill color = gray]{11}{23}{33}
\threeedge[outline = false, fill = true, fill color = gray]{12}{22}{32}
\threeedge[outline = false, fill = true, fill color = gray]{12}{23}{34}

\threeedge[outline = false, fill = true, fill color = gray]{13}{24}{34}
\threeedge[outline = false, fill = true, fill color = gray]{14}{24}{33}

\threeedge[outline = false, fill = true, fill color = gray]{15}{24}{34}
\threeedge[outline = false, fill = true, fill color = gray]{15}{26}{35}
\threeedge[outline = false, fill = true, fill color = gray]{15}{25}{35}
\threeedge[outline = false, fill = true, fill color = gray]{16}{26}{36}

\end{tikzpicture}
    }
    \caption[]{Sample hypergraph $\cG_{123}$ for a game with alphabet size $N = 6$ and 11 clauses. Representative vertices in the image of the map $\repnew{3,2}{3}{2}$ are indicated in red. The hypergraph is generated by clause set ($\sigma$ terms omitted since they don't affect the graph): 
    \begin{align*}
    S = \{ x_1^{(1)}x_1^{(2)}x_1^{(3)}, \
    x_1^{(1)}x_2^{(2)}x_1^{(3)}, \
    x_2^{(1)}x_2^{(2)}x_2^{(3)}, \
    x_1^{(1)}x_3^{(2)}x_3^{(3)}, \
    x_2^{(1)}x_3^{(2)}x_4^{(3)}, \
    x_3^{(1)}x_4^{(2)}x_4^{(3)}, \ \; \; \; \; \; \; \; \; \\
    x_4^{(1)}x_4^{(2)}x_3^{(3)}, \
    x_5^{(1)}x_4^{(2)}x_4^{(3)}, \
    x_5^{(1)}x_6^{(2)}x_5^{(3)}, \
    x_5^{(1)}x_5^{(2)}x_5^{(3)}, \
    x_6^{(1)}x_6^{(2)}x_6^{(3)}
    \}     
    \end{align*}
    }
    \label{fig:hypergraph}
\vspace{20pt}
    \resizebox{10cm}{!}{
        \begin{tikzpicture}

\vertex[xpos=0, ypos=2, label position = below, label name = $x_1^{(2)}$, color = red, label color=red]{21}
\vertex[xpos=1, ypos=2, label position = below, label name = $x_2^{(2)}$]{22}
\vertex[xpos=2, ypos=2, label position = below, label name = $x_3^{(2)}$, color = red, label color=red]{23}
\vertex[xpos=3, ypos=2, label position = below, label name = $x_4^{(2)}$]{24}
\vertex[xpos=4, ypos=2, label position = below, label name = $x_5^{(2)}$, color = red, label color=red]{25}
\vertex[xpos=5, ypos=2, label position = below, label name = $x_6^{(2)}$]{26}

\vertex[xpos=0, ypos=4, label position = above, label name = $x_1^{(3)}$]{31}
\vertex[xpos=1, ypos=4, label position = above, label name = $x_2^{(3)}$]{32}
\vertex[xpos=2, ypos=4, label position = above, label name = $x_3^{(3)}$]{33}
\vertex[xpos=3, ypos=4, label position = above, label name = $x_4^{(3)}$]{34}
\vertex[xpos=4, ypos=4, label position = above, label name = $x_5^{(3)}$]{35}
\vertex[xpos=5, ypos=4, label position = above, label name = $x_6^{(3)}$]{36}

\edge{21}{31}
\edge{22}{31}
\edge{23}{33}
\edge{22}{32}
\edge{23}{34}

\edge{24}{34}
\edge{24}{33}

\edge{24}{34}
\edge{26}{35}
\edge{25}{35}
\edge{26}{36}

\end{tikzpicture}
    }
    \caption{Graph $\cG_{23}$ corresponding to the same set of clauses as used to generate the hypergraph in \Cref{fig:hypergraph}. Representative vertices in $\cG_2$ are indicated in red.}
    \label{fig:hypergraphgraph}
\end{figure}

\end{samepage}

\begin{figure}[p]
    \centering
    \resizebox{10cm}{!}{
        \begin{tikzpicture}

\vertex[xpos=0, ypos=0, label name = $x_1^{(1)}$]{11}
\vertex[xpos=1, ypos=0, label name = $x_2^{(1)}$]{12}
\vertex[xpos=2, ypos=0, label name = $x_3^{(1)}$]{13}
\vertex[xpos=3, ypos=0, label name = $x_4^{(1)}$]{14}
\vertex[xpos=4, ypos=0, label name = $x_5^{(1)}$]{15}
\vertex[xpos=5, ypos=0, label name = $x_6^{(1)}$]{16}

\vertex[xpos=0, ypos=2, label position = above, label name = $x_1^{(2)}$, color = red, label color=red]{21}
\vertex[xpos=1, ypos=2, label position = above, label name = $x_2^{(2)}$]{22}
\vertex[xpos=2, ypos=2, label position = above, label name = $x_3^{(2)}$, color = red, label color=red]{23}
\vertex[xpos=3, ypos=2, label position = above, label name = $x_4^{(2)}$]{24}
\vertex[xpos=4, ypos=2, label position = above, label name = $x_5^{(2)}$, color = red, label color=red]{25}
\vertex[xpos=5, ypos=2, label position = above, label name = $x_6^{(2)}$]{26}

\vertex[xpos=0, ypos=4, label position = above, label name = $x_1^{(3)}$]{31}
\vertex[xpos=1, ypos=4, label position = above, label name = $x_2^{(3)}$]{32}
\vertex[xpos=2, ypos=4, label position = above, label name = $x_3^{(3)}$]{33}
\vertex[xpos=3, ypos=4, label position = above, label name = $x_4^{(3)}$]{34}
\vertex[xpos=4, ypos=4, label position = above, label name = $x_5^{(3)}$]{35}
\vertex[xpos=5, ypos=4, label position = above, label name = $x_6^{(3)}$]{36}

\threeedge[outline = false, fill = true, fill color = teal]{11}{21}{31}
\threeedge[outline = false, fill = true, fill color = teal]{11}{23}{33}
\threeedge[outline = false, fill = true, fill color = teal]{12}{23}{34}
\threeedge[outline = false, fill = true, fill color = teal]{15}{24}{34}
\threeedge[outline = false, fill = true, fill color = teal]{15}{25}{35}

\threeedge[outline = false, fill = true, fill color = gray]{11}{22}{31}

\threeedge[outline = false, fill = true, fill color = gray]{12}{22}{32}

\threeedge[outline = false, fill = true, fill color = gray]{13}{24}{34}
\threeedge[outline = false, fill = true, fill color = gray]{14}{24}{33}

\threeedge[outline = false, fill = true, fill color = gray]{15}{26}{35}

\threeedge[outline = false, fill = true, fill color = gray]{16}{26}{36}

\end{tikzpicture}
    }
    \caption{Hypergraph repeated from \Cref{fig:hypergraph}. A choice of path $Q(x_5^{(2)}, x_1^{(3)})$ is indicated in teal.}
    \label{fig:hypergraphPath}
    \resizebox{10cm}{!}{
        \begin{tikzpicture}

\vertex[xpos=0, ypos=0, label name = $x_1^{(1)}$]{11}
\vertex[xpos=1, ypos=0, label name = $x_2^{(1)}$]{12}
\vertex[xpos=2, ypos=0, label name = $x_3^{(1)}$]{13}
\vertex[xpos=3, ypos=0, label name = $x_4^{(1)}$]{14}
\vertex[xpos=4, ypos=0, label name = $x_5^{(1)}$]{15}
\vertex[xpos=5, ypos=0, label name = $x_6^{(1)}$]{16}

\vertex[xpos=0, ypos=2, label position = above, label name = $x_1^{(2)}$, color = red, label color=red]{21}
\vertex[xpos=1, ypos=2, label position = above, label name = $x_2^{(2)}$]{22}
\vertex[xpos=2, ypos=2, label position = above, label name = $x_3^{(2)}$, color = red, label color=red]{23}
\vertex[xpos=3, ypos=2, label position = above, label name = $x_4^{(2)}$]{24}
\vertex[xpos=4, ypos=2, label position = above, label name = $x_5^{(2)}$, color = red, label color=red]{25}
\vertex[xpos=5, ypos=2, label position = above, label name = $x_6^{(2)}$]{26}

\vertex[xpos=0, ypos=4, label position = above, label name = $x_1^{(3)}$]{31}
\vertex[xpos=1, ypos=4, label position = above, label name = $x_2^{(3)}$]{32}
\vertex[xpos=2, ypos=4, label position = above, label name = $x_3^{(3)}$]{33}
\vertex[xpos=3, ypos=4, label position = above, label name = $x_4^{(3)}$]{34}
\vertex[xpos=4, ypos=4, label position = above, label name = $x_5^{(3)}$]{35}
\vertex[xpos=5, ypos=4, label position = above, label name = $x_6^{(3)}$]{36}

\threeedge[outline = false, fill = true, fill color = teal, outline = true, color = black]{11}{21}{31}
\threeedge[outline = false, fill = true, fill color = teal, outline = true, color = black]{11}{23}{33}
\threeedge[outline = false, fill = true, fill color = teal, outline = true]{15}{25}{35}
\threeedge[outline = false, fill = true, fill color = teal, outline = true]{15}{24}{34}

\threeedge[outline = false, fill = true, fill color = gray]{11}{22}{31}

\threeedge[outline = false, fill = true, fill color = gray]{12}{22}{32}
\threeedge[outline = false, fill = true, fill color = teal]{12}{23}{34}

\threeedge[outline = false, fill = true, fill color = gray]{13}{24}{34}
\threeedge[outline = false, fill = true, fill color = gray]{14}{24}{33}

\threeedge[outline = false, fill = true, fill color = gray]{15}{26}{35}
\threeedge[outline = false, fill = true, fill color = gray]{16}{26}{36}

\end{tikzpicture}
    }
    \caption{Hypergraph repeated from \Cref{fig:hypergraph}. The path $Q(x_5^{(2)}, x_1^{(3)})$ is indicated in teal. The hyperedges making up $\protoGadget_2(x_5^{(2)}, x_1^{(3)})$ are outlined.}
    \label{fig:hypergraphPathSubPath}
\end{figure}
The following lemma summarizes the important properties of the gadget words $\protoGadget_2\left(x_i^{(3)}, x_1^{(3)}\right)$ and $\protoGadget_1\left(x_i^{(3)}, x_1^{(3)} \right)$. 

\nx{ $\protoGadget_2\left(x_i^{(3)}, x_1^{(3)} \right)$ and $\protoGadget_1\left(x_i^{(3)}, x_1^{(3)} \right)$, defined as in \Cref{eq:gamma2def}, satisfy:
\\
     C1. \  $\proj{1} \left( \protoGadget_2\left(x_i^{(3)}, x_1^{(3)} \right) \right) = \id$ \;\;\;\;\; and \;\;\;\;\;
    $\proj{2} \left( \protoGadget_1\left(x_i^{(3)}, x_1^{(3)} \right) \right) = \id$.
    \\
    C2. \ $\proj{2} \left(\cmap{3,2} \left( \proj{3} \left( \protoGadget_2 \left( x_i^{(3)}, x_1^{(3)} \right) \right) \right) \right)  = \proj{2} \left( \cmap{3,2} \left( x_i^{(3)}    x_1^{(3)}  \right) \right) $  \;\;\;\;\; and 
    \\ \mbox{} \qquad
    $\proj{1} \left(\cmap{3,1} \left( \proj{3} \left( \protoGadget_1 \left( x_i^{(3)}, x_1^{(3)} \right) \right) \right) \right)  = \proj{1} \left( \cmap{3,1} \left( x_i^{(3)}    x_1^{(3)}  \right) \right) $.
}


\begin{lem}
The words $\protoGadget_2\left(x_i^{(3)}, x_1^{(3)} \right)$ and $\protoGadget_1\left(x_i^{(3)}, x_1^{(3)} \right)$, defined as in \Cref{eq:gamma2def}, satisfy the following properties:
\begin{enumerate}[label=C\arabic*.,ref=C\arabic*, beginpenalty=10]
    \item $\proj{1} \left( \protoGadget_2\left(x_i^{(3)}, x_1^{(3)} \right) \right) = \id$ \;\;\;\;\; and \;\;\;\;\;
    $\proj{2} \left( \protoGadget_1\left(x_i^{(3)}, x_1^{(3)} \right) \right) = \id$.
    \label{prop:gadgetId}
    \item $\proj{2} \left(\cmap{3,2} \left( \proj{3} \left( \protoGadget_2 \left( x_i^{(3)}, x_1^{(3)} \right) \right) \right) \right)  = \proj{2} \left( \cmap{3,2} \left( x_i^{(3)}    x_1^{(3)}  \right) \right) $ \label{prop:gadgetCancel} \;\;\;\;\; and \\
    $\proj{1} \left(\cmap{3,1} \left( \proj{3} \left( \protoGadget_1 \left( x_i^{(3)}, x_1^{(3)} \right) \right) \right) \right)  = \proj{1} \left( \cmap{3,1} \left( x_i^{(3)}    x_1^{(3)}  \right) \right) $.
\end{enumerate}
\end{lem}

\begin{proof}
We show the $\protoGadget_2$ case. The proof in the  $\protoGadget_1$ case is identical up to a change of index. 

To begin the proof, we note the word $Q\left(\repnew{3,2}{3}{2}(x_i^{(3)}), x_1^{(3)}\right)$ corresponds to a minimal-length path and so there are never more than two adjacent clauses containing the same element in the $G_1$ subgroup. (If there were three or more adjacent hyperedges containing the same element in $G_1$, the middle hyperedges could be deleted and the path would remain connected, contradicting minimality).
Additionally, recall that each hyperedge in $\cG_{123}$ is of the form $(x_i^{(1)}, x_j^{(2)}, x_k^{(3)})$, i.e. the hyperedge contains exactly one vertex in each $G_\alpha$ for $\alpha \in \{1,2,3\}$. For these reasons the subsequence $s_1\left(Q\left(\repnew{3,2}{3}{2}(x_i^{(3)}), x_1^{(3)}\right)\right)$ consists of pairs of hyperedges $h_{y_j}, h_{z_j}$ which overlap on some vertex in $G_1$. Thus we can write
\begin{align}
    \protoGadget_2\left(x_i^{(3)}, x_1^{(3)}\right) = h_{y_1}h_{z_1} \; h_{y_2}h_{z_2} \; ... \; h_{y_L}h_{z_L} \label{eq:q2}
\end{align} 
where $\proj{1}\left(h_{y_j}h_{z_j}\right) = \id$. This shows Property~\labelcref{prop:gadgetId}. 

Next, we prove property~\ref{prop:gadgetCancel}. We start by numbering all clauses in the path $Q\left(\repnew{3,2}{3}{2}(x_i^{(3)}), x_1^{(3)}\right)$, so 
\begin{align}
    Q\left(\repnew{3,2}{3}{2}(x_i^{(3)}), x_1^{(3)}\right) = h_{p_1}h_{p_2} .... h_{p_R}. \label{eq:path}
\end{align}
Consider two adjacent hyperedges $h_{p_r} h_{p_{r+1}}$ in the path. Since these hyperedges appear in sequence they overlap on at least one vertex. 
\begin{enumerate}[a)]
    \item If this vertex is contained in $G_1$, then this pair of hyperedges is contained in the word $\protoGadget_2\left(x_i^{(3)}, x_1^{(3)}\right)$ and, using the notation of \Cref{eq:q2}, we have  $h_{p_r} h_{p_{r+1}} = h_{y_j} h_{z_j}$ for some~$j$.
    \item Otherwise these hyperedges overlap on a vertex corresponding to a generator of either $G_2$ or $G_3$ (equivalently, these hyperedges overlap on a vertex contained in the graph $\cG_{23}$). In that case $\proj{3}(h_{p_r})$ and $\proj{3}(h_{p_{r+1}})$ are in the same connected component in the graph $\cG_{23}$ so $\repnew{3,2}{3}{2}(\proj{3}(h_{p_r})) = \repnew{3,2}{3}{2}(\proj{3}(h_{p_{r+1}}))$. Consequently, 
    \begin{align}
        \proj{2}\left(\cmap{3,2}(\proj{3}(h_{p_r} h_{p_{r+1}}))\right) = \repnew{3,2}{3}{2}(\proj{3}(h_{p_r}h_{p_{r+1}})) = \id.
    \end{align}
    The first equality holds by \Cref{eq:proj2r12}. 
\end{enumerate} 
Now consider a contiguous string of hyperedges of the form $h_{z_j}, h_{p_{r+1}}, h_{p_{r+2}} ... , h_{p_{r + r'}}, h_{y_{j+1}}$ contained in the path \labelcref{eq:path}. Here $h_{z_j}$ and $h_{y_{j+1}}$ belong to the path $\protoGadget_2\left(x_i^{(3)}, x_1^{(3)}\right)$, but $h_{p_{r+1}} ... h_{p_{r+r'}}$ do not. By definition of the subsequence $\protoGadget_2\left(x_i^{(3)}, x_1^{(3)}\right)$, no adjacent hyperedges between $h_{z_j}$ and $h_{y_{j+1}}$ overlap on a vertex in the $G_1$ subspace, else they would be contained in the subsequence $\protoGadget_2\left(x_i^{(3)}, x_1^{(3)}\right)$, a contradiction. Now that the intermediate clauses $h_{p_{r+1}}, ..., h_{p_{r+r'}}$ are introduced we apply the observation of the previous paragraph inductively to see
\begin{align}
       \proj{2}\left(\cmap{3,2}(\proj{3}(h_{z_j}h_{p_{r+1}}))\right) = \proj{2}\left(\cmap{3,2}(\proj{3}(h_{p_{r+1}}h_{p_{r+2}}))\right) = ... 
       = \proj{2}\left(\cmap{3,2}(\proj{3}(h_{p_{r + r'}}h_{y_{j+1}}))\right) = 1
\end{align}
Multiplying these terms together and noting adjacent clauses cancel shows
\begin{align}
    \proj{2}\left(\cmap{3,2}(\proj{3}(h_{z_j} h_{y_{j+1}} ))\right) = \id.
\end{align}
for any $j < L$. Now we use this observation inductively, and compute
\begin{align}
    \proj{2}\left(\cmap{3,2}\left( \proj{3} \left( \protoGadget_2\left(x_i^{(3)}, x_1^{(3)}\right) \right) \right)\right) &= \proj{2}\left(\cmap{3,2}\left(\proj{3}\left( h_{y_1}h_{z_1}h_{y_2}h_{z_2}, ... , h_{y_L}h_{z_L} \right) \right) \right) \\
    &= \proj{2}\left(\cmap{3,2}\left(\proj{3}\left( h_{y_1} h_{z_L} \right) \right) \right) \\
    &= \proj{2} \left( \cmap{3,2}  \left( x_i^{(3)}  x_1^{(3)}  \right)\right)
\end{align}
where we used on the last line the fact that $\proj{3}(h_{y_1})$ was in the same connected component in $\cG_{23}$ as~$x_i^{(3)}$ and that $\proj{3}(h_{z_L}) = x_1^{(3)}$, so \begin{align}
\proj{2}\left(\cmap{3,2}\left(\proj{3}\left( h_{y_1} h_{z_L} \right) \right) \right) &= \repnew{3,2}{3}{2}\left(\proj{3}(h_{y_1})\right) \repnew{3,2}{3}{2}\left(\proj{3}(h_{z_L})\right)\\
&= \repnew{3,2}{3}{2}\left( x_i^{(3)} \right) \repnew{3,2}{3}{2}\left( x_1^{(3)} \right) \\
&= \proj{2} \left( \cmap{3,2}  \left( x_i^{(3)}  x_1^{(3)}  \right)\right).
\end{align}
by definition of $\repnew{3,2}{3}{2}$ and \Cref{eq:proj2r12}.
This proves Property~\labelcref{prop:gadgetCancel}. 
\end{proof}

In addition to the gadget words defined above, we will need to recall the properties of the paths $\cpathnewone{\alpha,\beta}$, defined analogously to the path $\cpathnewone{1,2}$ defined at \Cref{eq:cpath_def}; they are used to construct the homomorphisms $\cmap{\alpha,\beta}$. In particular, we care about the properties of those paths when $\alpha$ is $3$ and $\beta$ is $1$ or $2$. We recall the properties of those paths in the following lemma. 

\nx{
 $\cpathnew{3,\beta}{x_i^{(3)}}{\repnew{3, \beta }{3}{\beta}(x_i^{(3)})}$ a path satisfies:
\\ 
    D1. \ $\proj{3}\left(\cpathnew{3,\beta}{x_i^{(3)}}{\repnew{3, \beta}{3}{\beta}(x_i^{(3)})}\right) = x_i^{(3)}$ 
    \\
    D2. \ $\proj{\beta}\left(\cpathnew{3,\beta}{x_i^{(3)}}{\repnew{3, \beta }{3}{\beta}(x_i^{(3)})}\right) = \repnew{3, \beta }{3}{\beta}(x_i^{(3)})$ 
    \\
    D3. \ $\cpathnew{3,\beta}{x_i^{(3)}}{\repnew{3, \beta}{3}{\beta}(x_i^{(3)})} \cpathnew{3,\beta}{x_j^{(3)}}{\repnew{3, \beta}{3}{\beta}(x_j^{(3)})}^{-1} = \cmap{3,\beta}(x_i^{(3)}x_j^{(3)})$ 
    \\
    for $\beta \in \{1,2\}$ and $x_i^{(3)} \in G_3$
}
\begin{lem}
For $\beta \in \{1,2\}$ and $x_i^{(3)} \in G_3$, the path $\cpathnew{3,\beta}{x_i^{(3)}}{\repnew{3, \beta }{3}{\beta}(x_i^{(3)})}$ satisfies the following properties 
\begin{enumerate}[label=D\arabic*.,ref=D\arabic*, beginpenalty=10]
    \item $\proj{3}\left(\cpathnew{3,\beta}{x_i^{(3)}}{\repnew{3, \beta}{3}{\beta}(x_i^{(3)})}\right) = x_i^{(3)}$ 
    \label{prop:cpath3}
    \item $\proj{\beta}\left(\cpathnew{3,\beta}{x_i^{(3)}}{\repnew{3, \beta }{3}{\beta}(x_i^{(3)})}\right) = \repnew{3, \beta }{3}{\beta}(x_i^{(3)})$ \label{prop:cpath12}
    \item $\cpathnew{3,\beta}{x_i^{(3)}}{\repnew{3, \beta}{3}{\beta}(x_i^{(3)})} \cpathnew{3,\beta}{x_j^{(3)}}{\repnew{3, \beta}{3}{\beta}(x_j^{(3)})}^{-1} = \cmap{3,\beta}(x_i^{(3)}x_j^{(3)})$ \label{prop:cpath_cmap}
\end{enumerate}
\end{lem}
\begin{proof}
Properties~\labelcref{prop:cpath12,prop:cpath3} follow from the properties of paths in the graph $\cG_{3, \beta}$, as discussed in the proof of \Cref{lem:making_homomorphisms}. Property~\labelcref{prop:cpath_cmap} is just the definition of $\cmap{3,\beta}$,
analogous to \Cref{eq:cmap_def}.
\end{proof}

Now we use the gadget words $\protoGadget_1(x_i^{(3)}, x_1^{(3)})$ and $\protoGadget_2(x_i^{(3)}, x_1^{(3)})$ along with the paths $\cpathnewone{3,\beta}$ to prove \Cref{lem:gadgetmap}.

\begin{proof}[Proof (\Cref{lem:gadgetmap})]
Recall that \Cref{lem:gadgetmap} claimed the existence of homomorphisms $f_1$ and $f_2$ which map $G_3^E \rightarrow H^E$ and satisfy certain desiderata (Properties \labelcref{prop:pmapmatch,prop:pmapid,prop:mapmatch,prop:mapid}). We will now give an explicit construction of these homomorphisms.

Define the homomorphism $f_1 : G_3^E \rightarrow H^E$ by its action on the basis elements
\begin{align}
    &f_1(x_{i}^{(3)}x_{j}^{(3)} ) \nonumber \\
    &:= \cpathnew{3,1}{x_i^{(3)}}{\repnew{3, 1}{3}{1}(x_i^{(3)})}
    \protoGadget_1\left(x_{i}^{(3)}, x_1^{(3)}\right)
    \left(
    \cpathnew{3,1}{x_j^{(3)}}{\repnew{3, 1}{3}{1}(x_j^{(3)})}
    \protoGadget_1\left(x_{j}^{(3)}, x_1^{(3)}\right)
    \right)^{-1} \\
    &= 
    \cpathnew{3,1}{x_i^{(3)}}{\repnew{3, 1}{3}{1}(x_i^{(3)})}
    \protoGadget_1\left(x_{i}^{(3)}, x_1^{(3)}\right)
    \protoGadget_1\left(x_{j}^{(3)}, x_1^{(3)}\right)^{-1}
    \cpathnew{3,1}{x_j^{(3)}}{\repnew{3, 1}{3}{1}(x_j^{(3)})}^{-1},
\end{align}
with $f_2$ defined similarly. Both maps are homomorphisms by \Cref{lem:making_homomorphisms}. It remains to show they satisfy Properties \labelcref{prop:pmapmatch,prop:pmapid,prop:mapmatch,prop:mapid}.  We will show explicitly that the homomorphism $f_1$ satisfies these properties; the reader can see that the argument for $f_2$ is identical.

Property \labelcref{prop:mapmatch} applied to homomorphism $f_1$ requires that 
\begin{align}
    \proj{2}\left(f_1(v)\right) = \proj{2}\left(\cmap{3,1}(v) \right)
\end{align}
for any $v \in G_3^{E}$.  We prove this by checking the action of $f_1$ on the generators of $G_3^{E}$. Direct calculation gives
\begin{align}
    &\proj{2}\left( f_1 (x_{i}^{(3)}x_{j}^{(3)})\right) \nonumber \\ 
    &= \proj{2}\left(\cpathnew{3,1}{x_i^{(3)}}{\repnew{3, 1}{3}{1}(x_i^{(3)})}
    \protoGadget_1\left(x_{i}^{(3)}, x_1^{(3)}\right)
    \protoGadget_1\left(x_{j}^{(3)}, x_1^{(3)}\right)^{-1}
    \cpathnew{3,1}{x_j^{(3)}}{\repnew{3, 1}{3}{1}(x_j^{(3)})}^{-1}\right)
    \\
    &= \proj{2}\left(\cpathnew{3,1}{x_i^{(3)}}{\repnew{3, 1}{3}{1}(x_i^{(3)})}
    \cpathnew{3,1}{x_j^{(3)}}{\repnew{3, 1}{3}{1}(x_j^{(3)})}^{-1}\right)\\
    &=\proj{2} \left(\cmap{3,1}\left(x_{i}^{(3)}x_{j}^{(3)}\right) \right),
\end{align}
where we used Property~\ref{prop:gadgetId} of the words $\protoGadget_1(x_i^{(3)}, x_1^{(3)})$ to go from the second line to the third, and Property~\ref{prop:cpath_cmap} of the words $\cpathnewone{3,1}$. to go from the third line to fourth.

Property \labelcref{prop:mapid} applied to homomorphism $f_1$ requires that 
\begin{align}
 \proj{1}\left(f_1(v)\right) = \id
\end{align}
for any $v \in  G_3^{E}$ with $\proj{1}\left(\cmap{3,1}(v)\right) = \id$. The proof of this is similar to the proof of Property~\ref{abprop:bid} of the map $\cmap{\alpha,\beta}$. Recall the function $\repnew{3,1}{\alpha}{1}$, defined to map a vertex $x_i^{(\alpha)}$ in $G_\alpha$ with $\alpha \in \{1,3\}$ to the representative vertex $x_{j}^{(1)}$ in the connected component of graph $G_{13}$ containing $x_i^{(\alpha)}$. Define the homomorphism $\lambda_1 : G_1^{E} \rightarrow G_1^{E}$ by extending
\begin{align} \label{eq:typo5}
    \lambda_1 (x_i^{(1)} x_j^{(1)}) = \repnew{3,1}{1}{1}\left(x_i^{(1)}\right)\proj{1}\left( \protoGadget_1 \left(x_i^{(1)}, x_1^{(3)} \right) \right)  \left( \repnew{3,1}{1}{1}\left(x_j^{(1)}\right)\proj{1}\left(\protoGadget_1\left(x_j^{(1)}, x_1^{(3)}\right) \right)  \right)^{-1}
\end{align}
as in \Cref{lem:homFromGens}.

Then we claim
\begin{align}
    \lambda_1(\proj{1}(h)) = \proj{1}\left( \gadget{1}\left(\proj{3}\left( h \right) \right) \right) \label{eq:A2claim}
\end{align}
for any $h \in H^E$. 
As in the proof of Property~\labelcref{abprop:bid}, we check this claim directly on the generators of $H^E$:
\begin{align}
    &\proj{1}\left( \gadget{1}\left(\proj{3}\left( h_i h_j \right) \right) \right)\nonumber \\ 
    &= \proj{1}\left( \gadget{1}\left(\proj{3}\left( x_{a_i}^{(1)}x_{b_i}^{(2)}x_{c_i}^{(3)}x_{a_j}^{(1)}x_{b_j}^{(2)}x_{c_j}^{(3)}\sigma^{s_i + s_j}\right) \right) \right) \\
    &= \proj{1}\left( \gadget{1}\left(x_{c_i}^{(3)}x_{c_j}^{(3)} \right) \right) \\
    &= \proj{1} \left( \cpathnew{3,1}{x_{c_i}^{(3)}}{\repnew{3, 1}{3}{1}(x_{c_i}^{(3)})}
    \protoGadget_1\left(x_{c_i}^{(3)}, x_1^{(3)}\right)
    \protoGadget_1\left(x_{c_j}^{(3)}, x_1^{(3)}\right)^{-1}
    \cpathnew{3,1}{x_{c_j}^{(3)}}{\repnew{3, 1}{3}{1}(x_{c_j}^{(3)})}^{-1} \right) \\
    &= \repnew{3,1}{3}{1}\left(x_{c_i}^{(3)}\right)\proj{1}\left(\protoGadget_1\left(x_{c_i}^{(3)}, x_1^{(3)} \right)
    \left(\protoGadget_1\left(x_{c_j}^{(3)}, x_1^{(3)} \right)\right)^{-1}\right)
    \left(\repnew{3,1}{3}{1}\left(x_{c_j}^{(3)}\right)\right)^{-1}
    \label{eq:projecting_on_gadgets}\\
    &= \repnew{3,1}{1}{1}\left(x_{a_i}^{(1)}\right)\proj{1}\left(\protoGadget_1\left(x_{a_i}^{(1)}, x_1^{(3)} \right)
    \left(\protoGadget_1\left(x_{a_j}^{(1)}, x_1^{(3)} \right)\right)^{-1}\right)
    \left(\repnew{3,1}{1}{1}\left(x_{a_j}^{(1)}\right)\right)^{-1}
    \label{eq:shiftingVertex}\\
    &= \lambda_1 \left( x_{a_i}^{(1)} x_{a_j}^{(1)} \right) \label{eq:typo6}\\
    &= \lambda_1(\proj{1}(h_i h_j))
\end{align}
Note to get line~\labelcref{eq:projecting_on_gadgets} we used Property~\labelcref{prop:cpath12} of the paths $\cpathnewone{3,1}$. 
The key argument comes in getting to line \ref{eq:shiftingVertex} where we used the fact that $x_{a_i}^{(1)}$ and $x_{c_i}^{(3)}$ are both contained in the clause $h_i$, so the vertices corresponding to $x_{a_i}^{(1)}$ and $x_{c_i}^{(3)}$ are in the same connected component of $G_{13}$ and consequently,
\begin{align}
    \repnew{3,1}{3}{1}\left(x_{c_i}^{(3)}\right) = \repnew{3,1}{1}{1}\left(x_{a_i}^{(1)}\right)
\end{align}
and 
\begin{align}
    \protoGadget_1\left(x_{c_i}^{(3)}, x_1^{(3)}\right) = \protoGadget_1\left(x_{a_i}^{(1)}, x_1^{(3)}\right).
\end{align}
Since $\lambda_1, \proj{1}, \gadget{1},$ and $\proj{3}$ are all homomorphisms, this proves the claim. 

Next, for any $v \in G_3^E$ satisfying $\proj{1}\left(\cmap{3,1}(v)\right) = 1$ we use \Cref{eq:A2claim} with $h = \cmap{3,1}(v)$ to conclude (recalling that $v = \proj{3}\left(\cmap{3,1}(v)\right)$ by Property~\labelcref{abprop:afix}): 
\begin{align}
    \proj{1}\left(f_1(v)\right) &= \proj{1}\left(f_1\left( \proj{3} \left( \cmap{3,1} \left( v \right) \right) \right)\right) \\
    &= \lambda_1 \left( \proj{1}\left( \cmap{3,1} \left( v \right) \right) \right) \\
    &= \lambda_1 (1) = 1
\end{align}
which proves property \labelcref{prop:mapid}.

Property~\ref{prop:pmapmatch} applied to the homomorphism $f_1$ requires that 
\begin{align}
    \proj{2} \left( \cmap{3,2} \left( \proj{3} \left( f_1 (v) \right) \right) \right) = \proj{2} \left( \cmap{3,2} (v) ) \right) \label{eq:propB4}
\end{align}
for any $v \in G_3^{E}$.
This  follows from Property~\ref{prop:gadgetId} of the words $\protoGadget_1(x_i^{(3)}, x_1^{(3)})$ and Property~\ref{abprop:bid} of the map $\cmap{3,2}$.
Property~\ref{prop:gadgetId} gives 
\begin{align}
    \proj{2}\left(\protoGadget_1(x_i^{(3)}, x_1^{(3)})\right) = \id
\end{align}
and then Property~\ref{abprop:bid} gives 
\begin{align}
    \proj{2}\left(\cmap{3,2} \left(\proj{3} \left( \protoGadget_1(x_i^{(3)}, x_1^{(3)})\right)   \right) \right) = \id. \label{eq:longcancel}
\end{align}
The idea is that the gadget words inserted by the map $f_1$ map to the identity under $\proj{2}\left(\cmap{3,2} \left( \proj{3} \right) \right)$ and Property~\ref{prop:pmapmatch} follows. We verify Property~\ref{prop:pmapmatch} 
algebraically by checking that \Cref{eq:propB4} holds on the generators of $G_3^{E}$: 
\begin{align}
    &\proj{2}\left(\cmap{3,2} \left( \proj{3} \left( f_1(x_{i}^{(3)}x_{j}^{(3)} ) \right) \right) \right) \\
    &\hspace{20pt}= 
    \proj{2}\left(\cmap{3,2} \left( \proj{3} \left( \cpathnew{3,1}{x_i^{(3)}}{\repnew{3, 1}{3}{1}(x_i^{(3)})}
    \right) \right) \right) \nonumber \\
    &\hspace{40pt}\proj{2}\left(\cmap{3,2} \left( \proj{3} \left(
    \protoGadget_1\left(x_{i}^{(3)}, x_1^{(3)}\right)
    \right) \right) \right) 
    \proj{2}\left(\cmap{3,2} \left( \proj{3} \left(
    \protoGadget_1\left(x_{j}^{(3)}, x_1^{(3)}\right)^{-1}
    \right) \right) \right) \nonumber \\
    &\hspace{60pt}\proj{2}\left(\cmap{3,2} \left( \proj{3} \left(
    \cpathnew{3,1}{x_j^{(3)}}{\repnew{3, 1}{3}{1}(x_j^{(3)})}^{-1}
    \right) \right) \right) \\
    &\hspace{20pt}= 
    \proj{2}\left(\cmap{3,2} \left( \proj{3} \left( \cpathnew{3,1}{x_i^{(3)}}{\repnew{3, 1}{3}{1}(x_i^{(3)})}
    \right) \right) \right) \nonumber \\
    &\hspace{40pt}\proj{2}\left(\cmap{3,2} \left( \proj{3} \left(
    \cpathnew{3,1}{x_j^{(3)}}{\repnew{3, 1}{3}{1}(x_j^{(3)})}^{-1}
    \right) \right) \right) \\
    & \hspace{20pt} = \proj{2}\left(\cmap{3,2} \left( x_i^{(3)} x_j^{(3)} \right) \right).
\end{align}
Where we used \Cref{eq:longcancel} to go from the second line to the third, and Property \labelcref{prop:cpath3} of the paths $\cpathnewone{3,1}$ to go from the third line to the fourth. 

Finally, Property~\ref{prop:pmapid} applied to $f_1$ requires that 
\begin{align}
    \proj{1} \left( \cmap{3,1} \left( \proj{3} \left( f_1 (v) \right) \right) \right) = 1
\end{align}
for any $v \in G_3^{E}$. This relies heavily on 
Property~\ref{prop:gadgetCancel} of the words $\protoGadget_1(x_i^{(3)}, x_1^{(3)})$. Because $v$ has even length, we can write 
\begin{align}
    v = \prod_i  x_{o_i}^{(3)} x_{e_i}^{(3)}.
\end{align}
Then 
\begin{align}
    f_1(v) &= \prod_i \bigg(\cpathnew{3,1}{x_{o_i}^{(3)}}{\repnew{3, 1}{3}{1}(x_{o_i}^{(3)})}
    \protoGadget_1\left(x_{o_i}^{(3)}, x_1^{(3)}\right)\nonumber\\
    &\hspace{60pt}\protoGadget_1\left(x_{e_i}^{(3)}, x_1^{(3)}\right)^{-1}
    \cpathnew{3,1}{x_{e_i}^{(3)}}{\repnew{3, 1}{3}{1}(x_{e_i}^{(3)})}^{-1}\bigg)
\end{align}
and using Property~\ref{prop:gadgetCancel} of the words $\protoGadget_1(x_i^{(3)}, x_1^{(3)})$ and Property~\ref{prop:cpath3} of the paths $\cpathnewone{3,1}$ gives
\begin{align}
   &\proj{1} \left( \cmap{3,1} \left( \proj{3} \left( f_1(v) \right) \right) \right) \nonumber\\
   &\hspace{20pt}= \prod_i \proj{1} \bigg( \cmap{3,1} \bigg( \proj{3} \bigg( \cpathnew{3,1}{x_{o_i}^{(3)}}{\repnew{3, 1}{3}{1}(x_{o_i}^{(3)})}
    \protoGadget_1\left(x_{o_i}^{(3)}, x_1^{(3)}\right)\nonumber\\
    &\hspace{110pt}\protoGadget_1\left(x_{e_i}^{(3)}, x_1^{(3)}\right)^{-1}
    \cpathnew{3,1}{x_{e_i}^{(3)}}{\repnew{3, 1}{3}{1}(x_{e_i}^{(3)})}^{-1}\bigg)\bigg)\bigg)\\
    &\hspace{20pt}= \prod_{i} \proj{1} \left( \cmap{3,1} \left( \;\; x_{o_i}^{(3)}  (x_{o_i}^{(3)} x_1^{(3)})  (x_1^{(3)} x_{e_i}^{(3)})  x_{e_i}^{(3)} \;\; \right) \right) \\
    &\hspace{20pt}= \prod_{i} 1 = \id. 
\end{align}
This shows Property~\labelcref{prop:pmapid} and completes the proof of \Cref{lem:gadgetmap}.
\end{proof}

One final nice property of the maps $\gadget{1}, \gadget{2}$ that we need to show is that they map words inside the $K$ subgroup to words inside the $K$ subgroup.  We show that in the following lemma. 

\begin{lem} \label{lem:gadgetKtoK}
For any $v \in K \cap G_3^E$ we have 
\begin{align}
    \gadget{1}(v), \gadget{2}(v) \in K.
\end{align}
\end{lem}

\begin{proof}
By assumption, we can write 
\begin{align}
    v = \prod_i u_i \com{x_{a_{i_1}}^{(3)}x_{a_{i_2}}^{(3)}}{x_{a_{i_3}}^{(3)}x_{a_{i_4}}^{(3)}} u_i^{-1}.
\end{align}
Then,
\begin{align}
    \gadget{1}(v) &= \prod_i \gadget{1}\left(u_i\right) \gadget{1}\left(\com{x_{a_{i_1}}^{(3)}x_{a_{i_2}}^{(3)}}{x_{a_{i_3}}^{(3)}x_{a_{i_4}}^{(3)}}\right) \gadget{1}\left(u_i^{-1}\right) \\
    &= \prod_i \gadget{1}\left(u_i\right) \com{\gadget{1}\left(x_{a_{i_1}}^{(3)}x_{a_{i_2}}^{(3)}\right)}{\gadget{1}\left(x_{a_{i_3}}^{(3)}x_{a_{i_4}}^{(3)}\right)} \gadget{1}\left(u_i^{-1}\right)
\end{align}
We have $\gadget{1}\left(x_{a_{i_1}}^{(3)}x_{a_{i_2}}^{(3)}\right), \gadget{1}\left(x_{a_{i_3}}^{(3)}x_{a_{i_4}}^{(3)}\right) \in G^E$, so (by \Cref{lem:even_com_in_K} in the appendix)
\begin{align}
    \com{\gadget{1}\left(x_{a_{i_1}}^{(3)}x_{a_{i_2}}^{(3)}\right)}{\gadget{1}\left(x_{a_{i_3}}^{(3)}x_{a_{i_4}}^{(3)}\right)} \in K.
\end{align}
But $K$ is normal, so we also have 
\begin{align}
    \gadget{1}\left(u_i\right) \com{\gadget{1}\left(x_{a_{i_1}}^{(3)}x_{a_{i_2}}^{(3)}\right)}{\gadget{1}\left(x_{a_{i_3}}^{(3)}x_{a_{i_4}}^{(3)}\right)} \gadget{1}\left(u_i^{-1}\right) \in K
\end{align}
for all $i$, hence 
\begin{align}
   \gadget{1}(v) = \prod_i \gadget{1}\left(u_i\right) \com{\gadget{1}\left(x_{a_{i_1}}^{(3)}x_{a_{i_2}}^{(3)}\right)}{\gadget{1}\left(x_{a_{i_3}}^{(3)}x_{a_{i_4}}^{(3)}\right)} \gadget{1}\left(u_i^{-1}\right) \in K.
\end{align}
The proof for $\gadget{2}$ is identical. 
\end{proof}

As a corollary, we note that the maps $f_1$, $f_2$ don't introduce any undesired factors of $\sigma$.

\begin{cor} \label{lem:noSigmaInfs}
For any word $v \in K \cap G_3^{E}$, we have 
\begin{align}
    \proj{\sigma}(\gadget{1}(v)) = \proj{\sigma}(\gadget{2}(v)) = 1
\end{align}
\end{cor}

\begin{proof}
Similarly to the proof of \Cref{cor:rightInvId}, note that $\gadget{1}(v) \in K$ by \Cref{lem:gadgetKtoK}, so $\proj{\sigma}\left(\gadget{1}(v)\right) = 1$ by \Cref{lem:noSigmaInK}. The proof for $\gadget{2}$ is similar. 
\end{proof}

\subsubsection{Proof of \texorpdfstring{\Cref{thm:K_modding}}{Theorem 2.6}} \label{susec:finalProof}

Finally, we are ready to prove \Cref{thm:K_modding}.

\begin{proof}[Proof or \Cref{thm:K_modding}.]
It is immediate that 
\begin{align}
    \sigma \in H^E \implies [\sigma]_K \in H^E \pmod{K}.
\end{align}
To see the reverse direction, assume that $[\sigma]_K \in H^E \pmod{K}$. Then there exists some $w \in H^E$ satisfying $w = \sigma \pmod{K}$. By \Cref{lem:pre-process}, there exists a word $w' \in H$ satisfying $\proj{1}(w') = \proj{2}(w') = 1$ and $w' = \sigma \pmod{K}$. Note that the last condition implies that $w' = \sigma k$ for some $k \in K$, hence \begin{align}
    \proj{3}(w') = \proj{3}(\sigma k) = k \in K \cap G_3^{E}. \label{eq:primeinK}
\end{align}
We choose words $u_i \in G_3^{E}$ and indices $a_{i_1}, ... , a_{i_4} \in [N]$ so that
\begin{align}
\proj{3}(w') = \prod_i u_i \com{x_{a_{i_1}}^{(3)}x_{a_{i_2}}^{(3)}}{x_{a_{i_3}}^{(3)}x_{a_{i_4}}^{(3)}} u_i^{-1}.
\end{align}
Now we multiply gadgets onto $w'$. Consider the word
\begin{align}
    w'' = w' \cmap{3,1}\left(\proj{3}(w')\right)^{-1} f_1(\proj{3}\left(w'\right))
\end{align}
Note that $\proj{1}(w') = 1$, and $w' \in H$. Hence 
\begin{align}
    \proj{1}\left(\cmap{3,1}\left(\proj{3}(w') \right)\right) = \id \qquad \text{ and } \qquad \proj{1}\left(f_1\left(\proj{3}(w') \right)\right) = 1, \label{eq:final_proof_w'_properties}
\end{align}
the first by Property~\ref{abprop:bid} of the map $\cmap{3,1}$ and the second by property~\ref{prop:mapid} of $f_1$. Putting this all together,  
\begin{align}
    \proj{1}\left( w'' \right) &= \proj{1}\left( w' \right) \proj{1}\left(\cmap{3,1}\left(\proj{3}(w')\right)^{-1} \right) \proj{1}\left(f_1(\proj{3}\left(w'\right)) \right) = 1.
\end{align}
By Property~\ref{prop:mapmatch} of the map $f_1$ we have 
\begin{align}
    \proj{2}\left( w'' \right) &= \proj{2}\left( w' \right) \proj{2}\left(\cmap{3,1}\left(\proj{3}(w')\right)^{-1} \right) \proj{2}\left(f_1(\proj{3}\left(w'\right)) \right) \\
    &= \proj{2}\left(\cmap{3,1}\left(\proj{3}(w')\right)^{-1} \right) \proj{2}\left(\cmap{3,1}(\proj{3}\left(w'\right)) \right) = \id
\end{align} 
Finally 
\begin{align}
    \proj{3}\left( w'' \right) = \proj{3}\left(f_1(\proj{3}\left(w'\right)) \right)
\end{align}
by Property~\ref{abprop:afix} of the map $\cmap{3,1}$. Also note that $\proj{3}\left(f_1(\proj{3}\left(w'\right)) \right) \in K$ by \Cref{lem:gadgetKtoK} and the fact that $\proj{3}$ maps words in $K$ to words in $K$ (\Cref{lem:projKtoK}).

We summarize: 
\begin{align}
    \proj{1}(w'') = \proj{2}(w'') = 1, \label{eq:primeprimeis1}
\end{align}
and 
\begin{align}
    \proj{3}(w'') = \proj{3}\left(f_1(\proj{3}\left(w'\right)) \right) \in K. \label{eq:primeprimeinK}
\end{align}

Now we again multiply gadgets onto $w''$ with the $1$ and $2$ indices swapped. Recall
\begin{align}
    w'' =  w' \cmap{3,1}\left(\proj{3}(w')\right)^{-1} f_1(\proj{3}\left(w'\right)),
\end{align}
then define 
\begin{align}
    w''' = w'' \cmap{3,2}\left(\proj{3}(w'')\right)^{-1} f_2(\proj{3}\left(w''\right)) 
\end{align}
The same arguments as used to show \Cref{eq:primeprimeinK,eq:primeprimeis1} then give
\begin{align}
    \proj{1}(w''') = \proj{2}(w''') = \id \label{eq:penul_cancel}
\end{align}
and 
\begin{align}
    \proj{3}(w''') &= \proj{3}\left(f_2\left(\proj{3}\left( w'' \right)\right)\right) \\
    &= \proj{3}\left(f_2\left(\proj{3}\left(f_1(\proj{3}\left(w'\right)) \right)\right)\right) \in K.
\end{align}
We have, by assumption, 
\begin{align}
    \proj{3}(w') = \prod_i u_i \com{x_{a_{i_1}}^{(3)}x_{a_{i_2}}^{(3)}}{x_{a_{i_3}}^{(3)}x_{a_{i_4}}^{(3)}} u_i^{-1}. \label{eq:full_proof_q}
\end{align}
We define a composition of maps $F : G_3^E \rightarrow G_3^E$
\begin{align} 
    \total :=  \proj{3} \circ f_2 \circ \proj{3} \circ f_1. \label{eq:full_proof_composition}
\end{align}
Then we have
\begin{align}
    \proj{3}(w''') &= \total \left( \prod_i  u_i \com{x_{a_{i_1}}^{(3)}x_{a_{i_2}}^{(3)}}{x_{a_{i_3}}^{(3)}x_{a_{i_4}}^{(3)}} u_i^{-1} \right)\\
    &= \prod_i \total \left(  u_i \com{x_{a_{i_1}}^{(3)}x_{a_{i_2}}^{(3)}}{x_{a_{i_3}}^{(3)}x_{a_{i_4}}^{(3)}} u_i^{-1} \right) \\
    &= \prod_i \total \left(u_i\right) \com{ \total \left(x_{a_{i_1}}^{(3)}x_{a_{i_2}}^{(3)} \right)}{ \total \left(x_{a_{i_3}}^{(3)}x_{a_{i_4}}^{(3)}\right)} \total \left(u_i^{-1} \right). \label{eq:Full_proof_F_appliction}
\end{align}
where we used the fact that each word $u_i \com{x_{a_{i_1}}^{(3)}x_{a_{i_2}}^{(3)}}{x_{a_{i_3}}^{(3)}x_{a_{i_4}}^{(3)}} u_i^{-1}$ has even length on the first line, and that each word $u_i$ has even length on the second.

Now
\begin{align}
    \proj{2}\left(\cmap{3,2} \left( \total \left( x_j^{(3)}x_k^{(3)} \right) \right) \right)  &= \proj{2}\left(\cmap{3,2} \left( \proj{3} \left( f_{2} \left( \proj{3} \left( f_{1} \left( x_j^{(3)}x_k^{(3)} \right) \right) \right) \right)\right) \right) = \id \label{eq:final2cancel}
\end{align}
by Property~\ref{prop:pmapid}. Next 
\begin{align}
    \proj{1}\left(\cmap{3,1}\left( \total \left( x_j^{(3)}x_k^{(3)} \right) \right) \right) 
    &= \proj{1}\left(\cmap{3,1}\left( \proj{3} \left( f_{2} \left( \proj{3} \left( f_{1} \left( x_j^{(3)}x_k^{(3)} \right) \right) \right) \right)\right)\right) \\
    &= \proj{1}\left(\cmap{3,1}\left( \proj{3} \left( f_{1} \left( x_j^{(3)}x_k^{(3)} \right) \right) \right)\right) = \id \label{eq:final1cancel}
\end{align}
where we used Property~\ref{prop:pmapmatch} and then Property~\ref{prop:pmapid} of the maps $f_{2}$ and $f_{1}$.

Finally, consider the word\footnote{Below we could have replaces the $\cmap{3}$ appearing in the term $\cmap{3} \left( \total \left( {u}_i \right) \right)$ with either $\cmap{3,1}$ or $\cmap{3,2}$ and the proof would remain correct.} 
\begin{align}
    w^{''''} = \prod_i  \cmap{3} \left( \total \left( {u}_i \right) \right)  \com{\cmap{3,1}\left( \total \left( x_{a_{i_1}}^{(3)}x_{a_{i_2}}^{(3)} \right) \right)}
    {\cmap{3,2}\left( \total \left( x_{a_{i_3}}^{(3)}x_{a_{i_4}}^{(3)}\right) \right)} \cmap{3} \left( \total \left(  u_i^{-1}\right) \right) .
\end{align}
We have 
\begin{align}
    \proj{3}(w'''') &= \prod_i  \total \left(u_i \right) \com{\total \left(x_{a_{i_1}}^{(3)}x_{a_{i_2}}^{(3)}\right)}{\total \left(x_{a_{i_3}}^{(3)}x_{a_{i_4}}^{(3)}\right)} \total \left(u_i^{-1}\right) = \proj{3}(w''') \label{eq:final_proj3}.
\end{align}
\Cref{eq:final2cancel} gives 
\begin{align}
    &\proj{2}(w'''') \\
    &\hspace{15pt} = \prod_i  \proj{2} \left( \cmap{3} \left( \total \left( {u}_i \right) \right)  \right) \com{\proj{2} \left( \cmap{3,1}\left( \total \left( x_{a_{i_1}}^{(3)}x_{a_{i_2}}^{(3)} \right) \right) \right)}
    {\proj{2} \left( \cmap{3,2}\left( \total \left( x_{a_{i_3}}^{(3)}x_{a_{i_4}}^{(3)}\right) \right) \right)} \proj{2} \left( \cmap{3} \left( \total \left( u_i^{-1}\right) \right) \right)\\
    &\hspace{15pt} = \prod_i \proj{2}\left( \cmap{3} \left( \total \left(u_i\right) \right) \right) \com{\proj{2}\left(\cmap{3,1} \left( \total \left(x_{a_{i_1}}^{(3)}x_{a_{i_2}}^{(3)}\right)\right) \right)}{\id} \proj{2}\left( \cmap{3} \left( \total \left(u_i\right) \right) \right)^{-1}  = \id \label{eq:final_inv_id}
\end{align}
A similar argument using \Cref{eq:final1cancel} shows $\proj{1}(w'''') = \id$.
Finally, elements in the image of $\proj{\sigma}$ commute with each other (by an argument similar to the proof of \Cref{lem:noSigmaInfs}) hence
\begin{align}
    &\proj{\sigma}\left(w'''' \right)  \\
    &\hspace{15pt} = \prod_i \proj{\sigma}\left( \cmap{3} \left( \total \left( {u}_i \right) \right) \right) \proj{\sigma}\left( \com{\cmap{3,1}\left( \total \left( x_{a_{i_1}}^{(3)}x_{a_{i_2}}^{(3)} \right) \right)}
    {\cmap{3,2}\left( \total \left( x_{a_{i_3}}^{(3)}x_{a_{i_4}}^{(3)}\right) \right)} \right) \proj{\sigma} \left( \cmap{3} \left( \total \left( u_i^{-1}\right) \right) \right) \\
    &\hspace{15pt} = \prod_i \proj{\sigma}\left( \cmap{3} \left( \total \left( {u}_i \right) \right) \right)
    \com{ \proj{\sigma}\left( \cmap{3,1}\left( \total \left( x_{a_{i_1}}^{(3)}x_{a_{i_2}}^{(3)} \right) \right) \right) } {\proj{\sigma} \left( \cmap{3,2}\left( \total \left( x_{a_{i_3}}^{(3)}x_{a_{i_4}}^{(3)}\right) \right)\right)}
    \proj{\sigma} \left( \cmap{3} \left( \total \left( u_i^{-1}\right) \right) \right)  \\
    &\hspace{15pt} = \prod_i \proj{\sigma}\left( \cmap{3} \left( \total \left( {u}_i \right) \right) \right)
    \;\; \proj{\sigma} \left( \cmap{3} \left( \total \left( u_i^{-1}\right) \right) \right)  = 1. \label{eq:noSigmaInPenul}
\end{align}

To put this all together and complete the proof, consider the word $w''' w''''^{-1}$. Using equations \cref{eq:final_inv_id,eq:penul_cancel}
\begin{align}
    \proj{2}(w''' w''''^{-1}) &= \proj{2}(w''') \proj{2}( w''''^{-1}) = 1 
\end{align}
with a similar argument giving 
\begin{align}
    \proj{1}(w''' w''''^{-1}) &= \proj{1}(w''') \proj{1}( w'''')^{-1} = 1. 
\end{align}
\Cref{eq:final_proj3} gives
\begin{align}
    \proj{3}(w''' w''''^{-1}) &= \proj{3}(w''')\proj{3}(w'''')^{-1} \\
    &= \proj{3}(w''')\proj{3}(w''')^{-1} = 1.
\end{align}
Finally, \Cref{eq:noSigmaInPenul}, \Cref{lem:noSigmaInfs}, and \Cref{cor:rightInvId} 
give
\begin{align}
  \proj{\sigma}(w''' w''''^{-1}) &= \proj{\sigma}(w''')\\
  &= \proj{\sigma}(w'' \cmap{3,2}\left(\proj{3}(w'')\right)^{-1} f_2(\proj{3}\left(w''\right))) \\
  &= \proj{\sigma}(w'') \\
  &= \proj{\sigma}\left( w' \cmap{3,1}\left(\proj{3}(w')\right)^{-1} f_1(\proj{3}\left(w'\right)) \right) \\
  &= \proj{\sigma}\left( w' \right) \\
  &= \sigma.
\end{align}
We conclude $\sigma \in H^E$ and thus the proof is complete. 
\end{proof}

\begin{appendix}

\section{Properties of \texorpdfstring{$K$}{K} and its Interactions}

Here we prove several small facts used in the proof of 
\Cref{thm:K_modding}
as well as
some which add perspective
on $K$.

\subsection{Properties of \texorpdfstring{$K$}{K}}

\begin{lem} \label{lem:even_com_in_K}
Let $u,v$ be two even length words in $G_\alpha$. Then $\com{u}{v} \in K$. 
\end{lem}

\begin{proof}
Let $l(u)$ denote the length of $u$, with $l(v)$ defined similarly. Define $L = l(u) + l(v)$. We prove by induction on $L$. 

When $L = 4$, $u$ and $v$ must both have length 2, hence $\com{u}{v}$ is a generator of $K$. Then the result is immediate. 

Otherwise, we must have that either $l(u)$ or $l(v)$ is greater than 2. For now we assume $l(v) > 2$. Then we can write 
\begin{align}
    v = v' v''
\end{align}
with $v'$ and $v''$ both even length words. Note that 
\begin{align}
    l(v') + l(v'') = l(v)
\end{align}
so $v'$ and $v''$ both have length less than $v$. Then we can write 
\begin{align}
    \com{u}{v} &= \com{u}{v'v''} \\
    &= \com{u}{v'}v'^{-1}\com{u}{v''}v'
\end{align}
where we have used the commutator identity
\begin{align}
    \com{x}{yz} &= \com{x}{y}y^{-1}\com{x}{z}y
\end{align}
on the second line. $l(u) + l(v')$ and $l(u) + l(v'')$ are both less than $L$, so by the induction hypothesis we have $\com{u}{v'}$ and $\com{u}{v''}$ are both in $K$. Since $K$ is normal, that also implies 
\begin{align}
    v'^{-1}\com{u}{v''}v' \in K, 
\end{align}
and since $K$ is a group 
\begin{align}
    \com{x}{y}y^{-1}\com{x}{z}y \in K.
\end{align}
The proof when $l(u) > 2$ is almost identical, except we use the commutator identity 
\begin{align}
\com{xy}{z} &= y^{-1}\com{x}{z}y\com{y}{z}
\end{align}
\end{proof}



\def\ee{{\phi}}
\def\ees{{\phi}^*}

\def\smu{{S^{must}}}

\def\sok{{S^{ok}}}
\def\s2ok{{S^{2ok}}}

\def\snot{{ S^{not}  }  }

\def\eps{{\varepsilon}}

\def\cC{{\mathcal C}}

\def\cE{{\mathcal E}}

\def\cL{{\mathcal L}}

\def\cT{{\mathcal T}}

\def\cZ{{\mathcal Z}}

\def\etass{\eta_{poss} }
\def\etano{\eta_{not} }

\def\ben{\begin{enumerate} }
\def\een{\end{enumerate} }

\def\sK{{\sim_K}}
\def\tm{{\tilde m}}

\subsubsection{Canonical form for monomials mod K}

Consider the game group  $G$ is defined for $k$ players and let $\sim_K$ denote the equivalence relation on $G$ defined by modding out by $K$.
In this subsection we shall write down a canonical selection from the equivalence classes. This is not used in the proofs here, but might be in other proofs and it is certainly
useful in  computer experiments.
While $G$ is defined for $k$ players modding out by $K$ acts independently on  the variables
$x_j^{(\alpha)} \ \ j=1, \dots, n$ associated 
with each player $\alpha$. Thus, without loss of generality we can take $k=1$.
Also $G$ contains $\sigma$ but we shall ignore it, since $\sigma$ has no impact on the canonical form.

\bigskip

The core observation is  the following lemma.

\begin{lem}
\label{lem:canK}
Suppose $G$ is the game group of a 1-XOR quantum qame.
Monomials of the form
$$ w abcd q \quad  and  \quad
w cbad q \quad and \quad
 w adcb q$$
are all equal mod $K$. 
Here $a,b,c,d$
are generators of the $G$ and 
$w$ and $q$ are
arbitrary monomials.

For degree 3 or more monomials
this immediately implies that interchanging any two even
position  variables 
or any two odd position variables in a monomial $m$
produces a monomial $\tm$ with $m \sK{} \tm$.
\end{lem}

\proof We first show
$ abcd \ \sK \  adcb $ by noting
\begin{align}
\left(adcb\right)^{-1} abcd  = 
bcd \ bcd  =
bc \; dc \ cb \; cd \ \sK \ 1. 
\end{align}
where the last equation is true by definition of $K$. The proof  that the  first and third monomials are equivalent goes similarly. 
 
 If $m$ has degree 3 write it as $abc$, then the
 property just proved for  degree 4 gives
\begin{align}
    abc \sK \  abc xx  \ \iff 
 \  cba xx \sK \  cba 
\end{align}
 as claimed.
\qed

\bigskip

Given an ordering on the generators of $G$,
a canonical form
of a
monomial $m$ is seen easily from
the lemma.
We describe  it in terms of an algorithm.

{\bf {Algorithm $K,Q$}}
\ben
\item
Find  its even (resp, odd)  part, namely 
 the monomial whose entries are the variables
in the even (resp odd)  locations of $m$.
For example: 
take $m=zgabcdfzz$, then
$$  
even[m]:= gbdfz \qquad odd[m]:= zacez 
$$
\item 
Select a variable, say  $v$, and count how many times, $e$, it appears in
$even[m]$ and $o$ times in  $odd[m]$.\\
If $o \leq e$, then remove all variables 
$v$ from the list $odd[m]$ and also remove $o$
of the $v$'s from  $even[m]$. 
\ \ \ 
If $e \leq o$, then remove all $v$ from the list $even[m]$
and also remove $e$ of the $v$'s from  $odd[m]$.
The order of removal does not matter.
Do this for all variables (not just $v$) to get 
$evQ[m]$ and $oddQ[m]$.

Example revisited: take
$G$ to have generators equal to the alphabet $a, \dots, z$
with each generator having square equal to 1.
$e= 1$ and $o=2$
for the variable $z$. 
So
$evQ[m]=gbdf $ and $oddQ[m]= zace$.

\item
Order both lists. \ \ 
$alph[even] := bdfg,
\ \  alph[odd] := acez $
\item
Recombine these words to make one word.
$
canon[m]:= abcdefzg
$
\qed
\een

Application of Lemma \ref{lem:canK} proves
the Algorithm succeeds as is formalized  by the following.
\begin{prop}
For monomials  of degree $\geq 3$, we have that 
$ canon[m]$ is uniquely determined
and 
$m \sK \; canon[m]$.
That is, $canon[m] $ is a canonical form for $m$.

\end{prop}

\subsection{The interaction of \texorpdfstring{$\proj{\sigma}$}{phi sigma} and \texorpdfstring{$\proj{\alpha}$}{phi alpha} with \texorpdfstring{$K$}{K}}

\begin{lem} \label{lem:noSigmaInK}
For any $k \in K$, 
\begin{align}
    \proj{\sigma}(k) = 1. 
\end{align}
\end{lem}

\begin{proof}
We can write
\begin{align}
k = \prod_{i} u_i \com{x_{a_i}^{\alpha_i}x_{b_i}^{\alpha_i}}{x_{c_i}^{\alpha_i}x_{d_i}^{\alpha_i}} u_i^{-1}
\end{align}
Then,
\begin{align}
    \proj{\sigma}(k) &= \proj{\sigma}\left(\prod_{i} u_i \com{x_{a_i}^{\alpha_i}x_{b_i}^{\alpha_i}}{x_{c_i}^{\alpha_i}x_{d_i}^{\alpha_i}} u_i^{-1} \right) \\
    &= \prod_i \proj{\sigma}(u_i) \com{\proj{\sigma}\left(x_{a_i}^{\alpha_i}x_{b_i}^{\alpha_i}\right)}{\proj{\sigma}\left(x_{c_i}^{\alpha_i}x_{d_i}^{\alpha_i}\right)} \proj{\sigma}\left(u_i^{-1}\right) \\
    &= \prod_i \proj{\sigma}(u_i)\proj{\sigma}\left(u_i^{-1}\right) = 1
\end{align}
where we used that $\Im(\proj{\sigma}) = \{\sigma,1\}$ is a commutative group to show the commutator terms were the identity. 
\end{proof}

\begin{lem} \label{lem:projKtoK}
For any $k \in K$, and $\alpha \in \{1,2,3\}$:
\begin{align}
    \proj{\alpha}(k) \in K.
\end{align}
\end{lem}

\begin{proof}
Define the set $C$ to be all commutators of pairs, that is 
\begin{align}
    C = \left\{\com{x_i^\alpha x_j^\alpha}{x_k^\alpha x_l^\alpha} : i,j,k,l \in [n], \alpha \in [3] \right\}.
\end{align}
Recall that K was defied to be the normal closure of $C$ in $G^E$, that is:  
\begin{align}
    K = \angles{C}^{G^E}. 
\end{align}
We first show that 
\begin{align}
    \proj{\alpha}\left(c\right) \in C
\end{align}
for all $c \in C$. To see this, note
\begin{align}
    \proj{\alpha}\left( \com{x_i^{(\beta)} x_j^{(\beta), }}{x_k^{(\beta)} x_l^{(\beta)}} \right) = \com{x_i^{(\alpha)} x_j^{(\alpha), }}{x_k^{(\alpha)} x_l^{(\alpha)}} \in K 
\end{align}
for $\alpha = \beta$, and 
\begin{align}
\proj{\alpha}\left( \com{x_i^{(\beta)} x_j^{(\beta), }}{x_k^{(\beta)} x_l^{(\beta)}} \right) = 1 \in K. 
\end{align}
for $\alpha \neq \beta$. 

Then, since $\proj{\alpha}$ is a homomorphism mapping $G^E \rightarrow G_3^{E}$, and $\proj{\alpha}(C) \subset C$, we have 
\begin{align}
    \proj{\alpha} : \angles{C}^{G^E} \hookrightarrow \angles{C}^{G_3^E} \subset K. 
\end{align}
The result follows. 
\end{proof}

\subsection{Equivalence between a PREF and \texorpdfstring{$\sigma \in H \pmod{K}$}{sigma in H (mod K)}}
\label{app:PREF-sigmaInH}

In \cite{watts2018algorithms} an object called a \textit{parity refutation} was defined. A (paraphrased) version of that definition using the language of \Cref{subsec:groups} is repeated here. First, we define a \df{parity preserving permutation}.

\begin{defn} \label{def:pper}
A parity preserving permutation of a sequence of generators (written here as a product)
\begin{align}
    x_{a_1}^{(1)}x_{a_2}^{(1)} ... x_{a_{l_1}}^{(1)}x_{b_1}^{(2)}...x_{b_{l_2}}^{(2)}x_{c_{1}}^{(3)}...x_{c_{l_3}}^{(3)} \sigma^{s}
\end{align} 
is a permutation $P$ which satisfies 
\begin{align}
    P(x_{a_i}^{(1)}) = x_{a_j}^{(1)}
\end{align}
with $i = j \pmod{2}$, similar restrictions for $P(x_{b_{i'}}^{(2)})$ and $P(x_{c_{i''}}^{(3)})$ and the condition $P(\sigma) = \sigma$.
\end{defn}

An equivalent definition of parity preserving permutations which will be useful to use later are permutations $P$ which can be decomposed into products of transpositions of the form $\pi_{j,j+2}^{(\alpha)}$ with $\alpha \in [3]$ and 
\begin{align}
    \pi^{(\alpha)}_{j,j+2} \left(x^{(\alpha)}_{a_1}x^{(\alpha)}_{a_2} ... x^{(\alpha)}_{a_j}x^{(\alpha)}_{a_{j+1}}x^{(\alpha)}_{a_{j+2}} ... x^{(\alpha)}_{a_l} \right) = x^{(\alpha)}_{a_1}x^{(\alpha)}_{a_2} ... x^{(\alpha)}_{a_{j+2}}x^{(\alpha)}_{a_{j+1}}x^{(\alpha)}_{a_{j}} ... x^{(\alpha)}_{a_l} 
\end{align}

Parity preserving permutations can be used to define an equivalence relation on the words $g \in G$

\begin{defn}
Two words $g_1, g_2 \in G$ are \df{parity permutation equivalent}, written $g_1 \sim_p g_2$, if there is a sequence of generators
\begin{align}
x_{a_1}^{(1)}x_{a_2}^{(1)} ... x_{a_{l_1}}^{(1)}x_{b_1}^{(2)}...x_{b_{l_2}}^{(2)}x_{c_{1}}^{(3)}...x_{c_{l_3}}^{(3)}\sigma^{s} = g_1
\end{align}  
and a \df{parity preserving permutation} $P$ acting on that sequence of generators  satisfying
\begin{align}
    P(x_{a_1}^{(1)}x_{a_2}^{(1)} ... x_{a_{l_1}}^{(1)}x_{b_1}^{(2)}...x_{b_{l_2}}^{(2)}x_{c_{1}}^{(3)}...x_{c_{l_3}}^{(3)}\sigma^{s}) = g_2
\end{align}
\end{defn}
Routine calculation (given in \cite{watts2018algorithms}) shows $\sim_p$ is an equivalence relation on elements of $G$. Finally, we define a \df{parity refutation} (PREF). 
\begin{defn} \label{def:PREF}
A sequence of clauses 
$h_{r_1}, h_{r_2}, ..., h_{r_l}$
is called a \df{parity refutation} if $h_{r_1} h_{r_2} ... h_{r_l} \sim_p \sigma$. 
\end{defn}
Existence of a parity refutation is exactly equivalent to a word $\sigma \in H \pmod{K}$, as we show in the following theorem. (Actually, a stronger statement is true: the equivalence relation $\sim_p$ is exactly the same as the equivalence relation on $G$ induced by modding out by $K$. Small modifications to the proof below give that result.)
\begin{thm}
A sequence of clauses $h_{r_1} h_{r_2} ... h_{r_l}$ is a parity refutation iff the word $h_{r_1} h_{r_2} ... h_{r_l} \in H$ obtained by multiplying the clauses together satisfies
\begin{align}
    h_{r_1} h_{r_2} ... h_{r_l} = \sigma \pmod{K}
\end{align}
\end{thm}
\begin{proof}
Both directions of the proof are nontrivial. We first show that if a sequence of clauses $h_{r_1} h_{r_2} ... h_{r_l}$ forms a parity refutation then $h_{r_1} h_{r_2} ... h_{r_l} = \sigma \pmod{K}$. 
Recall that any parity preserving permutation $P$ can be decomposed into transpositions of the form $\pi^{(\alpha)}_{j,j+2}$, where 
\begin{align}
    \pi^{(\alpha)}_{j,j+2} \left(x^{(\alpha)}_{a_1}x^{(\alpha)}_{a_2} ... x^{(\alpha)}_{a_j}x^{(\alpha)}_{a_{j+1}}x^{(\alpha)}_{a_{j+2}} ... x^{(\alpha)}_{a_l} \right) = x^{(\alpha)}_{a_1}x^{(\alpha)}_{a_2} ... x^{(\alpha)}_{a_{j+2}}x^{(\alpha)}_{a_{j+1}}x^{(\alpha)}_{a_{j}} ... x^{(\alpha)}_{a_l} 
\end{align}
But we also have 
\begin{align}
    K \ni \com{x^{(\alpha)}_{a_{j+2}}x^{(\alpha)}_{a_{j+1}}}{x^{(\alpha)}_{a_{j}}x^{(\alpha)}_{a_{j+1}}} = x^{(\alpha)}_{a_{j+2}}x^{(\alpha)}_{a_{j+1}}x^{(\alpha)}_{a_{j}}x^{(\alpha)}_{a_{j+2}}x^{(\alpha)}_{a_{j+1}}x^{(\alpha)}_{a_{j}}
\end{align}
hence 
\begin{align}
    x^{(\alpha)}_{a_j}x^{(\alpha)}_{a_{j+1}}x^{(\alpha)}_{a_{j+2}} &= x^{(\alpha)}_{a_j}x^{(\alpha)}_{a_{j+1}}x^{(\alpha)}_{a_{j+2}} x^{(\alpha)}_{a_{j+2}}x^{(\alpha)}_{a_{j+1}}x^{(\alpha)}_{a_{j}}x^{(\alpha)}_{a_{j+2}}x^{(\alpha)}_{a_{j+1}}x^{(\alpha)}_{a_{j}} \pmod{K} \\
    &= x^{(\alpha)}_{a_{j+2}}x^{(\alpha)}_{a_{j+1}}x^{(\alpha)}_{a_{j}} \pmod{K}.
\end{align}
As a consequence, we also have 
\begin{align}
    x^{(\alpha)}_{a_1}x^{(\alpha)}_{a_2} ... x^{(\alpha)}_{a_j}x^{(\alpha)}_{a_{j+1}}x^{(\alpha)}_{a_{j+2}} ... x^{(\alpha)}_{a_l} &= x^{(\alpha)}_{a_1}x^{(\alpha)}_{a_2} ... x^{(\alpha)}_{a_{j+2}}x^{(\alpha)}_{a_{j+1}}x^{(\alpha)}_{a_{j}} ... x^{(\alpha)}_{a_l}  \pmod{K} \\
    &= \pi^{(\alpha)}_{j,j+2} \left(x^{(\alpha)}_{a_1}x^{(\alpha)}_{a_2} ... x^{(\alpha)}_{a_j}x^{(\alpha)}_{a_{j+1}}x^{(\alpha)}_{a_{j+2}} ... x^{(\alpha)}_{a_l} \right) 
    \pmod{K}. \label{eq:parityTransposeInK}
\end{align}
Since the word $x_{a_1}^{(\alpha)}x_{a_2}^{(\alpha)} ... x_{a_l}^{(\alpha)}$ was arbitrary and we could decompose $P$ into products of transpositions of the form $\pi_{j,j+2}^{(\alpha)}$ we conclude 
\begin{align}
    h_{r_1}h_{r_2}...h_{r_l} &= x_{a_{r_1}}^{(1)}x_{a_{r_2}}^{(1)} ... x_{c_{r_l}}^{(3)} \sigma^{s_{r_1} + s_{r_2} + ... s_{r_l}} \\
    &= P(x_{a_{r_1}}^{(1)}x_{a_{r_2}}^{(1)} ... x_{c_{r_l}}^{(3)}\sigma^{s_{r_1} + s_{r_2} + ... s_{r_l}}) \pmod{K} \label{eq:PinK} \\ 
    &= \sigma \pmod{K} \label{eq:parityRefImpliesSigmaInH}
\end{align}
Where line \labelcref{eq:PinK} follows from equation \labelcref{eq:parityTransposeInK} and line \labelcref{eq:parityRefImpliesSigmaInH} follows from the definition of a parity refutation. This completes the proof in one direction. 

It remains to show that if $h_{r_1}h_{r_2} ... h_{r_l} = \sigma \pmod{K}$ we also have  $h_{r_1}h_{r_2} ... h_{r_l} \sim_p \sigma$. Our first step is to note that the equivalence relation $\sim_p$ respects multiplication by construction -- that is we have $g_1 \sim_p g_2$ and $g_3 \sim_p g_4$ implies $g_1 g_2 \sim_p g_3 g_4$. We next note that for any set of generators $x^{(\alpha)}_i, x^{(\alpha)}_j, x^{(\alpha)}_s, x^{(\alpha)}_t$ and word $w \in G$ we have 
\begin{align}
    w \com{x^{(\alpha)}_ix^{(\alpha)}_j}{x^{(\alpha)}_s x^{(\alpha)}_t} w^{-1} &= w x^{(\alpha)}_ix^{(\alpha)}_jx^{(\alpha)}_s x^{(\alpha)}_t \left(x^{(\alpha)}_ix^{(\alpha)}_j\right)^{-1} \left( x^{(\alpha)}_s x^{(\alpha)}_t \right)^{-1} w^{-1} \\
    &\sim_p 
    w w^{-1}
    x^{(\alpha)}_ix^{(\alpha)}_j
    x^{(\alpha)}_s x^{(\alpha)}_t 
    \left(x^{(\alpha)}_ix^{(\alpha)}_j\right)^{-1}
    \left( x^{(\alpha)}_s x^{(\alpha)}_t \right)^{-1}\\
    &\sim_p 
    w w^{-1}
    x^{(\alpha)}_ix^{(\alpha)}_j
    x^{(\alpha)}_s x^{(\alpha)}_t 
    \left( x^{(\alpha)}_s x^{(\alpha)}_t \right)^{-1}
    \left(x^{(\alpha)}_ix^{(\alpha)}_j\right)^{-1}  = 1
\end{align}
since the permutations moving $w^{-1}$ to the other side of $\com{x^{(\alpha)}_ix^{(\alpha)}_j}{x^{(\alpha)}_s x^{(\alpha)}_t}$ and swapping $\left(x^{(\alpha)}_ix^{(\alpha)}_j\right)^{-1}$ and $\left( x^{(\alpha)}_s x^{(\alpha)}_t \right)^{-1}$ are both parity preserving permutations. It follows that for any $k \in K$, $k \sim_p 1$. Then, if
$ h_{r_1} h_{r_2} ... h_{r_l} = \sigma \pmod{K}$ we must also have $h_{r_1} h_{r_2} ... h_{r_l} k = \sigma$ for some $k \in K$, and hence 
\begin{align}
    h_{r_1} ... h_{r_l} = h_{r_1} ... h_{r_l} k k^{-1} 
    \sim_p \sigma (1) = \sigma
\end{align}
where we used that $\sim_p$ respected multiplication $h_{r_1} ... h_{r_l} k = \sigma$ and $k^{-1} \sim_p 1$ to obtain the equivalence. This completes the proof.

\subsection{MERP as a mod \texorpdfstring{$K$}{K} strategy}
\label{subsec:MerpModK}
Recall from 
\Cref{def:MERP}
the MERP strategies are a nice class of finite dimensional strategies which generalize the GHZ strategy. Here we give a direct proof that 
MERP strategies are annihilated by the $K$
relations.

\begin{thm} \label{thm:MerpValidModK}
The MERP strategy observables respect the mod $K$ relations. That is,
\begin{align}
    \com{X_i^{(\alpha)}X_{i'}^{(\alpha)}}{X_j^{(\alpha)}X_{j'}^{(\alpha)}} = 1
\end{align}
for all $\alpha, i, i', j, j'$ if the $X_i^{(\alpha)}$ are MERP strategy observables as defined above. 
\end{thm}

\begin{proof} The proof is 
computational, with some tricks about Pauli matrices. Let all the $X_i^{(\alpha)}$ be MERP strategy observables and note, for all indices 
\begin{align} 
    \com{X_i^{(\alpha)}X_{i'}^{(\alpha)}}{X_j^{(\alpha)}X_{j'}^{(\alpha)}}
    = I^{\otimes ((\alpha-1)} \otimes
    \com{M(\theta_i^{(\alpha)})M(\theta_{i'}^{(\alpha)})}{M(\theta_{j}^{(\alpha)}M(\theta_{j'})^{(\alpha)}))}
    \otimes I^{\otimes (k - \alpha)}
\end{align}
by the tensor product structure. Now, the Pauli matrices anti-commute, so 
\begin{align}
    \sigma_x \sigma_z = - \sigma_z \sigma_x
\end{align}
and 
\begin{align}
    \sigma_x \exp(i\theta \sigma_z) =  \exp(-i\theta \sigma_z) \sigma_x 
    \qquad 
     \exp(i\theta \sigma_z)\sigma_x =  \sigma_x \exp(-i\theta \sigma_z)  
    \label{eq:PauliExpCom}
\end{align}
where the later equalities can be shown by the Taylor series expansion of $\exp(i\theta \sigma_z)$. 
This lets us write our MERP strategy observables in a slightly simplier form, since 
\begin{align}
    M(\theta_i^{(\alpha)}) &= \merp{\theta_i^{(\alpha)}} \\
    &= \merptwo{\theta_i^{(\alpha)}}
\end{align}
As a more more significant application of \Cref{eq:PauliExpCom} we can show MERP strategy observables switch the sign on $\theta_i^{(\alpha)}$ when they commute since
\begin{align}
    M(\theta_i^{(\alpha)})M(\theta_{j}^{(\alpha)}) 
    &= \merptwo{\theta_i^{(\alpha)}}\merptwo{\theta_j^{(\alpha)}} \\
    &=  \sigma_x \exp(- 2 i\theta_i^{(\alpha)} \sigma_z)  \exp(2i \theta_j^{(\alpha)} \sigma_z)    \sigma_x
    \\
    &=  \sigma_x
    \exp(2i \theta_j^{(\alpha)} \sigma_z) 
    \exp(- 2 i\theta_i^{(\alpha)} \sigma_z)  
       \sigma_x
    \\
    &= \merptwom{\theta_j^{(\alpha)}}\merptwom{\theta_i^{(\alpha)}} \\
    &=  M(-\theta_j^{(\alpha)})M(-\theta_{i}^{(\alpha)}) \label{eq:MERPObsCom}
\end{align}
using \Cref{eq:PauliExpCom} on the second line. Now, repeatedly applying \Cref{eq:MERPObsCom} gives
\begin{align}
    M(\theta_i^{(\alpha)})M(\theta_{i'}^{(\alpha)})M(\theta_{j}^{(\alpha)}) M(\theta_{j'}^{(\alpha)}) &= M(\theta_{j}^{(\alpha)}) M(-\theta_i^{(\alpha)})M(-\theta_{i'}^{(\alpha)}) M(\theta_{j'}^{(\alpha)}) \\
    &= M(\theta_{j}^{(\alpha)})M(\theta_{j'}^{(\alpha)}) M(\theta_i^{(\alpha)})M(\theta_{i'}^{(\alpha)}) 
\end{align}
Hence 
\begin{align}
    \com{M(\theta_i^{(\alpha)})M(\theta_{i'}^{(\alpha)})}{M(\theta_j^{(\alpha)})M(\theta_{j'}^{(\alpha)})} = 1
\end{align}
and the result follows. 
\end{proof}
\end{proof}

\subsection{Some members of \texorpdfstring{$K \cap H^E$}{K intersect HE} and possible \texorpdfstring{$k$XOR}{kXOR} generalizations}
 
\label{app:K_intuition_and_K_players}

Here we give some intuition for dealing with  the subgroup 
$K$ in relation to 3XOR. 
A major component of our 3XOR analysis has been showing that 
the special word $\sigma$ is in $K \cap H^E$. This was difficult. For perspective, we ask a simpler question:
 is the intersection $K \cap H^E$ necessarily nonempty for a 3XOR game? The next lemma says yes.
 
\begin{lem} 
\label{lem:KHE}
Suppose a 3XOR  game is nontrivial in the sense that it contains at least two clauses which contain the same generator for $x_i^{(1)}$ for player 1, also two such clauses for player 2, then $K \cap H^E$ is not empty, indeed
at least \emph{some} generators of $K$ are necessarily contained in $H^E$.
\end{lem}

\proof 
Consider a pair of clauses $h_1, h_2 \in S$ corresponding to question vectors which send the same question to the first player, so
$ h_1 = x_{a_1}^{(1)}x_{b_1}^{(2)}x_{c_1}^{(3)}\sigma^{s_1}$, $h_2 = x_{a_2}^{(1)}x_{b_2}^{(2)}x_{c_2}^{(3)}\sigma^{s_2}
$ and $a_1 = a_2$. Similarly, let clauses $h_3, h_4$ be clauses which agree on the question sent to the second player so $x_{b_3} = x_{b_4}$.\footnote{These pairs of clauses don't need to exist, but XOR games where each question is asked only once are particularly simple, with $\omega = 1$, so we assume we are not in this case.} We then consider the commutator 
\begin{align}
    \com{h_1h_2}{h_3h_4} &= \com{x_{a_1}^{(1)}x_{a_2}^{(1)}}{x_{a_3}^{(1)}x_{a_4}^{(1)}}\com{x_{b_1}^{(2)}x_{b_2}^{(2)}}{x_{b_3}^{(2)}x_{b_4}^{(2)}}\com{x_{c_1}^{(3)}x_{c_2}^{(3)}}{x_{c_3}^{(3)}x_{c_4}^{(3)}}\com{\sigma^{s_1 + s_2}}{\sigma^{s_3 + s_4}} \\
    &= \com{\id}{x_{a_3}^{(1)}x_{a_4}^{(1)}}\com{x_{b_1}^{(2)}x_{b_2}^{(2)}}{\id}\com{x_{c_1}^{(3)}x_{c_2}^{(3)}}{x_{c_3}^{(3)}x_{c_4}^{(3)}}\com{\sigma^{s_1 + s_2}}{\sigma^{s_3 + s_4}} \\
    &= \com{x_{c_1}^{(3)}x_{c_2}^{(3)}}{x_{c_3}^{(3)}x_{c_4}^{(3)}}, 
\end{align}
where we have used the fact that group elements corresponding to different players commute on the first line, that $x_{a_1}^{(1)} x_{a_2}^{(1)} = \left(x_{a_1}^{(1)}\right)^2 = \id$ on the second line, and that $\com{w}{\id} = \id$ for any $w$ and $\sigma$ commutes with anything on the third. 
The conclusion is that 
\begin{align}
    \com{x_{c_1}^{(3)}x_{c_2}^{(3)}}{x_{c_3}^{(3)}x_{c_4}^{(3)}} = \com{h_1h_2}{h_3h_4} \in H^E.
\end{align}
We have just proved 
the set of all 
commutators of pairs of generators $x_i^{(\alpha)}$ which lie in $H^E$ is necessarily nonempty. Thus $K \cap H^E$ is nonempty as well. 
\qed 
\\

Our vague wish is that $K \cap H^E$ be large, so we point out that the same argument as above with any two pairs of clauses that cancel on two different players shows even  more generators are in $K \cap H^E$.

\subsubsection{Possible \texorpdfstring{$k$XOR}{kXOR} analogues of the subgroup \texorpdfstring{$K$}{K}}

Now we discuss possible $k$ player generalizations of the arguments in this paper. To generalize the arguments of this paper beyond 3 players, we would require a $k$-player analogue of \Cref{thm:K_modding}. This would be the statement that, for every clause group $H^E$ associated with a $k$XOR game and some normal subgroup $K' \triangleleft G^E$, 
\begin{align}
    \sigma \in H^E \pmod{K'} \Longleftrightarrow \sigma \in H^E.  
\end{align}
The ``123 Game'' presented in \cite{watts2018algorithms} shows that the above statement is false for $6$ player games when $K' = K$. (The ``123 Game'' is a 6 player game with a perfect commuting operator strategy, meaning $\sigma \notin H^E$, but no perfect MERP strategy, meaning $\sigma \notin H^E \pmod{K}$). 

The proof of \Cref{thm:K_modding} is involved, and it is unclear how it would generalize beyond the 3 player case. However the intuition presented in \Cref{lem:KHE} does generalize naturally to $k$-players. Following the same logic as used in the proof of \Cref{lem:KHE} we see that for a non-trivial $k$ player XOR game, elements of the form 
\begin{align}
    [...[[[x_{c_1}^{(1)}x_{c_2}^{(1)},x_{c_3}^{(1)}x_{c_4}^{(1)}],x_{c_5}^{(1)}x_{c_6}^{(1)}], ... ], x_{c_{2k-1}}^{(1)}x_{c_{2k}}^{(1)}] 
\end{align}
are necessarily contained in the group $H^E$. This observation encourages the speculation that a $k$ player analogue of \Cref{thm:K_modding} may hold with the subgroup $K'$ equal to the $k$-th entry in the lower central series of $H^E$, i.e. the subgroup of $H^E$ generated by elements of the form 
\begin{align}
    [...[[[h_1,h_2],h_3],...],h_k].
\end{align}
However, this intuition falls well short of proving the desired result.

\section{Subgroup Membership}
\label{app:AlgComb}

\begin{thm} \label{thm:SMPabelian}
The subgroup membership problem is solvable in polynomial time for any finitely generated abelian group.\footnote{Stronger versions of this statement are also true. In particular, the subgroup membership problem is solvable for any finitely generated metabelian group\cite{romanovskii1974some} (meaning commutators of commutators vanish) or finitely generated nilpotent group\cite{lohrey2013rational}. 
}
\end{thm}

\begin{proof}
It reduces to linear algebra over the integers. We can write all the relations in the group $G$ and generators of the subgroup $\tilde{G}$ as products of generators of $G$, raised to some power. When we multiply generators or apply a relation we just add or subtract the multiplicities of the relevant generators. So the subgroup membership problem just asks if a given vector (corresponding to the group element) is in the span of the vectors corresponding to the relations and subgroup generators. 
\end{proof}

\section{Declarations}

\subsection{Funding and Competing Interests}

\textbf{Financial Interests:} J.W. Helton thanks the Center for Mathematical Sciences  and Applications at Harvard for a stimulating stay leading to this collaboration and thanks to the NSF for its support through
 DMS1500835. A. Bene Watts was supported by NSF grant CCF-1729369. \\

\noindent \textbf{Non-Financial Interests:} None 
\subsection{Data Availability Statement}

Data sharing not applicable to this article as no datasets were generated or analysed during the current study.

\bibliographystyle{abbrv}
\bibliography{ref}

\begin{thebibliography}{10}

\bibitem{QKD1:barrett2005no}
J.~Barrett, L.~Hardy, and A.~Kent.
\newblock No signaling and quantum key distribution.
\newblock {\em Physical review letters}, 95(1):010503, 2005.

\bibitem{bell1964einstein}
J.~S. Bell.
\newblock On the einstein podolsky rosen paradox.
\newblock {\em Physics Physique Fizika}, 1(3):195, 1964.

\bibitem{brassard2005quantum}
G.~Brassard, A.~Broadbent, and A.~Tapp.
\newblock Quantum pseudo-telepathy.
\newblock {\em Foundations of Physics}, 35(11):1877--1907, 2005.

\bibitem{SDCbravyi2018quantum}
S.~Bravyi, D.~Gosset, and R.~K{\"o}nig.
\newblock Quantum advantage with shallow circuits.
\newblock {\em Science}, 362(6412):308--311, 2018.

\bibitem{briet2009multiplayer}
J.~Bri{\"e}t, H.~Buhrman, T.~Lee, and T.~Vidick.
\newblock Multipartite entanglement in xor games.
\newblock {\em Quantum Information and Computation}, 13(3-4):334--360, 2013.

\bibitem{briet2013explicit}
J.~Bri{\"e}t and T.~Vidick.
\newblock Explicit lower and upper bounds on the entangled value of multiplayer
  xor games.
\newblock {\em Communications in Mathematical Physics}, 321(1):181--207, 2013.

\bibitem{clauser1969proposed}
J.~F. Clauser, M.~A. Horne, A.~Shimony, and R.~A. Holt.
\newblock Proposed experiment to test local hidden-variable theories.
\newblock {\em Physical review letters}, 23(15):880, 1969.

\bibitem{cleve2017perfect}
R.~Cleve, L.~Liu, and W.~Slofstra.
\newblock Perfect commuting-operator strategies for linear system games.
\newblock {\em Journal of Mathematical Physics}, 58(1):012202, 2017.

\bibitem{cleve2014characterization}
R.~Cleve and R.~Mittal.
\newblock Characterization of binary constraint system games.
\newblock In {\em International Colloquium on Automata, Languages, and
  Programming}, pages 320--331. Springer, 2014.

\bibitem{coladangelo2018unconditional}
A.~Coladangelo and J.~Stark.
\newblock Unconditional separation of finite and infinite-dimensional quantum
  correlations.
\newblock {\em arXiv preprint arXiv:1804.05116}, 2018.

\bibitem{RGcolbeck2009quantum}
R.~Colbeck.
\newblock Quantum and relativistic protocols for secure multi-party
  computation.
\newblock {\em arXiv preprint arXiv:0911.3814}, 2009.

\bibitem{coudron2019complexity}
M.~Coudron and W.~Slofstra.
\newblock Complexity lower bounds for computing the approximately-commuting
  operator value of non-local games to high precision.
\newblock {\em arXiv preprint arXiv:1905.11635}, 2019.

\bibitem{doherty2008quantum}
A.~C. Doherty, Y.-C. Liang, B.~Toner, and S.~Wehner.
\newblock The quantum moment problem and bounds on entangled multi-prover
  games.
\newblock In {\em 2008 23rd Annual IEEE Conference on Computational
  Complexity}, pages 199--210. IEEE, 2008.

\bibitem{dykema2019non}
K.~Dykema, V.~I. Paulsen, and J.~Prakash.
\newblock Non-closure of the set of quantum correlations via graphs.
\newblock {\em Communications in Mathematical Physics}, 365(3):1125--1142,
  2019.

\bibitem{QKD2:ekert1991quantum}
A.~K. Ekert.
\newblock Quantum cryptography based on bell’s theorem.
\newblock {\em Physical review letters}, 67(6):661, 1991.

\bibitem{fritz2012tsirelson}
T.~Fritz.
\newblock Tsirelson's problem and kirchberg's conjecture.
\newblock {\em Reviews in Mathematical Physics}, 24(05):1250012, 2012.

\bibitem{greenberger1990bell}
D.~M. Greenberger, M.~A. Horne, A.~Shimony, and A.~Zeilinger.
\newblock Bell’s theorem without inequalities.
\newblock {\em American Journal of Physics}, 58(12):1131--1143, 1990.

\bibitem{SDCgrier2019interactive}
D.~Grier and L.~Schaeffer.
\newblock Interactive shallow clifford circuits: Quantum advantage against
  nc$^1$ and beyond.
\newblock In {\em Proceedings of the 52nd Annual ACM SIGACT Symposium on Theory
  of Computing}, pages 875--888, 2020.

\bibitem{haastad2001some}
J.~H{\aa}stad.
\newblock Some optimal inapproximability results.
\newblock {\em Journal of the ACM (JACM)}, 48(4):798--859, 2001.

\bibitem{helton2017algebras}
W.~Helton, K.~P. Meyer, V.~I. Paulsen, and M.~Satriano.
\newblock Algebras, synchronous games and chromatic numbers of graphs.
\newblock {\em arXiv preprint arXiv:1703.00960}, 2017.

\bibitem{hoyer2005quantum}
P.~H{\o}yer and R.~{\v{S}}palek.
\newblock Quantum fan-out is powerful.
\newblock {\em Theory of computing}, 1(1):81--103, 2005.

\bibitem{ji2020mip}
Z.~Ji, A.~Natarajan, T.~Vidick, J.~Wright, and H.~Yuen.
\newblock {MIP*= RE}.
\newblock {\em Communications of the ACM}, 64(11):131--138, 2021.

\bibitem{lohrey2013rational}
M.~Lohrey.
\newblock The rational subset membership problem for groups: a survey.
\newblock In {\em Groups St Andrews}, volume 422, pages 368--389, 2013.

\bibitem{mermin1990extreme}
N.~D. Mermin.
\newblock Extreme quantum entanglement in a superposition of macroscopically
  distinct states.
\newblock {\em Physical Review Letters}, 65(15):1838, 1990.

\bibitem{mihajlova1971occurrence}
K.~Mihajlova.
\newblock : The occurrence problem for direct products of groups.
\newblock {\em Journal of Symbolic Logic}, 36(3), 1971.

\bibitem{natarajan2019neexp}
A.~Natarajan and J.~Wright.
\newblock Neexp is contained in mip.
\newblock In {\em 2019 IEEE 60th Annual Symposium on Foundations of Computer
  Science (FOCS)}, pages 510--518. IEEE, 2019.

\bibitem{navascues2008convergent}
M.~Navascu{\'e}s, S.~Pironio, and A.~Ac{\'\i}n.
\newblock A convergent hierarchy of semidefinite programs characterizing the
  set of quantum correlations.
\newblock {\em New Journal of Physics}, 10(7):073013, 2008.

\bibitem{perez2008unbounded}
D.~P{\'e}rez-Garc{\'\i}a, M.~M. Wolf, C.~Palazuelos, I.~Villanueva, and
  M.~Junge.
\newblock Unbounded violation of tripartite bell inequalities.
\newblock {\em Communications in Mathematical Physics}, 279(2):455--486, 2008.

\bibitem{DQC:reichardt2013classical}
B.~W. Reichardt, F.~Unger, and U.~Vazirani.
\newblock Classical command of quantum systems.
\newblock {\em Nature}, 496(7446):456--460, 2013.

\bibitem{romanovskii1974some}
N.~Romanovskii.
\newblock Some algorithmic problems for solvable groups.
\newblock {\em Algebra and Logic}, 13(1):13--16, 1974.

\bibitem{scholz2008tsirelson}
V.~B. Scholz and R.~F. Werner.
\newblock Tsirelson's problem.
\newblock {\em arXiv preprint arXiv:0812.4305}, 2008.

\bibitem{slofstra2020tsirelson}
W.~Slofstra.
\newblock Tsirelson’s problem and an embedding theorem for groups arising
  from non-local games.
\newblock {\em Journal of the American Mathematical Society}, 33(1):1--56,
  2020.

\bibitem{tsirel1987quantum}
B.~S. Tsirel'son.
\newblock Quantum analogues of the {Bell} inequalities. {The} case of two
  spatially separated domains.
\newblock {\em Journal of Mathematical Sciences}, 36(4):557--570, 1987.

\bibitem{QKD3:vazirani2019fully}
U.~Vazirani and T.~Vidick.
\newblock Fully device independent quantum key distribution.
\newblock {\em Communications of the ACM}, 62(4):133--133, 2019.

\bibitem{watts2018algorithms}
A.~B. Watts, A.~W. Harrow, G.~Kanwar, and A.~Natarajan.
\newblock Algorithms, bounds, and strategies for entangled xor games.
\newblock {\em arXiv preprint arXiv:1801.00821}, 2018.

\bibitem{SDCwatts2019exponential}
A.~B. Watts, R.~Kothari, L.~Schaeffer, and A.~Tal.
\newblock Exponential separation between shallow quantum circuits and unbounded
  fan-in shallow classical circuits.
\newblock In {\em Proceedings of the 51st Annual ACM SIGACT Symposium on Theory
  of Computing}, pages 515--526, 2019.

\bibitem{werner2001all}
R.~F. Werner and M.~M. Wolf.
\newblock All-multipartite bell-correlation inequalities for two dichotomic
  observables per site.
\newblock {\em Physical Review A}, 64(3):032112, 2001.

\end{thebibliography}

\end{appendix}

\newpage

\section*{Changes to the Published Version}

 We are grateful to Taro Spring for suggesting a few changes (see below) to the published version of this article. They are implemented in this second arXiv version.

\begin{itemize}[-]
    \item In \Cref{eq:typo1} a $w$ was changed to $w'$. 
    \item In \Cref{eq:typo2}, $\tilde{\pi}_i(x_i^{(\alpha)})$ was changed to $\tilde{\pi}_i(x_j^{(\alpha)})$ for clarity.
    \item In \Cref{eq:typo3,eq:typo4} some incorrectly labeled $\tilde{\rho}$ were changed to $\rho$.
    \item A missing $\proj{\alpha}$ was added to the last line in the proof of \Cref{thm:connectedCompononetsEven}.
    \item In the first paragraph of \Cref{subsec:proof_of_gadgets} a mislabeled $\cG_{12}$ was changed to $\cG_{23}$.
    \item In \Cref{eq:typo5} a mislabeled $\repnew{3,1}{1}{3}$ was changed to $\repnew{3,1}{1}{1}$.
    \item In \Cref{eq:typo6} $\lambda_1 \left( x_{a_i}^{(1)}\right) \lambda_1 \left( x_{a_j}^{(1)} \right)^{-1}$ was changed to $\lambda_1 \left( x_{a_i}^{(1)} x_{a_j}^{(1)} \right)$.
    \item Mislabeled $H$ were changed to $H^E$ in \Cref{susec:finalProof}.
    \item A mislabeled $G_3$ was changed to $G_3^E$ right above \Cref{eq:full_proof_composition}.
    \item A clarifying footnote (footnote 19) was added to the proof sketch (part 2) on page 37.
\end{itemize}

\end{document}